\documentclass[acmsmall]{acmart}
\makeatletter
\def\input@path{{pldi2026/}{./}}
\makeatother

\usepackage{mathpartir}
\usepackage{bm}
\usepackage{enumitem}
\usepackage{tikz}
\usetikzlibrary{arrows.meta}
\usepackage{subcaption}
\usepackage{algorithm}
\usepackage[noend]{algpseudocode}
\usepackage{array}
\usepackage{setspace}
\usepackage{framed}
\usepackage[capitalize]{cleveref}
\usepackage[dvipsnames]{xcolor}
\graphicspath{{pldi2026/}{./}} 

\let\originaltable\table
\let\endoriginaltable\endtable
\renewenvironment{table}[1][tbp]{%
  \originaltable[#1]
  \centering
  \footnotesize
  }%
  {\vspace{-1ex}\endoriginaltable}

\let\originalsubtable\subtable
\let\endoriginalsubtable\endsubtable
\renewenvironment{subtable}[2][t]{%
  \originalsubtable[#1]{#2}
  \centering
  \footnotesize
  \alternaterowcolors}%
  {\endoriginalsubtable}

\let\originalfigure\figure
\let\endoriginalfigure\endfigure
\renewenvironment{figure}[1][tbp]{%
  \originalfigure[#1]
  \centering
  \small}%
  {\vspace{-1em}\endoriginalfigure}

\let\originalsubfigure\subfigure
\let\endoriginalsubfigure\endsubfigure
\renewenvironment{subfigure}[2][t]{%
  \originalsubfigure[#1]{#2}
  \centering
  \small}%
  {\vspace{-1ex}\endoriginalsubfigure}

\definecolor{kwblue}{RGB}{20,20,160}
\definecolor{linenogray}{gray}{0.5}

\newcommand{\mycaption}[1]{\vspace{-0.5em}\caption{#1}}

\newcommand{\RR}{\mathbb{R}}
\newcommand{\NN}{\mathbb{N}}

\newcommand{\Poisson}{\mathsf{Pois}}

\newcommand{\keyword}[1]{\textnormal{\textbf{\textsf{\color{violet} #1}}}}

\newcommand{\Unit}{\Record{}}
\newcommand{\Real}{\mathbb{R}}
\newcommand{\Bool}{\mathbb{B}}
\newcommand{\Nat}{\mathbb{N}}
\newcommand{\ProbTy}[3][]{\mathsf{P}_{#1}\,#2\,#3}
\newcommand{\Dist}[2][]{\mathsf{Dist}_{#1}\;#2}
\newcommand{\Map}[2]{\mathsf{Map}\;#1\;#2}
\newcommand{\Set}[1]{\mathsf{Set}\;#1}
\newcommand{\Collection}[1]{\mathsf{Coll}\;#1}
\newcommand{\Index}{\mathsf{Idx}}
\newcommand{\NameMap}[1]{\mathsf{NameMap}\;#1}
\newcommand{\NameSet}{\mathsf{NameSet}}
\newcommand{\List}[1]{\mathsf{List}\;#1}
\newcommand{\Name}{\mathsf{Name}}
\newcommand{\Record}[1]{\{ #1 \}}
\newcommand{\SubRecord}[1]{\mathsf{Sub} \{ #1 \}}
\newcommand{\Upd}[3][]{\mathsf{Upd}_{#1}(#2; #3)}
\newcommand{\clos}[1][]{\to_{#1}}
\newcommand{\typeOf}[1]{\mathrm{typeOf}(#1)}
\newcommand{\Constants}{\mathrm{Constants}}

\newcommand{\toType}[1]{\mathsf{CtxType}(#1)}
\newcommand{\freeVars}[1]{\mathsf{FV}(#1)}
\newcommand{\projectVar}[2]{\pi_{#1 \mapsto #2}} 
\newcommand{\projectVars}[2]{\pi_{#1 \mapsto #2}} 
\newcommand{\ctxToTpl}[1][]{\mathsf{tpl}_{#1}} 
\newcommand{\tplToCtx}[1][]{\mathsf{ctx}_{#1}} 

\newcommand{\lam}[1]{\keyword{\ensuremath{\bm{\lambda}}} #1 \ldotp}
\newcommand{\return}[1]{\keyword{ret}\;#1}
\newcommand{\sample}[1]{\keyword{sample} \; #1}
\newcommand{\at}{\; \keyword{@} \;}
\newcommand{\using}[1]{\keyword{using}\; (#1)}

\newcommand{\forWithUsing}[7][]{\keyword{for}\; #2 \; \keyword{in}#1 \; #3 \; \keyword{with} \; #4 := #5 \; \using{#6} \; \{ \; #7 \; \}}
\newcommand{\forRangeWithUsing}[6]{\forWithUsing[\keyword{Range}]{#1}{#2}{#3}{#4}{#5}{#6}}

\newcommand{\forWith}[6][]{\keyword{for}\; #2 \; \keyword{in}#1 \; #3 \; \keyword{with} \; #4 := #5 \; \{ \; #6 \; \}}
\newcommand{\forRangeWith}[5]{\forWith[\keyword{Range}]{#1}{#2}{#3}{#4}{#5}}

\newcommand{\forUsing}[5][]{\keyword{for}\; #2 \; \keyword{in}#1 \; #3 \; \using{#4} \; \{ \; #5 \; \}}

\newcommand{\forSetUsing}[4]{\forUsing[\keyword{Set}]{#1}{#2}{#3}{#4}}

\newcommand{\forIn}[4][]{\keyword{for}\; #2 \; \keyword{in}#1 \; #3 \; \{ \; #4 \; \}}
\newcommand{\forInRange}[3]{\forIn[\keyword{Range}]{#1}{#2}{#3}}
\newcommand{\forInSet}[3]{\forIn[\keyword{Set}]{#1}{#2}{#3}}

\newcommand{\ite}[3]{\keyword{if}\; #1 \; \keyword{then} \; #2 \; \keyword{else} \; #3}
\newcommand{\mkUpd}[2]{\keyword{mkUpd}(#1; #2)}
\newcommand{\applyUpd}[2]{#1.\keyword{apply}(#2)}
\newcommand{\mkUpdSem}[1][]{\mathrm{mkUpd}_{#1}}
\newcommand{\stepUpd}[1][]{\mathrm{stepUpd}_{#1}}
\newcommand{\has}[2][]{\keyword{has}_{#1}\;#2}
\newcommand{\trueLit}{\keyword{true}} 
\newcommand{\falseLit}{\keyword{false}} 
\newcommand{\unit}{\recordLit{}}
\newcommand{\recordLit}[1]{\{ #1 \}} 
\newcommand{\const}[1]{\mathsf{#1}}
\newcommand{\var}[1]{\mathsf{#1}}
\newcommand{\mvar}[1]{\mathit{#1}}
\newcommand{\dif}[1]{\mathrm{d}#1}
\newcommand{\lbl}[1]{\textnormal{#1}}
\newcommand{\field}[1]{\lbl{#1:}\,} 

\newcommand{\same}[1][]{\mathsf{same}_{#1}}
\newcommand{\sameAs}[1][]{\mathsf{sameAs}_{#1}}
\newcommand{\apply}[1][]{\mathsf{apply}_{#1}}

\newcommand{\Changed}[1]{\Delta(#1)}

\newcommand{\GenToInc}[2][]{\mathcal{W}_{#1}\!\left\{#2\right\}}
\newcommand{\IncToCore}[2][]{\mathcal{I}_{#1}\!\left\{#2\right\}}
\newcommand{\IncVer}[1]{\const{#1}_{\mathsf{core}}}
\newcommand{\CacheType}{\mathsf{CacheType}}

\newcommand{\istype}[1]{#1 \text{ is type}}
\newcommand{\isground}[1]{#1 \text{ is ground}}
\newcommand{\tyjudg}[3]{#1 \vdash #2 : #3}

\newcommand{\Labels}{\mathcal{L}}

\newcommand{\tr}{\var{tr}}
\newcommand{\dtr}{\var{dtr}}
\newcommand{\dw}{\var{dw}}
\newcommand{\concat}{\mathbin{+\!\!+}}
\newcommand{\issame}[2][]{\mathsf{isSame}_{#1}(#2)}
\newcommand{\restrict}[2]{#1|_{#2}}

\newcommand{\sem}[2][]{\left\llbracket #2 \right\rrbracket_{#1}\!}
\newcommand{\To}{\Rightarrow}
\newcommand{\domain}{\mathsf{dom}}
\newcommand{\default}[1][]{\mathsf{default}_{#1}}
\newcommand{\ret}[1]{\mathsf{ret}~#1}
\newcommand{\supp}[1]{\mathsf{supp}(#1)}
\newcommand{\names}[1]{\mathsf{names}(#1)}
\newcommand{\swap}{\mathsf{swap}}
\newcommand{\id}[1][]{\mathsf{id}_{#1}}

\newcommand{\name}[1]{\mathsf{id}_{#1}}

\definecolor{clusterBlue}{HTML}{0072B2}
\definecolor{clusterOrange}{HTML}{E69F00}
\definecolor{clusterGreen}{HTML}{009E73}
\definecolor{clusterPurple}{HTML}{CC79A7}
\newcommand{\nameBlue}[1]{\textcolor{clusterBlue}{\name{#1}}}
\newcommand{\nameOrange}[1]{\textcolor{clusterOrange}{\name{#1}}}
\newcommand{\nameGreen}[1]{\textcolor{clusterGreen}{\name{#1}}}
\newcommand{\namePurple}[1]{\textcolor{clusterPurple}{\name{#1}}}

\definecolor{corelangcolor}{RGB}{220,240,255}   
\definecolor{inclangcolor}{RGB}{230,255,220}    
\definecolor{genlangcolor}{RGB}{255,240,220}    
\definecolor{mixlangcolor}{RGB}{225,247,238} 
\colorlet{CORELANGCOLOR}{corelangcolor}
\colorlet{INCLANGCOLOR}{inclangcolor}
\colorlet{GENLANGCOLOR}{genlangcolor}
\colorlet{MIXLANGCOLOR}{mixlangcolor}
\newcommand{\flatcolorbox}[2]{%
  \begingroup
  \setlength{\fboxsep}{0pt}
  \colorbox{#1}{\strut #2}
  \endgroup
}
\newcommand{\corehighlight}[1]{\flatcolorbox{corelangcolor!50}{\ensuremath{#1}}}
\newcommand{\inchighlight}[1]{\flatcolorbox{inclangcolor!50}{\ensuremath{#1}}}
\newcommand{\genhighlight}[1]{\flatcolorbox{genlangcolor!50}{\ensuremath{#1}}}
\newcommand{\mixhighlight}[1]{\flatcolorbox{mixlangcolor!50}{\ensuremath{#1}}}
\newcommand{\corelang}{\texorpdfstring{\corehighlight{\lambda_{\mathsf{core}}}}{λ-core}}
\newcommand{\inclang}{\texorpdfstring{\inchighlight{\lambda_{\mathsf{inc}}}}{λ-inc}}
\newcommand{\genlang}{\texorpdfstring{\genhighlight{\lambda_{\mathsf{gen}}}}{λ-gen}}

\definecolor{commentgreen}{HTML}{006064}
\newcommand{\codecomment}[1]{\textit{\textcolor{commentgreen}{#1}}}
\newcommand{\codeshade}[1]{%
  {\setlength{\fboxsep}{2pt}\colorbox{black!5}{#1}}}
\newcommand{\codeshadeBlock}[1]{%
  {\setlength{\fboxsep}{2pt}%
   \colorbox{black!5}{%
     \begin{minipage}{\dimexpr\textwidth-2\fboxsep\relax}#1\end{minipage}}}}

\newcommand{\logrel}[1][]{\mathcal{R}_{#1}}
\newcommand{\changerel}[1][]{\mathcal{C}_{#1}}
\newcommand{\ValidUpd}[3][]{\mathcal{U}_{#1}^{#2,#3}}

\newcommand{\Ione}[1]{\sem{\IncToCore{#1}}_1}
\newcommand{\Itwo}[1]{\sem{\IncToCore{#1}}_2}
\newcommand{\semTuple}[2][]{\sem{#2}_{\mathsf{tpl}_{#1}}}

\setcopyright{cc}
\setcctype{by}
\acmDOI{10.1145/3808316}
\acmYear{2026}
\acmJournal{PACMPL}
\acmVolume{10}
\acmNumber{PLDI}
\acmArticle{238}
\acmMonth{6}
\acmSubmissionID{pldi26main-p491-p}
\received{2025-11-14}
\received[accepted]{2026-04-03}

\makeatletter
\renewcommand\paragraph{\@startsection{paragraph}{4}{\z@}%
    {1ex \@plus1ex \@minus.2ex}%
    {-0.5em}%
    {\normalfont\normalsize\bfseries}}
\makeatother

\makeatletter
\let\ACM@origparagraph\paragraph
\makeatother

\begin{document}

\title{Incremental Computation for Efficient Programmable Inference in Probabilistic Programs}

\author{Fabian Zaiser}
\orcid{0000-0001-5158-2002}
\affiliation{%
  \institution{Massachusetts Institute of Technology}
  \country{USA}
}
\email{fzaiser@mit.edu}

\author{Jack Czenszak}
\orcid{0009-0008-3242-6969}
\affiliation{%
  \institution{Yale University}
  \country{USA}
}
\email{jack.czenszak@yale.edu}

\author{Martin C. Rinard}
\orcid{0000-0001-8095-8523}
\affiliation{%
  \institution{Massachusetts Institute of Technology}
  \country{USA}
}
\email{rinard@csail.mit.edu}

\author{Vikash K. Mansinghka}
\orcid{0000-0003-2507-0833}
\affiliation{%
  \institution{Massachusetts Institute of Technology}
  \country{USA}
}
\email{vkm@mit.edu}

\author{Alexander K. Lew}
\orcid{0000-0002-9262-4392}
\affiliation{%
  \institution{Yale University}
  \country{USA}
}
\email{alexander.lew@yale.edu}

\begin{abstract}
Inference in probabilistic programs generally requires evaluating many possible program executions to find those of high posterior density. To scale inference to large datasets, it is crucial that expensive intermediate results are shared across these many evaluations, rather than recomputed from scratch. This paper presents a new approach to realizing this sharing, based on \textit{incremental computation}, a technique for efficiently recomputing (deterministic) program outputs when program inputs change. First, we show how expressive probabilistic programs can be compiled to deterministic ones that compute their density functions. Then, building on the incremental $\lambda$-calculus, we develop a general technique for compositionally incrementalizing expressive functional programs, and apply it to these densities. The resulting incremental densities can be used to accelerate a broad range of Monte Carlo inference algorithms, including for nonparametric models not well supported by existing systems. Furthermore, our decomposition of incremental density computation into separate density and incrementalization steps allows for modular reasoning about correctness---a key pain point in existing systems, where ad-hoc incrementalization features are a known source of soundness bugs. We develop denotational logical relations arguments for the correctness of each step independently, and implement the approach in a Julia prototype, finding that it leads to asymptotic runtime improvements in the size of the dataset on a range of models and inference algorithms.
\end{abstract}



\maketitle


\begin{figure}[t]
\centering
\footnotesize
\begin{tikzpicture}
  \tikzset{>=Stealth, every path/.style={line width=1pt}}
  \matrix[column sep=18mm] {
    \node[draw, rectangle, inner sep=0.5em, align=center] (gen) {generative program\\in $\genlang$}; &
    \node[draw, rectangle, inner sep=0.5em, align=center] (dens) {density function\\in $\inclang$}; &
    \node[draw, rectangle, inner sep=0.5em, align=center] (upd) {density updater\\in $\corelang$}; &
    \node[draw, rectangle, inner sep=0.5em, align=center] (inf) {inference\\program}; \\
    };
    \draw[->] (gen) -- node[midway, above=1pt, align=center]{density \\ transform} node[midway, below=1pt]{(\cref{sec:density-transformation})} (dens);
    \draw[->] (dens) -- node[midway, above=1pt]{incrementalize} node[midway, below=1pt]{(\cref{sec:incrementalization})} (upd);
    \draw[->, dashed] (inf) -- node[midway, above=1pt, align=center]{repeatedly \\ evaluates} node[midway, below=1pt]{(\cref{fig:marquee:inference-algorithms})} (upd);
\end{tikzpicture}
\mycaption{Overview of our staged approach to incremental density computation for programmable inference.}
\label{fig:schematic-overview}
\end{figure}

\section{Introduction}
\label{sec:introduction}

Probabilistic programming systems promise to accelerate the practice of probabilistic modeling and inference, by providing languages for expressing rich probabilistic models, and automation for aspects of inference.
Unfortunately, for many models, probabilistic programming systems produce slow inference code---in some cases, asymptotically slower than expert implementations of the same algorithms.
One culprit behind this efficiency gap is insufficient \textit{incrementalization}: inference can require thousands or millions of queries to the probability density function of a model, and it is crucial for scalability that expensive intermediate results are shared across these queries.

In this paper, we propose a new approach to closing this efficiency gap for a broad class of models and Monte Carlo inference algorithms, based on \textit{incremental computation}---a technique for efficiently recomputing (deterministic) program outputs when program inputs change. As illustrated in \cref{fig:schematic-overview}, we first show how expressive probabilistic programs can be compiled into deterministic programs that compute their density functions. Then, building on the incremental $\lambda$-calculus~\cite{GiarrussoRS19,CaiGRO14}, we develop a general technique for compositionally incrementalizing expressive higher-order functional programs, and apply it to these density functions.

Our approach addresses two key weaknesses of existing incremental features in probabilistic programming systems
~\cite{Cusumano-Towner19, MansinghkaSP14, TehraniALSNTTML20, MatheosLGRCM21, RitchieSG16, CastellanP19}. First, existing incrementalization is often limited to particular inference algorithms (e.g., single-site Metropolis-Hastings) or model classes (e.g., parametric models with a fixed number of random variables). By contrast, our approach treats densities as ordinary deterministic programs, and incrementalizes them with respect to arbitrary input changes. As a result, we support the incrementalization needs of a broad range of inference algorithms (\cref{fig:marquee:inference-algorithms}), including for nonparametric models such as Dirichlet process mixtures. Second, incrementalization in the presence of probability is notoriously difficult to get right, and existing implementations are brittle and prone to soundness bugs (see, e.g., \cref{fig:marquee:inference-bug-plot}). Our staged approach isolates all probabilistic reasoning in the density compilation step, so that incrementalization need only be proven correct for a deterministic language. We develop new techniques for carrying out this proof: a new approach to typed incrementalization of closures, and a coinductive \textit{updater} abstraction that encapsulates cached intermediate values and supports iterated sequences of updates, enabling a novel denotational correctness argument for caching-based higher-order incremental computation. Our approach is implemented in a Julia prototype that achieves asymptotic speedups over Gen, a state-of-the-art incrementalizing PPL, across a range of models and inference algorithms.

\subsection{Motivating Example}

\begin{figure}
\input{figures/marquee-figure.tex}
\mycaption{Example: modeling and inference for a 2D Gaussian mixture model with incremental densities.}
\label{fig:marquee}
\end{figure}

To illustrate our approach, we consider the problem of inferring a \textit{clustering} for a dataset of points $\{(x_i, y_i)\}_{i=1}^n$ in the plane (\cref{fig:marquee:inferring-clusters}). Our first step is to write a probabilistic program, $\const{gmm}$ (\cref{fig:marquee:generative-model}), encoding a \textit{generative model}: a hypothesized data-generating process by which we imagine our dataset might have been constructed. The program $\const{gmm}$ (for ``Gaussian mixture model'') accepts as input a number $n$ of datapoints to generate, and outputs a synthetic dataset: a list of $n$ pairs. As the simulated dataset in \cref{fig:marquee:sample-traces} illustrates, the key feature of this program is that it generates datasets in which points tend to cluster into a handful of ellipse-shaped components. Indeed, in each simulation, the program explicitly chooses a random number of components to generate, samples random parameters for each component, and randomly chooses a component from which to generate each datapoint. Thus, we can reframe the problem of \textit{inferring a clustering} for our dataset as the problem of finding \textit{executions} of the $\const{gmm}$ program that could plausibly have generated it.

\paragraph{Probabilistic Programs and Traces.} The $\const{gmm}$ program consists of a sequence of \textit{random assignment} statements $\var{x} \sim \var{t}$, with an identifier $\var{x}$ on the left and a probabilistic expression $\var{t}$ on the right (i.e., an expression whose meaning is a probability distribution that the system can sample). A \textit{trace} of a probabilistic program is a record of a particular execution: for each random assignment $\var{x} \sim \var{t}$, it stores an entry $\field{x} v$, recording the choice(s) $v$ made by $\var{t}$ during that execution. For example, consider the trace in \cref{fig:marquee:sample-traces}. The first line of $\const{gmm}$ draws a random number from a geometric distribution and generates a set of that many fresh \textit{names}; the trace records that this set, $\lbl{clusterNames}$, was $\{\nameBlue{73}, \nameOrange{64}, \nameGreen{12}, \namePurple{30}\}$, of size $4$. Next, $\const{gmm}$ loops over this set and for each name inside, generates a \textit{covariance} and \textit{mean} from particular \textit{prior distributions}. Because this loop is over an unordered \textit{set} of names, its trace (stored under the label ``$\lbl{params}$'') is an unordered \textit{map} from names to the values sampled during the corresponding iteration of the loop. Next, the $\const{gmm}$ program uses the $\const{dirichlet}$ primitive to sample random \textit{mixing weights}, resulting in a map from cluster names to real-valued weights, summing to 1. The program then uses the $\const{multinomial}$ primitive to draw a list of $n$ \textit{cluster assignments} according to the mixing weights. Finally, the program loops through this list of assignments to generate the $n$ datapoints. Each datapoint is generated using the $\const{normal}$ primitive, with parameters determined by the component to which we have assigned the point. Because this loop is over an (ordered) list of $n$ elements, its trace is also a list of $n$ elements, recording the values sampled in each iteration of the loop.

\paragraph{Density Functions.}
The \textit{density function} of a probabilistic program maps an execution trace to a non-negative score, which intuitively captures how likely the trace is. In \cref{sec:density-transformation}, we present a program transformation that automatically compiles probabilistic programs to deterministic ones that compute their density functions.
Roughly speaking, the density of a trace can be computed by running the program as usual, except instead of sampling from a primitive distribution, one looks up the recorded value in the trace and multiplies a running accumulator (the \textit{weight}) by the (known) probability mass or density of the recorded value under the primitive distribution.\footnote{This is a standard technique. However, since our language supports fresh name generation and unordered collections, there are several subtleties that complicate the derivation of density functions. We address these challenges in \cref{sec:density-transformation}.}
For the $\const{gmm}$ model, the value of the density function is that it helps distinguish good clusterings from bad ones: if we fix the $\lbl{data}$ entry to match our observed dataset, and vary the number of clusters, their parameters, and their assignments, the density will be low when the trace is incoherent, and high when it provides a plausible explanation of the data.

\paragraph{Incremental Density Computation.}
Inferring a plausible trace consistent with observed data typically requires the density function to be evaluated at many related traces. The goal of this paper is to accelerate inference by performing these density queries \textit{incrementally}. In~\cref{sec:incrementalization}, we present a second program transformation, which compiles the density program from \cref{sec:density-transformation} into an \textit{incremental version}. The incremental version can be run on an \textit{initial} trace $\tr$, to yield a density $\var{w}$ and an \textit{updater} $\var{u}$. The updater caches certain intermediate values from the computation, and thus helps to calculate the density at a slightly modified trace $\tr'$ without full recomputation. In many cases, incrementalization delivers asymptotic speedups: in the $\const{gmm}$ model, for example, we reduce the cost of some density queries from $O(n)$ to $O(1)$, by skipping unnecessary loop iterations.

Concretely, an updater $\var{u}$ is applied to a \textit{change description} $\dtr$, which precisely encodes the delta from $\tr$ to $\tr'$. The type of the change description is based on the type of $\tr$ (\cref{fig:marquee:trace-change}). For example,
\begin{align*}
  \dtr &= \recordLit{ \field{assignments} \recordLit{ \field{change} [(1, \recordLit{ \field{new} \nameOrange{64}, \field{same} \falseLit})] } }
\end{align*}
is one possible change to a trace of $\const{gmm}$, encoding a reassignment of datapoint $1$ to cluster $\nameOrange{64}$. Because the entry at the label $\lbl{assignments}$ in $\tr$ is a list, we supply a list-change, which may specify insertions, deletions, and element modifications; in this case, we only modify element 1 of the list.
Running $\var{u}.\keyword{apply}(\dtr)$ efficiently recomputes the density given this change, which only requires re-evaluating a single loop iteration (for the reassigned datapoint).
The updater also returns a new updater $\var{u}'$, which can be used to efficiently compute densities for traces similar to $\tr'$.

\paragraph{Accelerating Inference.}
We can exploit incremental density computation to speed up our inference algorithm of choice. \Cref{fig:marquee:inference-algorithms} gives pseudocode for three families of popular Monte Carlo algorithms (Metropolis-Hastings, Gibbs sampling, and sequential Monte Carlo), illustrating the role that incremental densities can play in each. We apply Metropolis-Hastings to our $\const{gmm}$ model. The algorithm first constructs an (arbitrary) initial trace whose $\lbl{data}$ field is fixed to our observed dataset. It then repeatedly proposes small changes to the trace, accepting or rejecting each change according to the resulting change in density. As \cref{fig:marquee:inference-time-plot} illustrates, incremental density computation leads to asymptotically faster evaluation of each proposal, and thus dramatically reduces the wall-clock time for inference to converge. Although existing systems, such as Gen~\citep{Cusumano-Towner19}, support incrementality to some extent, our approach exploits some opportunities for asymptotic speedups that existing systems fail to capture. Furthermore, due to the complexity of reasoning about incremental computation and probability, we have found some systems to have incrementalization bugs that can lead inference to silently fail (see \cref{fig:marquee:inference-bug-plot}).\footnote{
See, e.g., \url{https://github.com/probcomp/Gen.jl/issues/512} or \url{https://github.com/probcomp/Gen.jl/issues/193}.
} A key aim of this work is to provide a framework, based on modular program transformations and logical relations, for reasoning clearly about the correctness of incremental densities.

\subsection{Our Approach}

Our approach overcomes several key challenges to deliver sound, efficient incremental densities:

\paragraph{(C1) Combining Probabilities and Incrementality.} Existing support for incrementalization in probabilistic programming systems is generally based on a direct consideration of the probabilistic language's semantics.
Incremental computation in the presence of probabilities is nontrivial: even the existence of program variables can depend on random choices, which complicates the propagation of changes through the program.
As a result, many existing systems rely on complex dependency tracking based on data structures inspired by graphical models~\citep{MatheosLGRCM21,MansinghkaSHRCR18,TehraniALSNTTML20,Cusumano-Towner19}, which can be difficult to reason about (and empirically, is a source of several soundness bugs). Our approach isolates all probabilistic reasoning from incrementalization, by first compiling a probabilistic program to a deterministic density function. Incrementalization is then implemented (and proven correct) for only the deterministic language.

\paragraph{(C2) Incrementalizing Open-Universe Models.}
Probabilistic programs often generate latent collections of \emph{objects} whose number is not known in advance (e.g., the clusters in the $\const{gmm}$ program). To support efficient recomputation when objects are added and deleted (or, e.g., split and merged), traces and caches must use data structures that make such operations efficient. For example, it is crucial that $\const{gmm}$ uses dictionaries with IDs as keys (rather than, e.g., vectors with consecutive integer indices) to store cluster information; otherwise, deleting a latent cluster would require expensive re-indexing to maintain the invariant of consecutive component IDs, inducing $O(n)$ changes to cluster assignments. To overcome this challenge, we designed our probabilistic language to support features such as fresh name generation and unordered collections. Using these features, it is possible to express nonparametric, open-universe models in a way that makes them amenable to efficient incrementalization.

\paragraph{(C3) Densities for Open-Universe Models.}
Unfortunately, the unique names and unordered collections we use to efficiently encode open-universe models violate two key assumptions behind the density calculations in most PPLs: (1) that all choices of a given type are \textit{either} discrete \textit{or} continuous; and (2) that any given execution trace of a program could have been generated in exactly one way by that program. The first assumption is violated because \textit{names} are sometimes drawn from continuous distributions (to ensure that they will, with probability 1, be fresh~\citep{SabokSSW21}), and sometimes from discrete distributions (when we are choosing randomly to refer to one of a finite set of previously generated IDs). The second assumption is violated because, when an unordered set is stored in a trace, we erase information about the order in which the elements were generated. In \cref{sec:density-transformation}, we develop a family of reference measures that allow us to formulate a working notion of density in this setting.
We show how to compute correct densities by including (1) factorial corrections counting the number of ways in which an unordered collection could have been generated, and (2) additional ``freshness checks'' for names recorded in the trace.

\paragraph{(C4) Representing Changes.}
The probabilistic programming system must be able to handle a broad range of user-defined changes to capture trace updates that arise in real-world inference algorithms.
We address this with a type-directed approach: the kinds of changes to a trace depend on the kinds of data in the trace.
First, we use trace typing~\cite{LewCSCM20} to synthesize a rich static type for the traces of a probabilistic program, based on its control flow.
Second, inspired by the incremental $\lambda$-calculus~\citep{CaiGRO14,GiarrussoRS19,Morihata20}, we define change representations by induction on the trace type. This approach enables concise representations of several common changes that are cumbersome to express (and lead to inefficient incrementalization) in state-of-the-art PPLs.

\paragraph{(C5) Processing Changes Efficiently.}
A core technical contribution of this paper is a new approach to typed higher-order incremental computation with caches.
Given a trace change, our goal is to compute a density update with the same asymptotic runtime complexity as an expert manual implementation.
To our knowledge, no existing PPL achieves the best asymptotic performance for updates to nonparametric models such as Dirichlet process mixtures. The key source of high asymptotic complexity is the unnecessary re-execution of loops.
Thus our algorithm is designed to precisely propagate change information about the input trace ($\dtr$) through the density program, so that before entering a loop, we have enough information to compute precisely which loop iterations need to be re-executed. This change propagation is realized by per-primitive change propagators, and rules inspired by the caching incremental $\lambda$-calculus~\citep{GiarrussoRS19} for composing them. We develop a new approach to implementing this design in a typed, higher-order setting, which allows us to view loops themselves as combinator-like, higher-order primitives with their own built-in logic for caching, change propagation, and incrementalization.

\paragraph{(C6) Reasoning about Correctness.} Our two-stage approach allows us to reason separately about the correctness of density computation and incrementalization; we prove each transformation correct using independent logical relations arguments over the denotational semantics of our languages. Interestingly, previous work on cache-based incrementalization \cite{GiarrussoRS19} was formalized only in an untyped setting, because the types of caches of intermediate values quickly become unwieldy. We address this challenge by encapsulating caches in an opaque updater type $\Upd{\tau}{\tau'}$, which hides the cache type after construction.
The semantics treats updaters as coinductive data types (infinite trees of all possible future update sequences) and can ignore caching concerns. This enables a modular logical relations proof over a standard, set-theoretic denotational semantics.

\paragraph{Contributions.}
In summary, we make the following contributions.
\begin{itemize}[leftmargin=*]
    \item \textbf{A new approach to typed higher-order incremental computation (\cref{sec:incrementalization}):}
    We present a new program transformation for turning higher-order deterministic programs into \emph{updaters} that efficiently compute changes to their outputs given changes to their inputs.
    Our updaters support caches of intermediate results that are encapsulated via coinductive \emph{updater types}, giving the first typed treatment of caching incremental computation while staying within the simply-typed setting.
    Each update also returns a new updater to process subsequent changes, which is essential for iterative applications such as Markov chain Monte Carlo inference.
    \item \textbf{Incremental densities for open-universe models (\cref{sec:calculi,sec:density-transformation}):} We present a program transformation that compiles probabilistic programs into deterministic, incrementalizable density functions. These densities are defined on typed execution traces, based on a novel extension of \textit{trace typing}~\citep{LewCSCM20,BeckerHMWCSRSLRM26,WangHR21,LiASZ23} to support unordered collections and named objects. In combination with our incrementalization transformation, this enables efficient trace updates when inferring unknown collections of entities in open-universe and nonparametric models~\citep{MatheosLGRCM21,MilchR06}.
    \item \textbf{Denotational correctness arguments (\cref{sec:density-transformation:correctness,sec:incrementalization:correctness}):}
    We present a denotational account of the correctness of both stages of our pipeline.
    For the density transformation, we develop novel semantic techniques for reasoning about the densities of probabilistic programs that may generate unordered collections of names.
    For the incrementalization transformation, we present the first denotational argument for the correctness of caching incremental computation, which relies on a new coinduction-based reasoning principle for the correctness of updaters.
    \item \textbf{Empirical validation (\cref{sec:case-studies}):}
    We implement our approach for an embedded probabilistic DSL in Julia and compare it with Gen \cite{Cusumano-Towner19}, a state-of-the-art programmable inference system implemented in the same host language.
    We show that our approach offers significant speedups compared to naive density evaluation and its asymptotic runtime is the same or better than Gen.
\end{itemize}

\section{Probabilistic and Deterministic Languages}
\label{sec:calculi}

This section introduces the syntax and semantics of three related programming languages. As illustrated in \cref{fig:schematic-overview}, the transformations we develop in \cref{sec:density-transformation,sec:incrementalization} translate between these languages: user probabilistic programs in $\genlang$ are translated to deterministic density functions in $\inclang$, which are further translated into incremental density functions in $\corelang$. The grammars and typing rules for all three languages are presented in \cref{fig:syntax,fig:type-system}.

\subsection{Probabilistic Language}

Our probabilistic language, $\genlang$, is the user-facing language for encoding probabilistic models.

\paragraph{Ground Types.} $\genlang$ is a variant of the simply-typed $\lambda$-calculus with the following ground types:
\[ \sigma ::= \Bool \mid \Real \mid \Nat \mid \Name \mid \sigma_1 \times \sigma_2 \mid \List{\sigma} \mid \NameMap{\sigma} \mid \Record{\ell_1: \sigma_1, \dots, \ell_n: \sigma_n} \quad \quad (\ell_i \in \mathcal{L}) \]
Each type $\sigma$ is taken to denote a space of values $\sem{\sigma}$, and a term $\Gamma \vdash t : \tau$ denotes a map $\sem{t} : \sem{\Gamma} \to \sem{\tau}$.\footnote{As is standard in the semantics of higher-order probabilistic programming languages, these spaces are \textit{quasi-Borel spaces}~\citep{HeunenKSY17}, a drop-in replacement for measurable spaces with better support for function types.} The types $\Bool$, $\Nat$, and $\Real$ denote the standard spaces of Booleans, natural numbers, and reals respectively. To model \textit{names}, we take $\sem{\Name} = [0,1]$, exploiting an insight of~\citet{SabokSSW21} that if names are interpreted as real numbers, fresh name generation can be modeled as continuous sampling (as, with probability 1, two independent samples will not be equal). Record types $\Record{\ell_1: \tau_1, \dots, \ell_n: \tau_n}$, with field labels drawn from a countably infinite set $\mathcal{L}$, denote indexed products, with projections $\pi_{\ell_1}, \dots, \pi_{\ell_n}$. The empty record $\Unit$ serves as a unit type for $\genlang$. $\List{\sigma}$ denotes $\sqcup_{i \in \NN} \sem{\sigma}^i$. The type $\NameMap{\sigma}$ denotes finite dictionaries with keys in $\sem{\Name}$ and values in $\sem{\sigma}$, which we represent concretely as \textit{sorted} lists of $\sem{\Name} \times \sem{\sigma}$ pairs without duplicate keys (sorted by the usual order on $\sem{\Name}=[0,1]$). We write $\NameSet$ as sugar for $\NameMap{\Unit}$. We use the syntactic sugar $\ite{s}{t}{t'}$ for a primitive operation $\const{ite}_\sigma(s, t, t')$ where $t$ and $t'$ have ground type $\sigma$, and $s[t]$ for a primitive $\const{get}_\sigma(t, s)$ that indexes lists and maps.

\paragraph{Function Types.} Function types in $\genlang$ are written $\sigma \to_\kappa \tau$. The annotation $\kappa$ on the arrow specifies the type of data in a closure's captured environment, but does not affect the semantics: $\sem{\sigma \to_\kappa \tau} = \sem{\sigma} \To \sem{\tau}$ (where $\To$ constructs the semantic function space). For example, we have $$y : \Real, z : \Nat \vdash \lambda x. ~ x < y : \Real \to_{\Real} \Bool$$ because the function $\lambda x. ~ x < y$ closes over the variable $y$ of type $\Real$.\footnote{Technically, the typing rules compute the equivalent type $\Real \to_{\Unit{} \times \Real} \Bool$.} (It is admittedly unusual to track this information in the type of a function. More details will be provided in \cref{sec:incrementalization}, but the intuition is that our incrementalization transformation will represent \textit{changes} to functions of type $\sigma \to_\kappa \tau$ as changes to the data they close over, and so the type $\kappa$ of this data must be explicit.)

\begin{figure}
  \begin{align*}
    \text{Types } \sigma, \tau &::= \Bool \mid \Nat \mid \Real \mid \Name \mid \sigma \times \tau \mid \List{\sigma} \mid \NameMap{\sigma} \mid \Record{\ell_1: \tau_1, \dots, \ell_n: \tau_n} \\
    &\qquad \mid \underbrace{\genhighlight{\Dist[\kappa]{\tau}} \mid \genhighlight{\ProbTy[\kappa]{\sigma}{\tau} \mid \tau \clos[\sigma] \tau'}}_{\genlang} \mid \underbrace{\inchighlight{\tau \clos[\sigma] \tau'}}_{\inclang} \mid \underbrace{\corehighlight{\tau \to \tau' \mid \SubRecord{\ell_1: \tau_1, \dots, \ell_n: \tau_n} \mid \Upd{\tau}{\tau'}}}_{\corelang} \\
    \text{Terms: } s, t &::= x \mid c \mid x := s; t \mid \lam{x} t \mid s \; t \mid (s, t) \mid t.1 \mid t.2 \\
    \text{in \genlang:} &\qquad \mid \genhighlight{\return{t} \mid x \gets s; t \mid \sample{t} \mid t \at \ell} \\
    &\qquad \mid \genhighlight{\forSetUsing{x}{\mvar{xs}}{v_1 := y_1[u_1], \dots, v_n := y_n[u_n]}{t}} \\
    &\qquad \mid \genhighlight{\forRangeWithUsing{x}{w}{z}{s}{v_1 := y_1[u_1], \dots, v_n := y_n[u_n]}{t}} \\
    \text{in \inclang:} &\qquad \mid \inchighlight{t.\ell} \mid \inchighlight{\restrict{t}{\{\ell_1, \dots, \ell_n\}}} \mid \inchighlight{\forUsing{x}{\mvar{xs}}{v_1 := y_1[u_1], \dots, v_n := y_n[u_n]}{t}} \\
     &\qquad \mid \inchighlight{\forWithUsing{x}{\mvar{xs}}{z}{s}{v_1 := y_1[u_1], \dots, v_n := y_n[u_n]}{t}} \\
    \text{in \corelang:} &\qquad \mid \corehighlight{t.\ell} \mid \corehighlight{\restrict{t}{\{\ell_1, \dots, \ell_n\}}} \mid \corehighlight{\has[\ell]{t}} \mid \corehighlight{\mkUpd{s}{t}} \mid \corehighlight{\applyUpd{s}{t}}
  \end{align*}
  \vspace{-1.5em}
  \mycaption{Syntax of our languages $\genlang, \inclang, \corelang$. Labels $\ell$ are from a countable set $\Labels$.}
  \label{fig:syntax}
\end{figure}

\paragraph{Probabilistic Programs.} Probabilistic programs in $\genlang$ have types of the form $\ProbTy[\kappa]{\sigma}{\tau}$: $\tau$ is the return type of the probabilistic program, and $\sigma$ is its \textit{trace type}, i.e., the data type used to represent reified traces of its execution~\citep{LewCSCM20}.
As with function types, $\kappa$ is an environment annotation, reflecting the types of any free variables captured by the probabilistic computation. (We omit it when none are captured.) For example, the generative model of \cref{fig:marquee:generative-model} has type $\ProbTy[]{\sigma}{(\List{\Real^2})}$, where
\begin{align*}
    \sigma = \{&\lbl{clusterNames} : \NameSet, \lbl{params} : \NameMap{\{\lbl{mean} : \Real^2, \lbl{cov} : \Real^4\}}, \\
    &\lbl{weights} : \NameMap{\Real}, \lbl{assignments} : \List{\Name}, \lbl{data}: \List{\Real^2}\}
\end{align*}

Semantically, $\sem{\ProbTy[\kappa]{\sigma}{\tau}} = \mathsf{Prob}~\sem{\sigma} \times (\sem{\sigma} \To \sem{\tau})$. That is, a probabilistic program denotes a pair $(\mu, f)$, where $\mu \in \mathsf{Prob}~\sem{\sigma}$ is a distribution on traces, and $f : \sem{\sigma} \to \sem{\tau}$ is a deterministic function mapping traces to corresponding return values. Formally, $\mathsf{Prob}~X$ is a subset of the function space $(X \To [0,\infty]) \To [0,\infty]$: we identify a probability distribution $\mu$ with its \textit{expectation operator}, which sends functions $f : X \to [0, \infty]$ to their expected values $\mathbb{E}_\mu[f]$. For example, the uniform distribution on $[0, 1]$ is $\Lambda_{[0,1]} = f \mapsto \int_0^1 f(x) \dif{x}$, and the \textit{Dirac delta} distribution at a point $x$ is $\delta_x= f \mapsto f(x)$.
For $k : X \to \mathsf{Prob}~Y$, we write $\mu; k \in \mathsf{Prob}~Y$ for the composition $f \mapsto \mu(x \mapsto k(x)(f))$. We write $x \leftsquigarrow \mu; \nu$, where $x$ occurs free in $\nu$, to mean $\mu; (x \mapsto \nu)$.

Probabilistic programs are built by composing \textit{primitive distributions} of type $\Dist[\kappa]{\sigma}$, where $\sem{\Dist[\kappa]{\sigma}} = \mathsf{Prob}~\sem{\sigma}$. The key constructs are:

\begin{enumerate}[leftmargin=*]
\item \textit{Deterministic computation.} A probabilistic program can return a deterministically computed value using the construct $\return{t}$. The execution traces of deterministic programs are empty, because there are no random choices to record. For example, we have the judgments
\begin{align*}
\vdash \return{(5+3, \keyword{true})} : \ProbTy[]{\Unit}{(\Nat \times \Bool)} && x : \Real \vdash \return{x < 3} : \ProbTy[\Real]{\Unit}{\Bool}
\end{align*}
where on the right, the environment annotation $\kappa=\Real$ indicates that the program closes over a variable of type $\Real$. For $\Gamma \vdash t : \tau$ and $\gamma \in \sem{\Gamma}$, we have $\sem{\return{t}}(\gamma) = (\delta_{\Unit}, \_ \mapsto \sem{t}(\gamma))$.

\item \textit{Sampling primitive distributions.} The construct $\sample{t}$ generates and returns a sample from the primitive distribution specified by $t : \Dist[]{\sigma}$. For example, $\sample{\const{normal}(0,1)} : \ProbTy[]{\Real}{\Real}$ samples a standard Gaussian and returns it. The trace simply records the sampled value: $$\sem{\sample{t}}(\gamma)=(\sem{t}(\gamma), \mathrm{id})$$

In addition to the standard discrete and continuous primitive distributions, $\genlang$ has a primitive $\const{fresh} : \Dist{\Nat} \to \Dist{\NameSet}$. This primitive generates a random number $n$ from its argument distribution,\footnote{When we wish to generate a fixed number of fresh names, we write $\const{fresh}(n)$ as sugar for $\const{fresh}(\const{dirac}(n))$.} then returns a set of $n$ freshly generated names (i.e., $n$ independent samples from the uniform distribution on $[0, 1]$).

\item \textit{Labeling random choices.} The construct $t \at \ell$ \textit{labels} the random choice(s) made by $t$ with $\ell \in \mathcal{L}$; the trace type of $t \at \ell$ is the singleton record $\{ \ell : \sigma \}$, where $\sigma$ is the trace type of $t$. For example, we have $\sample{\const{normal}(0,1)} \at \lbl{val} : \ProbTy[]{\{\lbl{val} : \Real\}}{\Real}$. Semantically, \begin{align*}\sem{t \at \ell}(\gamma) = (x \leftsquigarrow \mu; \delta_{\{\ell \mapsto x\}}, f \circ \pi_\ell) && \text{where }(\mu, f) := \sem{t}(\gamma)\end{align*}

\item \textit{Sequencing.} The construct $x \gets t; s$ runs $t$ and binds its return value to $x$, then runs $s$. The terms $t$ and $s$ must both be probabilistic programs with disjoint record trace types; the trace type of $x \gets t; s$ is the merge of the two record types. We write $x \sim t$ as syntactic sugar for $x \gets t \at x$ (or for $x \gets \sample{t} \at x$ if $t : \Dist[]{\sigma}$). For example, we have the judgment
$$x : \Real \vdash y \sim \const{normal}(x, 1); z \sim {\const{beta}(1, 2)}; \return{z < y} : \ProbTy[\Real]{\{y : \Real, z : \Real\}}{\Bool}$$ where the trace type $\{y : \Real, z : \Real\}$ is the concatenation of $\{y : \Real\}$, $\{z : \Real\}$, and $\Unit$.
Semantically,
\begin{align*}\sem{x \gets t; s}(\gamma) = (\rho_1 \leftsquigarrow \mu; \rho_2 \leftsquigarrow k(f(\rho_1)); \delta_{\rho_1 +\!\!+ \rho_2}, \rho \mapsto g(f(\rho|_{L_1}))(\rho|_{L_2}))\end{align*}
where $(\mu, f)=\sem{t}(\gamma)$, $\langle k, g\rangle = \sem{s}(\gamma[x \mapsto -])$, and $L_1$ and $L_2$ are the label sets of the trace types of $t$ and $s$.
Our probabilistic program types form a \emph{graded monad} with grade $\sigma$.\footnote{However, trace types $\sigma$ form only a \textit{partial} monoid, because only disjoint record types $\sigma_1$ and $\sigma_2$ can be merged.}

\item \textit{Loops.} The loop $\forInRange{x}{d}{t}$, where $d : \Dist[]{\Nat}$, draws $n \sim d$, then iterates from $x := 1$ to $n$.\footnote{We write $\forInRange{x}{n}{t}$, where $n : \Nat$, as shorthand for $\forInRange{x}{\const{dirac}(n)}{t}$.} Its traces are lists, where each element is a trace of $t$, and $n$ is implicitly recorded as the length of the list. Its return value is also a list, of the values returned by $t$ at each iteration. Optionally, an accumulator may be initialized and updated each iteration:
{\small $$ \forRangeWith{\var{i}}{5}{\var{z}}{0}{\var{z}' := \var{z} + 1; \var{y} \sim \const{normal}(\var{z}',0.1); \return{(\var{y}, \var{z}')}} : \ProbTy[]{(\List{\{\lbl{y} : \Real\}})}{(\List{\Real})}$$}%
The construct $\forInSet{x}{d}{t}$ loops over an unordered set of names (no accumulator). This loop draws $\mvar{names} \sim d$, and runs $t$ (of type $\ProbTy[]{\sigma}{\tau}$) once for each name, collecting the results into a $\NameMap{\tau}$. Its trace type is $\NameMap{\sigma}$, as its traces store a trace of $t$ for each iteration.
The clause $\using{v_1:=y_1[u_1]; \dots; v_n:=y_n[u_n]}$ may also be added before a loop's body. This is equivalent to adding $v_1:=y_1[u_1];\dots;v_n:=y_n[u_n];$ to the loop body, but tells the system to track the indices of the collections $y_i$ accessed at each iteration for finer-grained incrementality.
\end{enumerate}

\begin{figure}\footnotesize
  \centering
  \begin{mathpar}
    \inferrule{ }{\toType{\bullet} = \Unit}

    \inferrule
      {\toType{\Gamma} = \tau}
      {\toType{\Gamma, x : \tau'} = \tau \times \tau'}

    \inferrule*[right={\scriptsize(\genlang, \inclang)}]
      {\tyjudg{\Gamma, x: \tau}{t}{\tau'}}
      {\tyjudg{\Gamma}{\lam{x} t}{\tau \clos[\toType{\Gamma|_{\freeVars{\lam{x} t}}}] \tau'}}
    
    \inferrule*[right={\scriptsize(\genlang, \inclang)}]
      {\tyjudg{\Gamma}{s}{\tau \clos[\sigma] \tau'} \\ \tyjudg{\Gamma}{t}{\tau}}
      {\tyjudg{\Gamma}{s \; t}{\tau'}}

    \inferrule*[right={\scriptsize(\genlang)}]
      {\tyjudg{\Gamma}{t}{\tau}}
      {\tyjudg{\Gamma}{\return{t}}{\ProbTy[\toType{\Gamma|_{\freeVars{t}}}]{\Unit}{\tau}}}

    \inferrule*[right={\scriptsize(\genlang)}]
      {\tyjudg{\Gamma}{t}{\ProbTy[\kappa]{\sigma}{\tau}}}
      {\tyjudg{\Gamma}{t \at \ell}{\ProbTy[\kappa]{\Record{\ell: \sigma}}{\tau}}}

    \inferrule*[right={\scriptsize(\genlang)}]
      {\tyjudg{\Gamma}{t}{\Dist[\kappa]{\sigma}}}
      {\tyjudg{\Gamma}{\sample{t}}{\ProbTy[\kappa]{\sigma}{\sigma}}}

      \inferrule*[right={\scriptsize(\genlang)}]
      {\tyjudg{\Gamma}{s}{\ProbTy[\kappa]{\Record{\ell_1: \sigma_1, \dots, \ell_m: \sigma_m}}{\tau}} \\\\
      \tyjudg{\Gamma, x : \tau}{t}{\ProbTy[\kappa']{\Record{\ell_1': \sigma_1', \dots, \ell_n': \sigma_n'}}{\tau'}} \\
      \{ \ell_1, \dots, \ell_m \} \cap \{ \ell_1', \dots, \ell_n' \} = \emptyset}
      {\tyjudg{\Gamma}{x \gets s; t}{\ProbTy[\toType{\Gamma|_{\freeVars{x \gets s; t}}}]{\Record{\ell_1: \sigma_1, \dots, \ell_m: \sigma_m, \ell_1': \sigma_1', \dots, \ell_n': \sigma_n'}}{\tau'}}}

      \resizebox{\textwidth}{!}{$
      \inferrule*[right={\scriptsize(\genlang)}]
      {\tyjudg{\Gamma}{w}{\Dist[\kappa]{\Nat}} \\
      \tyjudg{\Gamma}{s}{\rho} \\
      \tyjudg{\Gamma, x: \Nat, z: \rho}{y_i[u_i]}{\theta_i} \\
      \tyjudg{\Gamma, x: \Nat, z: \rho, v_1: \theta_1, \dots, v_n: \theta_n}{t}{\ProbTy[\kappa']{\sigma}{(\tau \times \rho)}}}
      {\tyjudg{\Gamma}{\forRangeWithUsing{x}{w}{z}{s}{v_1 := y_1[u_1], \dots, v_n:=y_n[u_n]}{t}}{\ProbTy[\toType{\Gamma|_{\freeVars{...}}}]{(\List \sigma)}{(\List{\tau})}}}
      $}

      \resizebox{\textwidth}{!}{$
      \inferrule*[right={\scriptsize(\genlang)}]
      {\tyjudg{\Gamma}{\mvar{xs}}{\Dist[\kappa]{\NameSet}} \\
      \tyjudg{\Gamma, x: \Name}{y_i[u_i]}{\theta_i} \\
      \tyjudg{\Gamma, x: \Name, v_1: \theta_1, \dots, v_n: \theta_n}{t}{\ProbTy[\kappa']{\sigma}{\tau'}}}
      {\tyjudg{\Gamma}{\forSetUsing{x}{\mvar{xs}}{v_1 := y_1[u_1], \dots, v_n := y_n[u_n]}{t}}{\ProbTy[\toType{\Gamma|_{\freeVars{...}}}]{(\NameMap{\sigma})}{(\NameMap{\tau'})}}}
      $}

    \inferrule*[right={\scriptsize(\inclang)}]
      {\tyjudg{\Gamma}{\mvar{xs}}{\List \tau} \\
      \tyjudg{\Gamma}{s}{\rho} \\
      \tyjudg{\Gamma, x: \tau, z: \rho}{y_i[u_i]}{\theta_i} \\
      \tyjudg{\Gamma, x: \tau, z: \rho, v_1: \theta_1, \dots, v_n: \theta_n}{t}{\tau' \times \rho}}
      {\tyjudg{\Gamma}{\forWithUsing{x}{\mvar{xs}}{z}{s}{v_1 := y_1[u_1], \dots, v_n := y_n[u_n]}{t}}{\List{\tau'}}}

    \inferrule*[Right={\scriptsize(\inclang)}]
      {\tyjudg{\Gamma}{\mvar{xs}}{\NameSet} \\
      \tyjudg{\Gamma, x: \Name}{y_i[u_i]}{\theta_i} \\
      \tyjudg{\Gamma, x: \Name, v_1: \theta_1, \dots, v_n: \theta_n}{t}{\tau'}}
      {\tyjudg{\Gamma}{\forUsing{x}{\mvar{xs}}{v_1 := y_1[u_1], \dots, v_n := y_n[u_n]}{t}}{\NameMap{\tau'}}}

    \inferrule*[right={(\scriptsize\inclang, \corelang)}]
      {\tyjudg{\Gamma}{t}{\Record{\ell_1: \tau_1, \dots, \ell_n: \tau_n}}}
      {\tyjudg{\Gamma}{t.\ell_i}{\tau_i}}

    \inferrule*[right={(\scriptsize\inclang, \corelang)}]
      {\tyjudg{\Gamma}{t}{\Record{\ell_1: \tau_1, \dots, \ell_n: \tau_n}} \\
      \{i_1 < \dots < i_k\} \subseteq \{1, \dots, n\}}
      {\tyjudg{\Gamma}{\restrict{t}{\{\ell_{i_1}, \dots, \ell_{i_k}\}}}{\Record{\ell_{i_1}: \tau_{i_1}, \dots, \ell_{i_k}: \tau_{i_k}}}}

    \inferrule*[Right={\scriptsize(\corelang)}]
      {\tyjudg{\Gamma}{s}{\sigma} \\ \tyjudg{\Gamma}{t}{\tau \times \sigma \to \tau' \times \sigma}}
      {\tyjudg{\Gamma}{\mkUpd s t}{\Upd{\tau}{\tau'}}}\;

    \inferrule*[Right={\scriptsize(\corelang)}]
      {\tyjudg{\Gamma}{s}{\Upd{\tau}{\tau'}} \\ \tyjudg{\Gamma}{t}{\tau}}
      {\tyjudg{\Gamma}{\applyUpd{s}{t}}{\tau' \times \Upd{\tau}{\tau'}}}\;

    \inferrule*[right={\scriptsize(\corelang)}]
      {\tyjudg{\Gamma}{t}{\SubRecord{\ell_1: \tau_1, \dots, \ell_n: \tau_n}}}
      {\tyjudg{\Gamma}{\has[\ell_i]{t}}{\Bool}}
  \end{mathpar}\vspace{-1em}
  \mycaption{Selected Typing Rules. We write $\Gamma|_{\freeVars{t}}$ for the restriction of $\Gamma$ to the free variables of $t$.}
  \label{fig:type-system}
\end{figure}

\subsection{Deterministic Languages}
\label{sec:calculi:semantics}

In addition to $\genlang$, we have two deterministic languages: $\inclang$ (the \textit{source language} of our incrementalization transformation) and $\corelang$ (the \textit{target language} of our incrementalization transformation). Both $\inclang$ and $\corelang$ have set-theoretic denotational semantics (\cref{fig:inc-semantics,fig:core-semantics}): each type $\tau$ denotes a set $\sem{\tau}$, and each term $\Gamma \vdash t : \tau$ denotes a function $\sem{t} : \sem{\Gamma} \to \sem{\tau}$. The deterministic languages inherit the ground types of $\genlang$.

\paragraph{The $\inclang$ Calculus.} As in $\genlang$, $\inclang$ function types $\sigma \to_\kappa \tau$ are annotated with captured environment types $\kappa$. But in $\inclang$, the \textit{semantics} of a function tracks how it depends on its closed-over data:
$$\sem{{\sigma \to_\kappa \tau}} = \sem{\kappa} \times (\sem{\kappa} \times \sem{\sigma} \To \sem{\tau})$$
The semantics of function abstraction and application explicitly manipulate these closures.

The $\inclang$ language also features a \textit{deterministic} looping construct, $\forIn{x}{\mvar{xs}}{t}$, where $\mvar{xs}$ may either be of type $\List{\sigma}$ (in which case $x : \sigma \vdash t : \tau$ and the loop returns a $\List{\tau}$) or of type $\NameSet$ (in which case $x : \Name \vdash t : \tau$ and the loop returns a $\NameMap{\tau}$).
As in $\genlang$, the loop supports a $\using{v_1 := y_1[u_1], \dots, v_n := y_n[u_n]}$ clause for fine-grained tracking of the indices at which each iteration accesses the collections $y_i$.
Iterating over lists (but not unordered collections) also supports a variant with an accumulator, introduced by the clause $\keyword{with}~z := s$ before the loop body.

\begin{figure}\footnotesize
\begin{gather*}
  \sem{\Bool} = \{ 0, 1 \} \quad
  \sem{\Name} = [0, 1] \quad
  \textstyle\sem{\List{\sigma}} = \sem{\sigma}^* := \bigcup_{n \in \NN} \sem{\sigma}^n \quad
  \sem{\NameMap{\tau}} = \{ f : A \to \sem{\tau} \mid A \subset_{\mathsf{fin}} \sem{\Name}\} \\
    \sem{\tau \clos[\kappa] \tau'} = \sem{\kappa} \times (\sem{\kappa} \times \sem{\tau} \To \sem{\tau'}) \qquad
  \textstyle\sem{\Record{\ell_1: \tau_1, \dots, \ell_n: \tau_n}} = \{r : \{ \ell_1, \dots, \ell_n \} \to \bigcup_{i=1}^n \sem{\tau_i} \mid r(\ell_i) \in \sem{\tau_i} \} \\
\textstyle\sem{\Gamma} = \{ \gamma : \{ x_1, \dots, x_n \} \to \bigcup_{i=1}^n \sem{\tau_i} \mid \gamma(x_i) \in \sem{\tau_i} \} \text{ for contexts } \Gamma = x_1: \tau_1, \dots, x_n: \tau_n
\end{gather*}\vspace{-2mm}
\begin{align*}
  \sem{\lam{x} t}(\gamma) &= \left((\gamma[x_1], \dots, \gamma[x_n]), ((v_1, \dots, v_n), v) \mapsto \sem{t}(\gamma[x_1 \mapsto v_1, \dots, x_n \mapsto v_n, x \mapsto v])\right) \\&\qquad \text{ where } \Gamma|_{\freeVars{\lam{x} t}} = x_1: \tau_1, \dots, x_n: \tau_n \\
  \sem{s \; t}(\gamma) &= {f(\mvar{env}, \sem{t}(\gamma)) \text{ where } \sem{s}(\gamma) = (\mvar{env}, f)} \\
  \sem{\scriptsize\begin{aligned}&\forWithUsing{x}{\mvar{xs}}{z}{s\\ &}{v_1 := y_1[u_1], \\ &\quad \dots, v_n := y_n[u_n]}{t}\end{aligned}}(\gamma) &= (t_1, \dots, t_N) \text{ where } \sem{\mvar{xs}}(\gamma) = (a_1, \dots, a_N), z_0 = \sem{s}(\gamma), \text{ and for } j = 1, \dots, N, \\[-1em]
  &\qquad \qquad \qquad (t_j, z_j) = \sem{v_1 := y_1[u_1]; \dots; v_n := y_n[u_n]; t}(\gamma[x \mapsto a_j, z \mapsto z_{j-1}])
\end{align*}\vspace{-5.5mm}
\caption{Semantics of $\inclang$ (selected rules).}
\label{fig:inc-semantics}
\end{figure}

\begin{figure}
\begin{gather*}
  \sem{\tau \to \tau'} = \sem{\tau} \To \sem{\tau'} \qquad
  \textstyle\sem{\Upd{\tau}{\tau'}} = \sem{\tau}^+ \To \sem{\tau'} \text{ where } S^+ = \bigcup_{n = 1}^\infty S^n \\\textstyle\sem{\SubRecord{\ell_1: \tau_1, \dots, \ell_n: \tau_n}} = \{ r : S \to \bigcup_{i=1}^n \sem{\tau_i} \mid S \subseteq \{\ell_1, \dots, \ell_n\}, r(\ell_i) \in \sem{\tau_i} \}
\end{gather*}
\begin{align*}
  \sem{\mkUpd{s}{t}}(\gamma) &= ((x_1, \dots, x_n) \in \sem{\tau}^n) \mapsto y_n \text{ where } \mvar{cache}_1 = \sem{s}(\gamma), f = \sem{t}(\gamma) \\
  &\qquad \text{and } (y_i, \mvar{cache}_{i+1}) = f(x_i, \mvar{cache}_i) \text{ for } i = 1, \dots, n \\
  \sem{\applyUpd{s}{t}}(\gamma) &= (f(x \in \sem{\tau}^1), (\mvar{xs} \in \sem{\tau}^n) \mapsto f((x, \mvar{xs} ) \in \sem{\tau}^{n+1})) \\
  &\qquad \text{where } s : \Upd{\tau}{\tau'} \text{ and } f = \sem{s}(\gamma) : \sem{\tau}^+ \to \sem{\tau'} \text{ and } x = \sem{t}(\gamma) \in \sem{\tau}
\end{align*}
  \vspace{-1.5em}
\mycaption{Semantics of $\corelang$ (selected rules).}
\label{fig:core-semantics}
\end{figure}

\paragraph{The $\corelang$ Calculus.} The semantics of $\corelang$ has standard function types (without environment annotations $\kappa$).
In addition to the ground types of $\genlang$, it has \textit{subrecord types} $\SubRecord{\ell_1 : \sigma_1, \dots, \ell_n : \sigma_n}$. Values of a subrecord may assign data to only a subset of the record's labels.
When accessing a label that is not present in a subrecord, we return a default value $\default[\tau] \in \sem{\tau}$, chosen by induction on $\tau$ (e.g., $0$ for $\Nat$, the empty map for $\NameMap{\tau}$, $x \mapsto \default[\tau']$ for functions $\tau \to \tau'$ etc.).
The core language also has an \textit{updater} type $\Upd{\tau}{\tau'}$. Updaters are constructed using $\mkUpd{s}{t}$ and applied using $\applyUpd{s}{t}$.
We defer the discussion of these constructs to \cref{sec:incrementalization}.

\section{Density Transformation}
\label{sec:density-transformation}

The first step in our pipeline is to translate a probabilistic $\genlang$ program into a deterministic $\inclang$ program that computes its trace distribution's \textit{density function}. This isolates all probabilistic reasoning from incrementalization, which we discuss in \cref{sec:incrementalization}.

\subsection{Densities of Traced Programs with Name Generation}

Given a program of type $\ProbTy{\sigma}{\tau}$, we wish to compute its \textit{density function} of type $\sigma \to \mathbb{R}$. This function maps an execution trace, specifying values for all random choices, to a nonnegative real number, which quantifies the likelihood of the given trace. A trace is likely if the individual random choices are; the overall density is a product of the densities of each recorded primitive choice. To make these ideas precise, we require the following definition from measure theory.

\begin{definition}[Radon-Nikodym derivative]
    Let $\mu$ be a probability distribution on a space $X$, and $\nu$ a measure on $X$, called the \textit{reference measure}. A function $\frac{\dif{\mu}}{\dif{\nu}} : X \to \RR$ is called a \textit{density} or \textit{Radon-Nikodym derivative} of $\mu$ with respect to $\nu$ if for all $f : X \to \RR$, $\mathbb{E}_{\mu}[f] = \int f(x) \frac{\dif{\mu}}{\dif{\nu}}(x) \nu(\dif{x})$.
\end{definition}

Here, $\mathbb{E}_\mu[f]$ denotes the expected value of $f$ with inputs sampled from $\mu$. Radon-Nikodym derivatives generalize the standard notions of \textit{probability mass function} and \textit{probability density function}: when $\nu$ is the counting measure $\#_X$ on a set $X$, $\frac{\dif{\mu}}{\dif{\nu}}$ is $\mu$'s mass function,
and when $\nu$ is the Lebesgue measure $\Lambda_\RR$, it is $\mu$'s density function.
Probabilistic programs generally make both discrete and continuous choices, so their traces are typically heterogeneous products. The standard approach to thinking about densities of such programs is to use inductively defined, type-indexed reference measures $\nu_\sigma$, designed so that the overall density will be a product of the mass functions of all discrete choices and the density functions of all continuous choices:
\begin{align*}
\nu_\Real = \Lambda_\Real \quad \quad \nu_\Bool = \#_\Bool \quad\quad \nu_\Nat = \#_\Nat \quad \quad \nu_{\sigma \times \tau} = \nu_\sigma \otimes \nu_\tau \quad\quad \dots
\end{align*}

Our setting, however, is complicated by the presence of \textit{names}. Semantically, $\sem{\Name}$ is the real interval $[0, 1]$. However, whether a randomly chosen name is a \textit{continuous} or a \textit{discrete} choice (i.e., whether we should use a \textit{density} or \textit{mass} function for its likelihood) varies from program to program. For example, consider the program in \cref{fig:density-example}, which generates several fresh names and later chooses one of them uniformly at random. The fresh names (at address $\lbl{opts}$) are continuous samples, but the selected name (at address $\lbl{chosen}$) is a discrete choice.
Thus, the trace density should have as factors both the \textit{density} of the $\const{fresh}$ primitive,\footnote{The distribution over sets of $n$ fresh names is the uniform distribution on $\{(x_1, \dots, x_n) \in [0,1]^n \mid x_1 < \dots < x_n\}$, a set with volume $\frac{1}{n!}$. Thus the density at any point within this set is $1/\frac{1}{n!} = n!$, and the density at other points is $0$.} and the \textit{mass} of the $\const{unif}$ primitive.

\begin{figure}
\hfill
\begin{minipage}[t]{0.27\textwidth}
\codeshade{
\centering
\setlength{\jot}{0pt}
$\begin{aligned}[t]
  &\keyword{def}\ \const{program} := \\
  &\quad \var{n} \sim \const{poisson}(3); \\
  &\quad \var{opts} \sim \const{fresh}(\var{n}+1); \\
  & \\
  &\quad \var{xs} \sim \keyword{for}\ \_ \ \keyword{inSet}\ \var{opts}\ \{ \\
  &\quad\quad \sample{\const{normal}(0,1)} \\
  &\quad \}; \\
  &\quad \var{chosen} \sim \const{unif}(\var{opts}); \\
  &\quad {\return{\var{xs}[\var{chosen}]}}
\end{aligned}$
}
\end{minipage}
\hfill
\begin{minipage}[t]{0.53\textwidth}
\codeshade{
\centering
\setlength{\jot}{0pt}
$\begin{aligned}[t]
  &\keyword{def}\ \const{density} := \lam{(\var{t},\var{U})} \\
  &\quad \var{w_1} := \const{poissonD}(3)(\var{t}.\lbl{n},\var{U}); \\
  &\quad \var{w_2} := \const{freshD}(\var{t}.\lbl{n}+1)(\var{t}.\lbl{opts},\var{U}); \\
  &\quad \var{U} := \var{U} \cup \var{t}.\lbl{opts}; \\
  &\quad \var{w_3} := \const{product}(\keyword{for}\ \var{o}\ \keyword{in}\ \var{t}.\lbl{opts}\ \keyword{using}\ (\var{x} := \var{t}.\lbl{xs}[\var{o}])\ \{ \\
  &\quad\quad \const{normalD}(0,1)(\var{x},\var{U}) \\
  &\quad \}); \\
  &\quad \var{w_4} := \const{unifD}(\var{t}.\lbl{opts})(\var{t}.\lbl{chosen},\var{U}); \\
  &\quad {(\var{t}.\lbl{xs}[\var{t}.\lbl{chosen}],\,
  \var{w_1}\cdot\var{w_2}\cdot\var{w_3}\cdot\var{w_4})}
\end{aligned}$
}
\end{minipage}
\hfill{}{}

\[
\frac{\dif{\mu}}{\dif{\nu}_\emptyset^\sigma}(t)=
\Poisson(t.\lbl{n};3)\cdot
\underbrace{
  (t.\lbl{n}+1)!\cdot
  \mathbf{1}\!\left[|t.\lbl{opts}|=t.\lbl{n}+1\right]
}_{\const{freshD}(t.\lbl{n}+1)(t.\lbl{opts},\emptyset)}
\cdot
\prod_{o \in t.\lbl{opts}} \mathcal{N}(t.\lbl{xs}[o];0,1)\cdot
\underbrace{\frac{\mathbf{1}[t.\lbl{chosen}\in t.\lbl{opts}]}{t.\lbl{n}+1}}_{\const{unifD}(t.\lbl{opts})(t.\lbl{chosen}, t.\lbl{opts})}
\]
\vspace{-5mm}
\caption{Computing the density of a probabilistic program.
In this figure, $\const{distD}$ denotes the density function of a distribution $\const{dist}$, and $\sigma$ (used in $\nu_\emptyset^\sigma$) is the trace type of $\const{program}$.}
\label{fig:density-example}
\end{figure}

Our inductively defined reference measures, in \cref{fig:reference-measures}, formalize this intuition. The reference measures are indexed by both a type $\sigma$ and a \textit{support} $U$, which is a countable subset of $\sem{\Name}$. The reference measure $\nu^\Name_U$ is a \textit{sum} of $\Lambda_{[0,1]}$ and $\#_U$, indicating that randomly chosen names may either be continuous choices (i.e., freshly generated), or discrete choices from the support set $U$ of available names.
The reference measure for product types makes the names that appear in the first component ($\names{a}$) available as new atoms for discrete choices in the second component.\footnote{The order of products is arbitrary, in that for all ground types $\sigma$, $\tau$, we have $\nu_U^{\sigma \times \tau} = \swap_*\nu_U^{\tau \times \sigma}$ (proven in \cref{sec:stock-measures}).} These definitions give us a sufficiently general notion of \textit{density} to handle all programs in $\genlang$:
\begin{lemma}
    Let $\vdash t : \ProbTy{\sigma}{\tau}$ in $\genlang$. Then $\sem{t}_1$ has a density $\frac{\dif{\sem{t}_1}}{\dif{\nu}_\emptyset^\sigma} : \sem{\sigma} \to \RR$ with respect to $\nu_\emptyset^\sigma$.
\end{lemma}

\begin{figure}
\begin{tabular}{rllrll}
$\nu_{U}^\sigma$ &=& $\#_{\sem{\sigma}} \text{ for } \sigma \in \{\Unit, \Bool, \Nat\}$ &
$\nu_U^\Real$ &=& $\Lambda_\Real$ \\
$\nu_U^\Name$ &=& $\#_U + \Lambda_{[0,1]}$ &
$\nu_U^{\sigma \times \tau}$ &=& $a \leftsquigarrow \nu_U^\sigma; b \leftsquigarrow \nu_{U \cup \names{a}}^\tau; \ret{(a,b)}$ \\
$\nu_U^{\List{\sigma}}$ &=& $\sum_{n=0}^\infty \nu_U^{\sigma^n}$ &
$\nu_U^{\NameMap{\sigma}}$ &=& $\sum_{n=0}^\infty \nu_U^{(\Name \times \sigma)^n}$
\end{tabular}
\mycaption{Indexed families of reference measures, defined inductively on the ground types of $\genlang$.}
\label{fig:reference-measures}
\end{figure}

\subsection{Density Transformation}
We now present a program transformation, $\GenToInc{-}$, which mechanically translates probabilistic programs into deterministic ones computing their density functions. More precisely, we translate a program of type $\ProbTy{\sigma}{\tau}$ into a deterministic program with two inputs and two outputs, as illustrated in \cref{fig:density-example}. The inputs are a \textit{trace} $\var{t}$ and a set of previously generated names $\var{U}$. The outputs are a return value---the value returned by the program when its random choices are fixed to those recorded in $\var{t}$---and a nonnegative real, the density of the program with respect to $\nu_{\var{U}}^\sigma$ at the trace $\var{t}$.

Intuitively, the density program works by executing the same logic as the original probabilistic program. When a primitive random choice is encountered, we look up the corresponding value in the trace. We use the primitive distribution's built-in density (passing in the choice's value and the current support set $\var{U}$ of names) to compute its likelihood. When we find new names in the trace, we add them to $\var{U}$.\footnote{Tracking the support $\var{U}$ allows the density function to return $0$ when multiple purportedly fresh names in a trace are identical. In particular, the built-in density $\const{freshD}(n)(\var{names}, \var{U})$ is either $n!$ (if $|\var{names}|=n$ and $\var{names} \cap \var{U} = \emptyset$), or $0$ otherwise. This logic is necessary for the density to be technically correct, but properly implemented inference algorithms should never consider traces that trigger freshness violations. As such, much like array bounds checking, support tracking is a feature that can be turned off for performance after the user is confident that inference is correctly implemented.} Then we compute the return value and the product of all primitive densities.

This algorithm is formalized as a program transformation in \cref{fig:gen-to-inc-types,fig:gen-to-inc-terms}. Given an open term $\Gamma \vdash t : \tau$ in $\genlang$, we translate it to a term $\GenToInc{\Gamma} \vdash \GenToInc{t} : \GenToInc{\tau}$ in the target language, $\inclang$. Because the density function largely follows the same logic as the original probabilistic program, $\GenToInc{-}$ is defined to ``pass through'' many constructs, sending tuples to tuples, projections to projections, abstractions to abstractions, and applications to applications (preserving the environment annotations $\kappa$ on higher types). The interesting cases are how it handles probabilistic program terms. Deterministic programs $\return{t}$ have constant density $1$ at the empty trace. The density of a program that makes a single primitive random sample, $\sample{t}$, is just the density of the primitive distribution $t$. Compound programs $x \gets t_1; t_2$ are handled by separately computing the density of $t_1$ (on the subtrace containing the choices made by $t_1$) and the density of $t_2$ (on the subtrace containing the choices made by $t_2$), then taking their product. (Any new names appearing in the trace of $t_1$ are also unioned into the running support set $\var{U}$, which is passed to the density for $t_2$. This reflects that fresh names generated by $t_1$ may occur non-fresh, as discrete choices, in $t_2$.)

Note that loops in $\genlang$ are translated into more elaborate loops in $\inclang$. The target-language loops iterate over a collection of subtraces, computing the density of each subtrace under the loop body while maintaining a growing set of used names $V$. These densities are then multiplied together---along with the density of the distribution used to decide how many iterations to run, if this was a random choice---to yield a density for the overall loop.

\begin{figure}
  \begin{tabular}{rllrll}
    $\GenToInc{\sigma}$ &=& $\sigma \text{ for ground types } \sigma$ &
    $\GenToInc{\sigma \times \tau}$ &=& $\GenToInc{\sigma} \times \GenToInc{\tau}$ \\
    $\GenToInc{\ProbTy[\kappa]{\sigma}{\tau}}$ &=& $\GenToInc{\sigma} \times \NameSet \clos[\GenToInc{\kappa}] \GenToInc{\tau} \times \Real$ &
    $\GenToInc{\sigma \to_\kappa \tau}$ &=& $\GenToInc{\sigma} \clos[\GenToInc{\kappa}] \GenToInc{\tau}$
  \end{tabular}
  \mycaption{Density transformation from $\genlang$ to $\inclang$ on types (selected rules).}
  \label{fig:gen-to-inc-types}
\end{figure}

\begin{figure}
  \begin{mathpar}
    \GenToInc{c} = c_{\mathsf{inc}} \and
    \GenToInc{\return{t}} = \lam{(\tr, \var{U})}{(\GenToInc{t}, 1)} \and
    \GenToInc{\sample{t}} = \lam{(\tr, \var{U})} (\tr, (\GenToInc{t}(\tr, \var{U})).2) \and
    \GenToInc{t \at \ell} = \lam{(\tr, \var{U})} \GenToInc{t}(\tr.\ell, \var{U})
  \end{mathpar}
  \resizebox{\textwidth}{!}{$
  \begin{aligned}
    \setlength{\jot}{0pt}
    \GenToInc{x \gets t_1; t_2} &= \lam{(\tr, \var{U})} (x, w_1) := \GenToInc{t_1}(\restrict{\tr}{L_1}, \var{U});
     (y, w_2) := \GenToInc{t_2} (\restrict{\tr}{L_2}, \var{U} \cup \names{\restrict{\tr}{L_1}});  (y, w_1 \cdot w_2) \\
      & \quad \text{ where } t_1 : \ProbTy{\Record{\ell: \sigma_\ell \mid \ell \in L_1}}{\tau_1} \text{ and } t_2 : \ProbTy{\Record{\ell: \sigma_\ell' \mid \ell \in L_2}}{\tau_2}
    \end{aligned}
    $} \\[0.5em]
    \resizebox{\textwidth}{!}{\(\displaystyle
    \GenToInc{\begin{aligned}
      &\forRangeWithUsing{x}{d: \Dist{\Nat}}{z}{s \\
      &}{v_1 := y_1[u_1], \dots, v_n := y_n[u_n]}{t}
    \end{aligned}} = \left(\begin{aligned}
      \lam{(\tr, \var{U})} {} & \var{N} := \const{length} \; \tr; \; (\_, \var{w}_{\var{N}}) := \GenToInc{d}(\var{N}, \var{U}); \\
        &\var{ys} := \forWithUsing{x}{\const{range} \; \var{N}}{(z, \var{V})}{(\GenToInc{s}, \var{U}) \\
        &\qquad }{v_1 := y_1[\GenToInc{u_1}], \dots, v_n := y_n[\GenToInc{u_n}], \\
        &\qquad \quad \tr' := \tr[x]}{  ((r, z), w) :=  \GenToInc{t}(\tr', \var{V}); \\ &\qquad \quad ((r, w), (z, \var{V} \cup \names{\tr'})) }; \\
        &(\var{rs}, \var{ws}) := \const{unzip} \; \var{ys}; \; (\var{rs}, \var{w}_{\var{N}} \cdot \const{product} \; \var{ws})
    \end{aligned} \right)\)}
  \vspace{-0.5em}
  \mycaption{Density transformation from $\genlang$ to $\inclang$ on terms (selected rules).}
  \label{fig:gen-to-inc-terms}
\end{figure}

\subsection{Correctness}
\label{sec:density-transformation:correctness}

The density transformation is correct in the following sense:

\begin{theorem}
    Let $t$ be a closed $\genlang$ term of type $\ProbTy{\sigma}{\tau}$, with trace distribution $\mu = \sem{t}_1$. Then for all traces $\rho \in \sem{\sigma}$, $\frac{\dif{\mu}}{\dif{\nu}_\emptyset^\sigma}(\rho)=\sem{(\GenToInc{t} (x, \emptyset)).2}(x \mapsto \rho)$.
\end{theorem}

That is, running the density program on a trace $\rho$ and the empty name set $\emptyset$, then extracting the second component of the output, does yield the correct density of the program's trace distribution, evaluated at $\rho$. This result is established by a logical relations argument, developed in \cref{sec:full-correctness-density}.

\section{Incrementalization}
\label{sec:incrementalization}

We now describe a program transformation $\IncToCore{-}$ for incrementalizing a deterministic program. The incremental version of a function $f : \sigma \to \tau$ (on ground types $\sigma$ and $\tau$) has type
\[ f' : \sigma \to \tau \times \Upd{\Changed{\sigma}}{\Changed{\tau}} \]
Like $f$, it accepts an argument of type $\sigma$, and computes an initial output of type $\tau$.
Additionally, it returns an \emph{updater}, which encapsulates a cache of intermediate results and can thus efficiently compute output changes given changes to the input.
Concretely, let $(y_1, u_1) := f'(x_1)$ and $\mvar{dx}_1$ an input change from $x_1$ to $x_2$.
Then $\applyUpd{u_1}{\mvar{dx}_1} =: (\mvar{dy}_1, u_2)$ yields an output change $\mvar{dy}_1$ from $y_1 := f(x_1)$ to $y_2 := f(x_2)$, and a new updater $u_2$ for changes on top of $x_2$.
This allows us the flexibility to freely stack input changes and compute the corresponding output changes.

The key idea of the transformation is that each primitive operation (e.g., arithmetic operations and list/map operations) has an incremental version as above.
Our program transformation composes these incremental primitives and wires all the change information and updaters together.
Even if a given primitive does not benefit from incrementalization itself, it is useful to propagate information about its output changes to downstream computations.
This is most valuable for loops, enabling us to rerun as few loop iterations as possible, which can yield asymptotic speedups.

\subsection{Change Types and Updaters}

We associate every type $\tau$ with a change type $\Changed{\tau}$, which represents changes to values of type $\tau$ (\cref{fig:inc-to-core-types}).
For basic types like $\Real$ or $\Bool$, a change is simply a new value together with a flag indicating whether the value has changed.
For product types, we make a distinction between pairs and records: a change to a pair is a pair of changes to the components, whereas a change to a record is a subrecord of changes to the fields (omitting unchanged fields).
Changes to collections (lists and maps) store changes to the elements at specific indices or keys and insertions/deletions of elements.
For closures, we do not allow the function itself (``its code'') to change, only its captured environment, to keep updates efficient.
Finally, changes to contexts are simply changes to the variables in the context.
The change types are designed to strike a balance between keeping change representations small and including enough information to avoid caching.
For instance, changes to atomic ground types include the new value, so operations do not need to cache their arguments.

\paragraph{Updaters.}
Updaters encapsulate cached intermediate values and the logic for using them to efficiently recompute an output.
We define $\Upd{\Changed{\tau}}{\Changed{\tau'}}$ as an \emph{opaque} type with a constructor $\mkUpd{\var{cache}}{\var{update}}$ accepting an initial $\var{cache} : \sigma$ and a function $\var{update} : \Changed{\tau} \times \sigma \to \Changed{\tau'} \times \sigma$, which takes an input change and a cache and returns an output change and a new cache.
Note that updater types hide the cache type $\sigma$ after construction, like existential types.
This is an innovation over previous work, and enables type-directed reasoning about updaters.

To use an updater $\var{u}$, one applies it to a change $\var{dx}$, written $\applyUpd{\var{u}}{\var{dx}} =: (\var{dy}, \var{u}')$, which yields an output change $\var{dy}$ and a new updater $\var{u}'$ that can be used to compute output changes on top of $\var{dx}$.
Updater construction and application satisfy the following program equivalence:
\[ \applyUpd{\mkUpd{\var{cache}}{\var{upd}}}{\var{dx}} = \{ (\var{dy}, \var{cache}') := \var{upd} \; (\var{dx}, \var{cache}); (\var{dy}, \mkUpd{\var{cache}'}{\var{upd}}) \} \]

\paragraph{Semantics of Updaters.}
Semantically, updaters are treated as coinductive data types:
\[ \text{\emph{Intuition: }} \quad \keyword{codata} \; \Upd{\Changed{\tau}}{\Changed{\tau'}} \approx \Changed{\tau} \to \Changed{\tau'} \times \Upd{\Changed{\tau}}{\Changed{\tau'}} \]
Thus its denotation is an infinite-depth tree whose edges are labeled by values of $\Changed{\tau}$ and nodes by values of $\Changed{\tau'}$ (see \cref{fig:core-semantics}).
(Category theoretically, this is the final coalgebra of the functor $F(X) = \sem{\Changed{\tau}} \To \sem{\Changed{\tau'}} \times X$.)
Formally, we represent such trees as functions $\sem{\Changed{\tau}}^+ \to \sem{\Changed{\tau'}}$ from non-empty lists of $\Changed{\tau}$ (representing the path in the tree) to $\Changed{\tau'}$ (the value at the node reached at the end of that path).
Note that the idea of recursive updaters was mentioned by \citet{Morihata20}, but to our knowledge has not been formalized before.

The constructor $\mkUpd{\var{cache}}{\var{update}}$ builds a tree where the value of the node $y_n$ reached along the path $(x_1, \dots, x_n)$ results from iteratively applying $\var{update}$ to the inputs $x_1, \dots, x_n$, together with the most recent cache, starting from the given $\var{cache}$ (see \cref{fig:core-semantics}).
The application $\applyUpd{\var{u}}{\var{dx}}$ returns the node value in $\var{u}$ at the path with single edge $\var{dx}$, and the updater rooted at that node.

\subsection{Program Transformation}

\paragraph{Type Transformations.}
Our program transformation $\IncToCore{-}$ behaves like the identity on types (see \cref{fig:inc-to-core-types}), except for closure types, which are translated into incremental versions
\[ \IncToCore{\sigma \clos[\kappa] \tau} := \IncToCore{\sigma} \to \IncToCore{\tau} \times \Upd{\Changed{\kappa \times \sigma}}{\Changed{\tau}} \]
with updaters that can process both input changes and changes to captured variables in $\kappa$.

\paragraph{Primitives.}
For every constant $c: \tau$ in $\inclang$, we assume a corresponding incrementalized version $c_{\mathsf{core}} : \IncToCore{\tau}$ in $\corelang$.
For instance, consider the function $\const{length} : \List{\Nat} \to \Nat$ in $\inclang$.
Its incremental version caches the length and constructs an updater with an update function $\const{updateLength}$:\footnote{We omit the empty environment change $\Changed{\kappa} = \Changed{\Record{}} = \SubRecord{}$ for simplicity.}
\begin{align*}
  \setlength{\jot}{0pt}
  &\keyword{def}\ \IncVer{length} := \lam{\var{xs}} \var{len} := \const{length}(\var{xs}); \; (\var{len}, \mkUpd{\var{len}}{\const{updateLength}}) \\
  &\keyword{def}\ \const{updateLength} := \lam{(\var{dxs}, \var{cachedLen})} \\
  &\qquad \var{len} := \var{cachedLen} + \const{length}(\var{dxs}.\lbl{insert}) - \const{length}(\var{dxs}.\lbl{remove}); \\
  &\qquad (\recordLit{\field{new} \var{len}, \field{same} \var{len} == \var{cachedLen}}, \var{len})
\end{align*}
The function $\const{updateLength}$ computes the new length $\var{len}$ taking insertions and deletions into account.
It returns the length change of type $\Changed{\Nat}$ together with the new length to be cached.

\paragraph{Composition.}
Suppose we have a function $\const{evens} : \List{\Nat} \to \List{\Nat}$ that filters out odd numbers from a list and we compose it with the length function: $\const{countEvens}(\var{xs}) := \const{length}(\const{evens}(\var{xs}))$.
Intuitively, to incrementalize it, we can take the incremental versions of the two functions and create a new updater that uses the component updaters as its cache:\footnote{This is a simplified version of what our program transformation produces; we again ignore empty environments.}

\vspace{0.5em}\noindent
\begin{minipage}[t]{0.48\linewidth}
$\begin{aligned}
  \setlength{\jot}{0pt}
  &\keyword{def}\ \IncVer{countEvens} := \lam{(\var{xs})} \\
  &\qquad (\var{ys}, u_{\const{evens}}) := \IncVer{evens}(\var{xs}); \\
  &\qquad (\var{len}, u_{\const{length}}) := \IncVer{length}(\var{ys}); \\
  &\qquad (\var{len}, \mkUpd{(u_{\const{evens}}, u_{\const{length}})}{\const{updateCE}})
\end{aligned}$
\end{minipage}\hfill
\begin{minipage}[t]{0.5\linewidth}
$\begin{aligned}
  \setlength{\jot}{0pt}
  &\keyword{def}\ \const{updateCE} := \lam{(\var{dxs}, (u_{\const{evens}}, u_{\const{length}}))} \\
  &\qquad (\var{dy}, u_{\const{evens}}') := \applyUpd{u_{\const{evens}}}{\var{dxs}}; \\
  &\qquad (\var{dlen}, u_{\const{length}}') := \applyUpd{u_{\const{length}}}{\var{dy}}; \\
  &\qquad (\var{dlen}, (u_{\const{evens}}', u_{\const{length}}'))
\end{aligned}$
\end{minipage}

Our program transformation automates this logic.
A term $\tyjudg{\Gamma}{t}{\tau}$ in $\inclang$ is translated to a term $\tyjudg{\IncToCore{\Gamma}}{\IncToCore[\Gamma]{t}}{\IncToCore{\tau} \times \Upd{\Changed{\toType{\Gamma}}}{\Changed{\tau}}}$ in $\corelang$ as defined in \cref{fig:inc-to-core-terms}.
The idea is that we must not only be able to evaluate a given term, but also to obtain changes to its evaluation depending on environment changes.
To handle such changes in an updater, we reify them as values of type $\Changed{\toType{\Gamma}}$: tuples of changes to the variables in $\Gamma$ (using $\toType{-}$ from \cref{fig:type-system}).

The transformation rule for variables (\cref{fig:inc-to-core-terms}) yields a cache-less updater that simply projects out the variable change from the environment change (using $\projectVar{\Gamma}{x}$ from \cref{fig:project-var}).
For abstractions $\lam{x} t : \sigma \clos[\kappa] \tau$, we translate the body $t$ in the environment that includes exactly the free variables, which yields an updater $\Upd{\Changed{\kappa \times \sigma}}{\Changed{\tau}}$ for the closure change $\Changed{\kappa}$ and the input change $\Changed{\sigma}$.
For applications $s \; t : \tau$, the translation of $s$ yields an updater for closure changes and $t$ for input changes, which are composed to yield an updater for the application of the closure.
For record access $t.\ell : \tau$, we translate the term $t$ to an updater for changes to the record.
However, such a change is a subrecord that may not contain the field $\ell$.
If a trivial change $\same[\tau]$ exists (e.g., $\unit$ for record types), we can use it to obtain the change to the field.
Otherwise, we have to cache the previous value $r_l$ and use the trivial change $\sameAs(r_l)$ if $\ell$ is not in the record change (see \cref{fig:diff-helpers}).

\begin{figure}
  \begin{gather*}
    \IncToCore{\sigma} = \sigma \text{ for } \sigma \text{ ground} \qquad
    \IncToCore{\sigma \times \tau} = \IncToCore{\sigma} \times \IncToCore{\tau} \qquad \IncToCore{\sigma \clos[\kappa] \tau} = \IncToCore{\sigma} \to \IncToCore{\tau} \times \Upd{\Changed{\kappa \times \sigma}}{\Changed{\tau}} \\
    \IncToCore{\Gamma} = (x_1: \IncToCore{\sigma_1}, \dots, x_n: \IncToCore{\sigma_n}) \text{ for contexts } \Gamma = (x_1: \sigma_1, \dots, x_n: \sigma_n) \\
    \Changed{\sigma} = \Record{\lbl{new}: \sigma, \lbl{same}: \Bool} \quad \text{ for } \sigma = \Bool, \Nat, \Real, \Name \qquad
    \Changed{\sigma \times \tau} = \Changed{\sigma} \times \Changed{\tau} \\
    \Changed{\List{\sigma}} = \Record{\lbl{insert}: \List{(\Nat \times \IncToCore{\sigma})}, \lbl{remove}: \List{\Nat}, \lbl{change}: \List{(\Nat \times \Changed{\sigma})}} \\
    \Changed{\NameMap{\sigma}} = \Record{\lbl{add}: \NameMap{\IncToCore{\sigma}}, \lbl{remove}: \NameSet, \lbl{change}: \NameMap{\Changed{\sigma}}} \\
    \Changed{\Record{\ell_1: \tau_1, \dots, \ell_n: \tau_n}} = \SubRecord{\ell_1: \Changed{\tau_1}, \dots, \ell_n: \Changed{\tau_n}} \qquad
    \Changed{\sigma \clos[\kappa] \tau} = \Changed{\kappa} \\
    \Changed{\Gamma} = (x_1 : \Changed{\sigma_1}, \dots, x_n: \Changed{\sigma_n}) \text{ for contexts } \Gamma = (x_1: \sigma_1, \dots, x_n: \sigma_n)
  \end{gather*}
  \vspace{-1.5em}
  \mycaption{Incrementalizing transformation $\IncToCore{-}$ from $\inclang$ to $\corelang$ on types and change types $\Changed{-}$.}
  \label{fig:inc-to-core-types}
\end{figure}

\begin{figure}
  \begin{gather*}
    \begin{aligned}
    \projectVar{\Gamma}{x} &: \toType{\Gamma} \to \sigma \quad \text{(assuming $(x: \sigma) \in \Gamma$)} \\
    \projectVar{\Gamma, x : \tau}{x} & = \lam{(\_, x)} x \\
    \projectVar{\Gamma, y : \tau}{x} &= \lam{(\gamma, \_)} \projectVar{\Gamma}{x}(\gamma) \text{ if } y \neq x
    \end{aligned}
    \;
    \begin{aligned}
    \projectVars{\Gamma}{\mvar{xs}} &: \toType{\Gamma} \to \toType{\Gamma|_{\mvar{xs}}} \\
    \projectVars{\bullet}{\emptyset} &= \lam{\_} \unit \\
    \projectVars{(\Gamma, x: \tau)}{\mvar{xs}} &= \begin{cases}
      \lam{(\gamma, v)} (\projectVars{\Gamma}{\mvar{xs} \setminus \{x\}}(\gamma), v) &\text{if } x \in \mvar{xs} \\
      \lam{(\gamma, \_)} \projectVars{\Gamma}{\mvar{xs}}(\gamma) &\text{otherwise}
    \end{cases}
    \end{aligned}
  \end{gather*}
  \vspace{-1.5em}
  \mycaption{Functions to project out variables from contexts and tuples.}
  \label{fig:project-var}
  \end{figure}

\begin{figure}
\begin{align*}
    \same[\tau] &: \Changed{\tau} &&\text{Trivial change} &&\text{(only supported for some $\tau$)} \\
    \sameAs[\tau] &: \tau \to \Changed{\tau} &&\text{Trivial change at a particular value} &&\text{(for ground types $\tau$)} \\
    \apply[\tau] &: \tau \times \Changed{\tau} \to \tau &&\text{Apply a change to a value} &&\text{(for ground types $\tau$)}
\end{align*}
  \vspace{-1.5em}
\mycaption{Additional primitives involving changes, needed for incrementalization.}
\label{fig:diff-helpers}
\end{figure}

\begin{figure}
  \begin{align*}
    \setlength{\jot}{0pt}
    \IncToCore[\Gamma]{x} &= (x, \mkUpd{\unit}{\lam{(d\gamma, \_)} (\projectVar{\Changed{\Gamma}}{x} \; d\gamma, \unit)}) \\
    \IncToCore[\Gamma]{\lam{x} t} &=
      (\lam{x} \IncToCore[\Gamma|_{\freeVars{\lam{x} t}}, x: \sigma]{t}, \; \mkUpd{\unit}{\lam{(d\gamma, \_)} (\projectVars{\Changed{\Gamma}}{\freeVars{\lam{x} t}} \; d\gamma, \unit)}) \\
    \IncToCore[\Gamma]{s \; t} &=
      (f, u_f) := \IncToCore[\Gamma]{s}; \;\; (x, u_x) := \IncToCore[\Gamma]{t}; \;\; (y, u_y) := f \; x; \\
      &\quad \big(y, \;\; \mkUpd{ \;\; (u_f, u_x, u_y)}{\;\;\lam{(d\gamma, (u_f, u_x, u_y))} \;\; (\mvar{df}, u_f') := \applyUpd{u_f}{d\gamma}; \\
      &\qquad \qquad \qquad (\mvar{dx}, u_x') := \applyUpd{u_x}{d\gamma}; \;\; (\mvar{dy}, u_y') := \applyUpd{u_y}{(\mvar{df}, \mvar{dx})}; \;\; (\mvar{dy}, (u_f', u_x', u_y'))}\big) \\
    \IncToCore[\Gamma]{t.\ell} &=
      (r, u_r) := \IncToCore[\Gamma]{t}; \; (r.\ell, \mkUpd{(r.\ell, u_r)}{\lam{(d\gamma, (r_l, u_r))} \;\; (\mvar{dr}, u_r') := \applyUpd{u_r}{d\gamma}; \\
      &\qquad \qquad \qquad \qquad \qquad \ite{\has[\ell]{\mvar{dr}}}{(\mvar{dr}.\ell, (\apply(r_l, \mvar{dr}.\ell), u_r'))}{(\sameAs(r_l), (r_l, u_r'))}})
  \end{align*}
  \vspace{-1.5em}
  \mycaption{Incrementalizing transformation from $\inclang$ to $\corelang$ on terms (selected rules).
  }
  \label{fig:inc-to-core-terms}
\end{figure}

\begin{figure}
\begin{minipage}[t]{0.45\linewidth}
\begin{algorithmic}
  \Function{$\IncToCore[\Gamma]{\forIn{x : \tau}{\mvar{xs}}{t}}$}{}
    \State $\var{results} \gets [\,];$
    \State $\var{upds} \gets [\,]$;
    \State $\var{xs} \gets \mvar{xs}$
    \For{$\var{i} \gets 1$ \textbf{to} $\const{length}(\var{xs})$}
      \State $\var{x} \gets \var{xs}[\var{i}]$
      \State $(\var{y}, \var{u}_{\var{y}}) \gets \IncToCore[\Gamma, x: \tau]{t}$
      \State $\const{push}(\var{results}, \var{y})$
      \State $\const{push}(\var{upds}, \var{u}_{\var{y}})$
    \EndFor
    \State $\begin{aligned}[t]
      \var{cache} \gets \recordLit{ &\field{inputs} \var{xs}, \field{upds} \var{upds} }
    \end{aligned}$
    \State \Return $(\var{results}, \mkUpd{\var{cache}}{\textsc{Update}})$
  \EndFunction
\end{algorithmic}

\vspace{1em}\centering\textbf{Cache type:}\\[0.2em]
$\begin{aligned}
  \Record{ &\quad \field{inputs} \List{\tau}, \\
  &\quad \field{upds} \List{\Upd{\tau}{\tau'}} &}
\end{aligned}$
\end{minipage}%
\hfill%
\begin{minipage}[t]{0.55\linewidth}
\begin{algorithmic}
  \Function{Update}{$\var{d\gamma}, \; \var{cache}$}
    \State $\var{dxs} \gets \projectVar{\Changed{\Gamma}}{\var{xs}} \; \var{d\gamma}$
    \State Extract $\var{inputs}$, $\var{upds}$ from $\var{cache}$
    \State $\var{toVisit} \gets \text{indices modified by } \var{dxs}$
    \If{$\var{d\gamma}$ changes any free variable of $t$ besides $\var{xs}$}
      \State $\var{toVisit} \gets \{1,\dots,\const{length}(\var{inputs})\}$
    \EndIf
    \State $\var{changed} \gets \emptyset$
    \For{$\var{i} \gets 1$ \textbf{to} $\const{length}(\var{inputs})$}
      \If{$\var{i} \notin \var{toVisit}$} \textbf{continue} \EndIf
      \State $\var{dx} \gets \var{dxs}[\var{i}]$ or $\sameAs(\var{inputs}[\var{i}])$ if $\var{i} \notin \var{dxs}$
      \State $(\var{dy}, \var{upds}[\var{i}]) \gets \applyUpd{\var{upds}[\var{i}]}{(\var{d\gamma}, \var{dx})}$
      \State $\var{inputs}[\var{i}] \gets \apply(\var{inputs}[\var{i}], \var{dx})$
      \If{$\lnot \issame{\var{dy}}$}
        \State $\var{changed}[\var{i}] \gets \var{dy}$
      \EndIf
    \EndFor
    \State $\var{cache} \gets \recordLit{ \field{inputs} \var{inputs}, \field{upds} \var{upds} }$
    \State \Return $(\recordLit{ \field{changed} \var{changed} }, \var{cache})$
  \EndFunction
\end{algorithmic}
\end{minipage}
\mycaption{Pseudocode for incrementalizing simplified for-expressions (without accumulators or $\using{\cdot}$ clauses).
It only handles modifications to list elements (no insertions or deletions) and assumes $\mvar{xs}$ is a variable.}
\label{fig:inc-to-core-terms-for-simple}
\end{figure}

\paragraph{Incrementalizing Loops.}
A loop $\forIn{x}{\mvar{xs}}{t}$ can be viewed as a sophisticated, higher-order primitive $\const{loopCombinator}(\mvar{xs}, \lam{x} t)$ that takes a collection $\mvar{xs}$ and a loop body $t$ and returns a loop result.
Such a combinator-like primitive can be equipped with its own (complex) incremental version.
To minimize recomputation, the loop combinator caches the values produced at each iteration and maintains mappings between collection indices and the iterations that accessed them.
The collections that are accessed by the loop body are communicated to the system via the $\using{\var{element} := \var{collection}[\var{index}], \ldots}$ clause.
It uses this information, together with incoming change descriptions, to determine which iterations need to be rerun, and thus which entries of the loop's output must change.
The incrementalization of simple loops ($\forIn{x}{\mvar{xs}}{t}$) is presented in \cref{fig:inc-to-core-terms-for-simple}.
For fully featured loops, this is more complicated and relegated to the appendix.

\subsection{Correctness}
\label{sec:incrementalization:correctness}

\paragraph{Change Validity.}
In order to reason about the correctness of the incrementalization, we first need to define when a change is valid.
For this we introduce a type-indexed relation $\changerel[\tau] \subseteq \sem{\tau} \times \sem{\Changed{\tau}} \times \sem{\tau}$ where $(x, \mvar{dx}, x') \in \changerel[\tau]$ means that $\mvar{dx}$ is a valid change from $x$ to $x'$.
The relation is defined inductively on types in \cref{fig:change-rel}.
For instance, $\mvar{dx}$ is a valid change from $x$ to $x'$ for basic types if its ``new'' field is $x'$ and its ``same'' field being $\trueLit$ implies $x = x'$.
Changes to tuples are valid if the individual changes to their components are valid.
Finally, a change to a closure is valid if it is a valid change to its captured environment and the function itself is unchanged.

\begin{figure}
  \vspace{-0.5em}
  \begin{align*}
    \changerel[\sigma] &= \{ (x, \mvar{dx}, x') \in \sem{\sigma} \times \sem{\Changed{\sigma}} \times \sem{\sigma} \mid \mvar{dx}(\lbl{new}) = x' \land (\mvar{dx}(\lbl{same}) =  \trueLit \implies x = x') \} \\
    &\quad \text{ for } \sigma = \Bool, \Nat, \Real, \Name \\
    \changerel[\sigma \times \tau] &= \{ ((x, y), (\mvar{dx}, \mvar{dy}), (x', y')) \in \sem{\sigma \times \tau} \times \sem{\Changed{\sigma \times \tau}} \times \sem{\sigma \times \tau} \mid (x, \mvar{dx}, x') \in \changerel[\sigma] \land (y, \mvar{dy}, y') \in \changerel[\tau] \} \\
    \changerel[{\sigma \clos[\kappa] \tau}] &= \{ ((e, f), \mvar{de}, (e', f')) \in \sem{\sigma \clos[\kappa] \tau} \times \sem{\Changed{\kappa}} \times \sem{\sigma \clos[\kappa] \tau} \mid (e, \mvar{de}, e') \in \changerel[\kappa], f = f' \}
  \end{align*}
  \vspace{-1.5em}
  \mycaption{Relation of valid changes (selected rules).}
  \label{fig:change-rel}
\end{figure}

\paragraph{Updater Validity.}
Similarly, we need to define when an updater $\Upd{\Changed{\sigma}}{\Changed{\tau}}$ is valid for a function $f : \sem{\sigma} \to \sem{\tau}$, previously evaluated at an input $x \in \sem{\sigma}$.
To this end, we introduce the validity predicate $\ValidUpd[\sigma \leadsto \tau]{f}{x} \subseteq \sem{\Upd{\Changed{\sigma}}{\Changed{\tau}}}$, defined \emph{coinductively} (as the greatest fixed point):
\begin{align*}
  \ValidUpd[\sigma \leadsto \tau]{f}{x} = \{ u \in \sem{\Upd{\Changed{\sigma}}{\Changed{\tau}}} &\mid \forall \mvar{dx}, x' \text{ with } (x, \mvar{dx}, x') \in \changerel[\sigma] : (y, \mvar{dy}, y') \in \changerel[\tau] \land u' \in \ValidUpd[\sigma \leadsto \tau]{f}{x'} \\
  &\!\!\!\! \text{ where } y = f(x), y' = f(x'), \text{ and applying $u$ to $\mvar{dx}$ yields }  (\mvar{dy}, u') \}
\end{align*}
In words, $u$ is a valid updater for the function $f$ at the input $x$ if for any valid change $\mvar{dx}$ from $x$ to $x'$, applying the updater to $\mvar{dx}$ yields a valid output change $\mvar{dy}$ from $f(x)$ to $f(x')$ and a new updater $u'$ that is again valid for $f$, but now at input $x'$.

\begin{figure}
  \begin{align*}
    \logrel[\sigma] &= \{ (x, x') \mid x = x' \} \text{ (for ground types } \sigma\text{)} \quad\quad
    \logrel[\sigma \times \tau] = \{ ((x, y), (x', y')) \mid (x, x') \in \logrel[\sigma] \land (y, y') \in \logrel[\tau] \} \\
    \logrel[{\sigma \clos[\kappa] \tau}] &= \{ ((e, f), f') \mid \forall (x, x') \in \logrel[\sigma] \ldotp (y, y') \in \logrel[\tau] \land u \in \ValidUpd[\kappa \times \sigma \leadsto \tau]{f}{(e, x)} \text{ where } y = f(e, x), (y', u) = f'(x') \} \\
    \logrel[\Gamma] &= \{ (\gamma, \gamma') \mid \forall (x: \sigma) \in \Gamma \ldotp (\gamma(x), \gamma'(x)) \in \logrel[\sigma] \}
  \end{align*}
  \vspace{-1.5em}
  \mycaption{Logical relation between $\inclang$ and $\corelang$ (selected rules).}
  \label{fig:logical-rel}
\end{figure}
\begin{figure}
  \begin{gather*}
    \begin{aligned}
    \ctxToTpl[\Gamma] &: \sem{\Gamma} \to \sem{\toType{\Gamma}} \\
    \ctxToTpl[\bullet](\emptyset) &= \unit, \quad \ctxToTpl[\Gamma, x : \tau](\gamma) = (\ctxToTpl[\Gamma](\gamma \setminus \{x\}), \gamma(x))
    \end{aligned}
    \qquad
    \begin{aligned}
    \tplToCtx[\Gamma] &: \sem{\toType{\Gamma}} \to \sem{\Gamma} \\
    \tplToCtx[\bullet](\unit) &= \emptyset, \quad \tplToCtx[\Gamma, x : \tau](\gamma) = \tplToCtx[\Gamma](\pi_1(\gamma))[x \mapsto \pi_2(\gamma)]
    \end{aligned}
  \end{gather*}
  \vspace{-1.2em}
  \mycaption{Functions to convert between contexts and tuples.}
  \label{fig:context-tuple-conversion}
\end{figure}

\paragraph{Logical Relation.}
With this machinery in place, we define a logical relation $\logrel[\tau] \subseteq \sem{\tau} \times \sem{\IncToCore{\tau}}$ between terms in $\inclang$ and $\corelang$, which characterizes when the $\corelang$ term is a correct incremental version of the $\inclang$ term (\cref{fig:logical-rel}).
The relation is the identity at ground types and at composite types, it is defined in the standard way: two values are related if all their components are.

The interesting case is closure types $\sigma \clos[\kappa] \tau$ because they have incremental versions. The semantics of such a closure is $(\mvar{env}, f) \in \sem{\kappa} \times (\sem{\kappa} \times \sem{\sigma} \to \sem{\tau})$, consisting of the captured environment and a “function pointer” for computing outputs from inputs and the environment.
Such a closure is translated to an ordinary $\corelang$ function $f'$ of type $\sem{\IncToCore{\sigma \clos[\kappa] \tau}} = \sem{\IncToCore{\sigma}} \to \sem{\IncToCore{\tau}} \times \sem{\Upd{\Changed{\kappa \times \sigma}}{\Changed{\tau}}}$.
The translation sends transformed inputs to transformed outputs, together with an updater from environment and input changes to output changes.
Now, $(\mvar{env}_0, f)$ and $f'$ are related if for all related inputs $x_0 \in \sem{\sigma}, x_0' \in \sem{\IncToCore{\sigma}}$, letting $y_0 = f(\mvar{env}_0, x_0)$ and $(y_0', u_0) = f'(x_0')$, we have:
(1) the outputs $y_0 \in \sem{\tau}, y_0' \in \sem{\IncToCore{\tau}}$ are related, and
(2) the resulting updater $u_0 \in \sem{\Upd{\Changed{\kappa \times \sigma}}{\Changed{\tau}}}$ is a valid updater of $f$ for the input $(\mvar{env}_0, x_0)$.

To state correctness of $\IncToCore{-}$, we need a few helper functions operating on the tuple representation of contexts (\cref{fig:context-tuple-conversion}).
Using this, we obtain the following result, proven in \cref{sec:incrementalization-proof}.
\begin{lemma}[Fundamental Lemma]
  Assume $\tyjudg{\Gamma}{t}{\tau}$ in $\inclang$, and thus $\tyjudg{\IncToCore{\Gamma}}{\IncToCore[\Gamma]{t}}{\IncToCore{\tau} \times \Upd{\Changed{\toType{\Gamma}}}{\Changed{\tau}}}$ in $\corelang$.
  Then for all $(\gamma, \gamma') \in \logrel[\Gamma]$, we have
\[ (\sem{t}(\gamma), \pi_1(\sem{\IncToCore{t}}(\gamma'))) \in \logrel[\tau] \quad \text{ and } \quad \pi_2(\sem{\IncToCore{t}}(\gamma')) \in \ValidUpd[\toType{\Gamma} \leadsto \tau]{\sem{t} \circ \tplToCtx[\Gamma]}{\ctxToTpl[\Gamma](\gamma)} \]
\end{lemma}

Considering the special case of one-element contexts and unfolding the definition of the validity predicate for updaters to handle sequences of input changes, one obtains the following corollary:

\begin{corollary}
  Let $\tyjudg{x : \sigma}{t}{\tau}$ be a $\inclang$ term denoting a function $f : \sem{\sigma} \to \sem{\tau}$, for ground types $\sigma$, $\tau$.
  Let $g = \sem{\IncToCore{t}}$ and $x_0, x_1, x_2, \ldots \in \sem{\sigma}$ be a sequence of inputs.
  Let $\mvar{dx}_0, \mvar{dx}_1, \mvar{dx}_2, \ldots \in \sem{\Changed{\sigma}}$ be valid change descriptions: $(x_i, \mvar{dx}_i, x_{i+1}) \in \changerel[\sigma]$.
  Finally, let $(y_0, u_0) = g(x_0)$, and for all $i$, let $(\mvar{dy}_i, u_{i+1})$ be the result of applying $u_i$ to $\mvar{dx}_i$.
  Then we have $y_0=f(x_0)$ and $(f(x_i), \mvar{dy}_i, f(x_{i+1})) \in \changerel[\tau]$ for all $i$.
\end{corollary}

\section{Case Studies}
\label{sec:case-studies}

We have implemented a prototype of our approach by embedding the $\genlang$ and $\inclang$ DSLs in Julia. We investigate the impact of incrementalization across a variety of models and custom MCMC kernels. We also compare to Gen~\citep{Cusumano-Towner19,Cusumano-TownerLM20}, a state-of-the-art PPL with incrementalization.
\paragraph{Models and Kernels.} To test our approach, we implement 16 MCMC kernels in 6 Bayesian models. (1) \textit{Robust regression~\citep{BoxT68}:} Bayesian linear regression with outliers. We implement Gaussian drift moves on the parameters and Metropolis-Hastings moves on outlier indicators.
(2) \textit{Binary Gaussian mixture~\citep{DieboltR94}:} 2D Gaussian mixture model with inferred parameters, mixing weight, and cluster assignments. Kernels update parameters and mixing weights, or one cluster assignment.
(3) \textit{Hidden Markov Model with Inferred Parameters~\citep{Scott02}:} A hidden Markov model with an a priori unknown transition matrix in $\RR^{100 \times 100}$, and continuous observations. Our MCMC kernels update a single row of the transition matrix or a single latent state.
(4) \textit{Mixture of Finite Mixtures~\citep{MillerH18}:} A 2D Gaussian mixture model where the number of clusters is unknown with Poisson prior. Empty birth/death moves create new empty clusters or remove unoccupied clusters. Singleton birth/death moves propose to split a datapoint into its own singleton cluster, or to merge a singleton cluster into an existing cluster. Cluster parameter changes propose updates to a single latent cluster's parameters. Mixing weight changes re-propose the mixing proportions of the current latent clusters. Cluster assignment changes reassign a single datapoint to a different (pre-existing) cluster.
 (5) \textit{Latent Dirichlet Allocation~\citep{BleiNJ03}:} A topic model. We consider MCMC moves that change the theme (topic mixture) of a single document, the lexicon (word mixture) of a single topic, or the topic assignment for a single word.
(6) \textit{Dirichlet Process Mixture Model~\citep{Neal00}:} A nonparametric mixture model with infinitely many clusters, marginalized via a Chinese Restaurant Process. We consider Gibbs updates to cluster assignments and parameters as described in Algorithm 8 of~\citet{Neal00}.

\paragraph{Systems.} We implement each model both in our Julia embedding of $\genlang$, and in Gen's static modeling language, with Gen's experimental support for fine-grained incremental computation turned on (using the {\footnotesize\textbf{\texttt{@gen (static, diffs) function}}} syntax). We implement each inference kernel in Julia, using three different strategies for computing the necessary model densities: (1) our own implementation of incremental density computation, (2) our own implementation of full density recomputation, and (3) Gen's implementation of incremental density computation.

\begin{table}[t]
\centering
\mycaption{Fixed-size benchmarks: median time to complete one iteration of a given MCMC kernel.}\vspace{-1em}
\label{tab:benchmarks}
\begin{tabular}{p{0.4\linewidth}rrrrrr}
\toprule
\textbf{Model} & \multicolumn{3}{c}{\textbf{Absolute Runtime}} & \multicolumn{2}{c}{\textbf{Relative Runtime}} \\
\quad\textit{MCMC Kernel} & {Incremental} & {Full} & {Gen} & $\frac{\mathbf{Full}}{\mathbf{Inc}}$ & $\frac{\mathbf{Gen}}{\mathbf{Inc}}$ \\
\midrule
\multicolumn{6}{l}{Robust Regression \textit{(\citet{BoxT68})}, $N=1000$} \\
\quad \textit{Parameter change} & 5.93 ms & 3.05 ms & 1.77 ms & \textcolor{red}{0.51×} & 0.30× \\
\quad \textit{Outlier status change} & 0.017 ms & 3.06 ms & 0.010 ms & 176× & 0.59× \\
\multicolumn{6}{l}{Binary Gaussian Mixture \textit{(\citet{DieboltR94})}, $N = 500$}  \\
\quad \textit{Parameters change} & 3.91 ms & 1.85 ms & 0.78 ms & \textcolor{red}{0.47×} & 0.20× \\
\quad \textit{Cluster assignment change} & 0.028 ms & 1.68 ms & 0.013 ms & 61× & 0.46× \\
\multicolumn{6}{l}{Hidden Markov Model with Inferred Parameters \textit{(\citet{Scott02})}, $T=1000$}  \\
\quad \textit{Transition matrix row update} & 0.211 ms & 3.88 ms & 1.24 ms & 18× & \textcolor{ForestGreen}{5.88×} \\
\quad \textit{Latent state update} & 0.034 ms & 3.81 ms & 0.014 ms & 112× & 0.41× \\
\multicolumn{6}{l}{Mixture of Finite Mixtures \textit{(\citet{MillerH18})}, $N = 200$}  \\
\quad \textit{Birth/death move (empty cluster)} & 0.238 ms & 2.35 ms & 3.61 ms & 9.9× & \textcolor{ForestGreen}{15.17×} \\
\quad \textit{Birth/death move (singleton cluster)} & 0.219 ms & 2.32 ms & 4.10 ms & 10.6× & \textcolor{ForestGreen}{18.72×} \\
\quad \textit{Cluster parameter change} & 2.13 ms & 2.19 ms & 0.761 ms & 1.03× & 0.36× \\
\quad \textit{Mixing weights change} & 0.148 ms & 2.37 ms & 1.37 ms & 16× & \textcolor{ForestGreen}{9.26×} \\
\quad \textit{Cluster assignment change} & 0.225 ms & 8.76 ms & 0.032 ms & 39× & 0.14× \\
\multicolumn{6}{l}{Latent Dirichlet Allocation \textit{(\citet{BleiNJ03})}, $N = 900~ (30 \times 30)$} \\
\quad \textit{Document theme change} & 0.206 ms & 9.07 ms & 0.031 ms & 44× & 0.15× \\
\quad \textit{Topic change} & 0.716 ms & 9.42 ms & 0.835 ms & 13× & {1.17×} \\
\quad \textit{Topic assignment change (MH)} & 0.069 ms & 9.07 ms & 0.019 ms & 131× & 0.28× \\
\multicolumn{6}{l}{Dirichlet Process Mixture Model \textit{(\citet{Neal00})}, $N = 200$} \\
\quad \textit{Cluster assignment update} & 0.183 ms & 6.64 ms & 5.14 ms & 36× & \textcolor{ForestGreen}{28.09×} \\
\quad \textit{Cluster parameter update} & 0.733 ms & 1.28 ms & 0.421 ms & 1.75× & 0.57× \\
\bottomrule
\end{tabular}\vspace{-1em}
\end{table}

\begin{figure}[t]
\begin{subfigure}{0.3\textwidth}
\includegraphics[width=\linewidth]{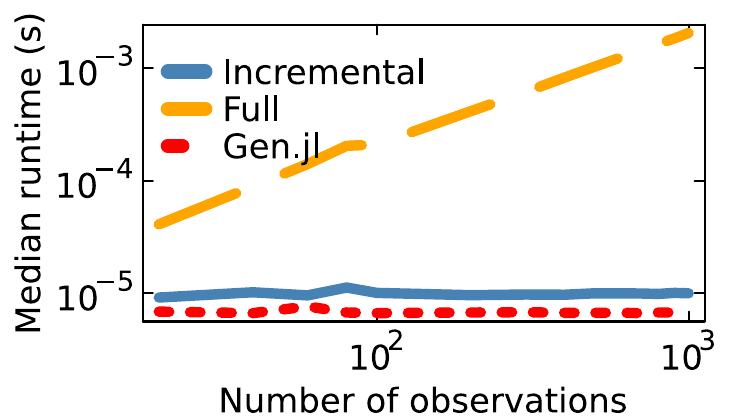}
\subcaption{Outlier Status Change}
\end{subfigure}%
\hfil
\begin{subfigure}{0.3\textwidth}
\includegraphics[width=\linewidth]{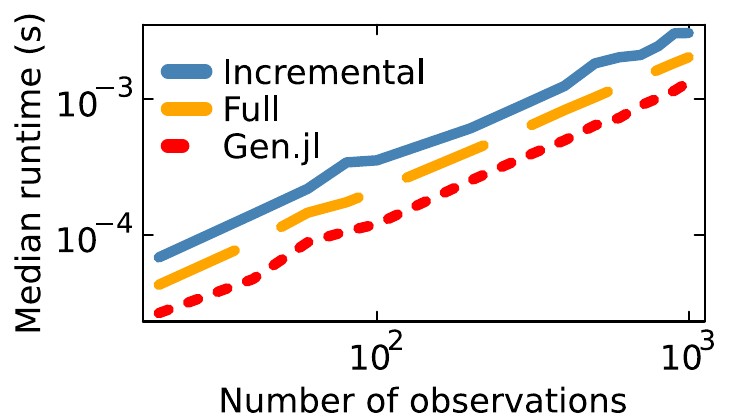}
\subcaption{Parameter Drift Update}
\end{subfigure}%
\hfil
\begin{subfigure}{0.3\textwidth}
\includegraphics[width=\linewidth]{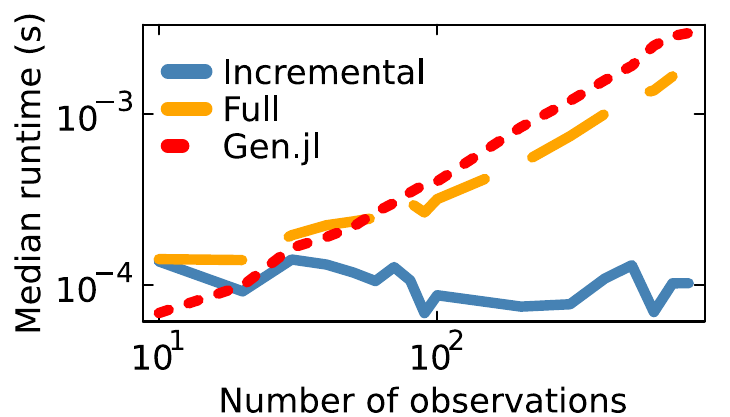}
\subcaption{Singleton Birth/Death}
\end{subfigure}%
\mycaption{Representative log-log scaling plots for individual MCMC moves. Although our approach incurs constant-factor overhead, it is the only one to achieve the correct asymptotic scaling across moves.}
\label{fig:scaling-log-log-individual}
\end{figure}

\begin{figure}[t]
\begin{subfigure}[b]{0.3\textwidth}
\includegraphics[width=\linewidth]{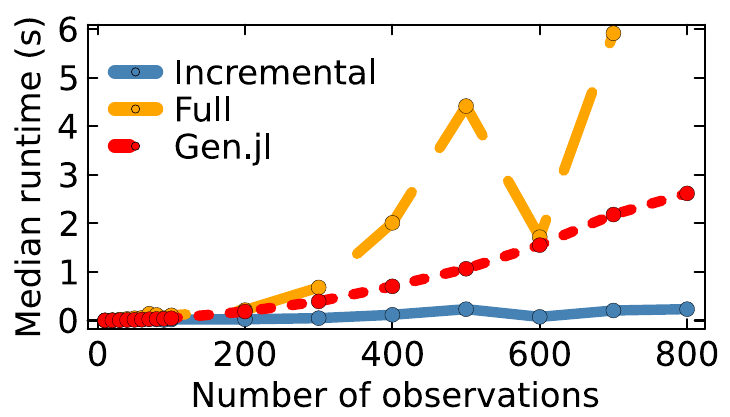}
\subcaption{Mixture of Finite Mixtures}
\end{subfigure}%
\hfill
\begin{subfigure}[b]{0.3\textwidth}
\includegraphics[width=\linewidth]{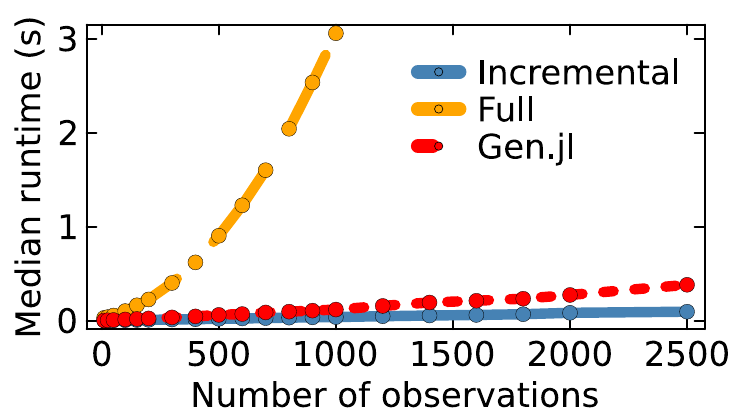}
\subcaption{Hidden Markov Model}
\end{subfigure}%
\hfill
\begin{subfigure}[b]{0.3\textwidth}
\includegraphics[width=\linewidth]{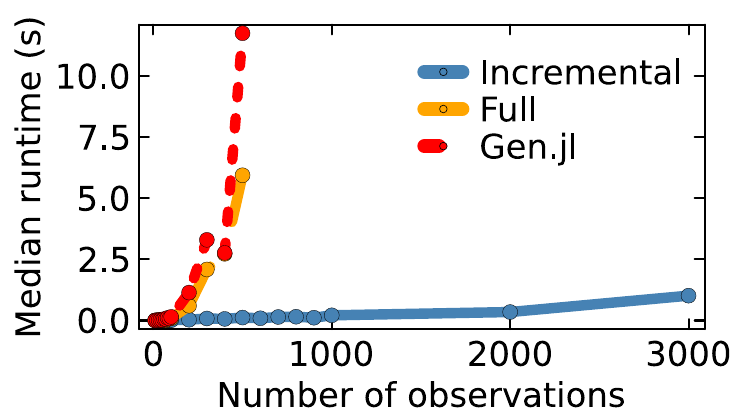}
\subcaption{Dirichlet Process Mixture}
\end{subfigure}
\mycaption{Linear scaling plots for entire MCMC parameter sweeps.}
\label{fig:scaling-aggregate}
\end{figure}

\paragraph{Results.}
Table~\ref{tab:benchmarks} presents absolute and relative runtimes of each MCMC kernel, for each approach to density (re)computation. Timings were recorded for each application of each kernel during 100 steps of MCMC inference on datasets of fixed size (chosen separately per model), and the median across all applications of a given kernel is reported. The log-log scaling plots in \cref{fig:scaling-log-log-individual} show how median times for executing a single kernel scale with dataset size. Each iteration of MCMC inference may execute $O(N)$ kernels (e.g., in a mixture model, to consider reassigning each datapoint's cluster assignment, then updating cluster parameters); \Cref{fig:scaling-aggregate} shows how these aggregate MCMC moves (and thus end-to-end inference) scale with the number of datapoints.

\paragraph{Findings.} The log-log scaling plots in Fig.~\ref{fig:scaling-log-log-individual} are representative of three patterns across our experiments. (a) In many cases, our approach successfully reduces the asymptotic complexity of an update (e.g., from linear-time to constant-time), but so does Gen's state-of-the-art incrementalization, often with better constant factors. (b) In rare cases (see red entries in \cref{tab:benchmarks}), the density must be fully recomputed by all three systems, and our approach to incrementalization adds overhead relative to the non-incremental version. (c) In some cases (see green entries in \cref{tab:benchmarks}), our approach delivers asymptotic savings not captured by Gen's state-of-the-art incrementalization. Our approach is the only one that achieves the desired asymptotic complexity for all updates. \Cref{fig:scaling-aggregate} illustrates how savings on individual updates affect the runtime of an MCMC algorithm that composes these updates into realistic schedules (e.g., ``run one parameter update and $N$ cluster assignment updates per iteration''). The plots again surface several interesting patterns. (a) In some cases, our approach reduces the cost of an update from linear to constant-time, enabling it to be efficiently run on every datapoint at every iteration. The failure of Gen to incrementalize these updates leads to $O(N^2)$ runtime overall. (b) When a single row of an HMM's transition matrix is changed, our approach (but not Gen's) efficiently revisits only the time steps that executed the changed transition. This leads to an $O(K)$ reduction in runtime per step. But this is just a constant-factor speed-up with respect to the sequence length. (c) In a Dirichlet process mixture, the expected number of latent clusters grows logarithmically with $N$, so each Gibbs cluster assignment move requires more density queries as $N$ increases. This means even our approach scales super-linearly, despite reducing the cost of each density evaluation to a constant. This is even more noticeable without incrementalization.

\paragraph{Soundness.} We found that Gen sometimes crashed, silently computed incorrect densities, or resorted to full recomputation, requiring various workarounds.
For example, our initial GMM program sampled {cluster assignments} and observations in separate loops.\footnote{By ``loop'' we mean an application of Gen's \texttt{Map} or \texttt{Unfold} combinator to a kernel implementing the loop body.} We had to merge the loops to achieve $O(1)$ assignment updates in Gen.
Conversely, we had to sample the latent means and covariances in \textit{separate} loops, because Gen did not correctly track changes to tuples.

\paragraph{Implementation Details.}
Our implementation faithfully realizes the algorithmic choices formalized in \cref{sec:calculi,sec:density-transformation,sec:incrementalization}, but differs in implementation details due to our choice of Julia as a host language (chosen for its scientific computing ecosystem and for fair comparison with Gen).
While our formalism uses explicit program transformations, in Julia we use a combination of macros, nonstandard interpretation, and function overloading via multiple dispatch to achieve the same effect.
This is a common way of embedding DSLs in Julia, also used in Gen~\citep{Cusumano-Towner19} and Turing~\citep{GeXG18}.
Our implementation features an abstract \texttt{Updater} class with separate subclasses for each primitive operation, where each subclass stores an operation-specific cache and has an \texttt{update} method corresponding to applying the updater to an input change.
In alignment with our theoretical presentation, the implementation uses immutable, persistent collections (including for its caches).

\section{Related Work and Discussion}
\label{sec:related-work}

\paragraph{Incrementalizing $\lambda$-Calculi.}
Our work is inspired by the incremental $\lambda$-calculus of \citet{CaiGRO14}, which introduces the notion of change types and an incrementalizing program transformation.
Their change type for functions $\sigma \to \tau$ is $\sigma \to \Changed{\sigma} \to \Changed{\tau}$, which is more general than our closure-based approach.
However, their work does not support caching.

\Citet{GiarrussoRS19} add caching of intermediate results in what they call \emph{cache-transfer style}.
They work in the untyped $\lambda$-calculus in order to avoid tracking cache types.
They require programs to be preprocessed into A'-normal form and lambda-lifted. \Citet{MatsudaFWW23} solve the typing issues in the first-order setting using existential types for the caches. \Citet{Morihata20} instead incrementalizes functions $\sigma \to \tau$ to \emph{caching derivative-associated} functions $\sigma \to \tau \times (\Delta \sigma \to \Delta \tau)$, where intermediate results are stored in the inner closure's environment, so the caches do not need to be represented in the types.
He also mentions that a series of modifications can be supported via ``recursively-caching'' functions, but does not elaborate.
This was inspiration for our updater types.

\paragraph{Dynamic Dependency Tracking.}
An influential \emph{dynamic} incrementalization technique is self-adjusting computation \cite{Acar05}, where the initial run produces a computation trace memoizing intermediate results.
Changed inputs require only the recomputation of affected parts of the trace (see \cite{Acar09}).
Adapton \cite{HammerKHF14} instead constructs a demand-driven dependency graph at runtime.
Follow-up work \cite{HammerDHLFHH15} also uses first-class names, recognizing them as a ``critical linguistic feature for efficient incremental computation''.
Whereas they use unique names to identify subcomputations, we use them to identify domain entities in probabilistic models.

\paragraph{Incrementalization in Probabilistic Programming.}
Several program transformations aim to speed up single-site Metropolis-Hastings specifically.
Shred \cite{YangHG14} optimizes changes to random choices that preserve the control flow by compiling to a straight-line program and only recomputing its transitive dependencies.
C3 \cite{RitchieSG16} avoids recomputation by combining continuation-passing style with call-site caching.
Similarly, \citet{CastellanP19} describe a translation from a probabilistic program to an event structure, tracking dependencies and avoiding unnecessary sampling.

For more complex inference algorithms, many systems analyze probabilistic programs using graphical models because changing a variable only affects its \emph{Markov blanket}.
An early instance of this is BLOG \cite{MilchR06}, a declarative PPL for open-universe models, whose Metropolis-Hastings inference maintains a dynamic dependency graph to avoid recomputing factors in the acceptance ratio.
As changing a variable can change the dependency graph, proposals require recomputing the dependency structure.
The Bean Machine \cite{TehraniALSNTTML20} extends this with programmable proposal distributions that can be composed and blocked across variables.
Venture \cite{MansinghkaSP14,MansinghkaSHRCR18} was perhaps the first to introduce dynamic dependency tracking for a higher-order generative language with programmable inference, enabling optimal asymptotic scaling for a broad range of inference algorithms.
Its dependency tracking also handles aliasing from memoization and was the first to be extended to sublinear-time inference updates, by subsampling dependencies.
\citet{BockC25} present a program analysis to factorize the density function in the presence of loops, which can speed up single-site MH, black-box variational inference, and sequential Monte Carlo.

Gen \cite{Cusumano-Towner19} allows users to customize incrementalization via an \emph{update} function in its generative function interface.
By default, its static modeling language (for straight-line probabilistic programs) automates incremental updates via dependency tracking and partial evaluation.
It also has undocumented, experimental support for incrementalization based on ``diffs''~\citep{Cusumano-Towner20}, disabled by default due to bugs and limited support.
Our work drew inspiration from Gen, and could inform future extensions to Gen, e.g., to handle features such as names, tuples, and higher-order functions.

\citet{LimLLRY26} also aim to improve the efficiency of density computations in probabilistic programs, but via auto-vectorization based on tensor operations rather than incremental computation; notably, incrementalization gives asymptotic speedups even on a single CPU, without parallel hardware.
Instead of incremental density evaluation, \citet{Cusumano-TownerBGVM18} tackle the related problem of incrementally generating \emph{samples} for changes to the input \emph{program} by translating their traces.

\paragraph{Density Semantics.}
\citet{LeeYRY20} define both measure and density semantics for an imperative, first-order probabilistic language, but their reference measure only supports densities when all primitive distributions are continuous, whereas our traces may contain heterogeneous types.
\citet{TassarottiT23} present a verified compilation pipeline from a simplified version of Stan to optimized density code in C, but their language and approach (operational semantics, imperative) are very different from ours.
\citet{BeckerHMWCSRSLRM26} use a similar semantics to ours, but their language does not feature name generation or unordered collections---key contributions of ours.

\paragraph{Nonparametric Models.}
Existing systems struggle with exchangeable sequences and nonparametric models.
Gen \cite{Cusumano-Towner19} can express nonparametric models, but its traces introduce artificial orderings on exchangeable sequences, making incrementalization less effective because insertions and deletions shift indices.
\citet{MatheosLGRCM21} automate involutive MCMC (which can express many MCMC algorithms) for open-universe models, but have to deal with complicated and error-prone bookkeeping of indices into variable-length sequences.
By contrast, our approach of unordered collections with name-based access offers stable identifiers for elements of exchangeable sequences, simplifying the implementation and improving the effectiveness of incrementalization.
\paragraph{Limitations.} Long, straight-line programs made of many individually inexpensive operations---e.g., many large Bayes nets---do not benefit from incrementalization in our approach, which must still execute an incrementalized version of each line of such a program.
As it stands, our approach is also limited to loops over finite data structures.
Extending our implementation to handle general $\keyword{while}$ loops would be straightforward, requiring only small changes to our existing density and incrementalization transformations.
General recursion would be more challenging, as traces of recursive probabilistic programs would need to be assigned \emph{recursive trace types}, for which we would need to define new change representations and incrementalization rules.
Whether via recursion or while loops, the possibility of non-termination would complicate our semantics and correctness proofs (requiring, e.g., the use of quasi-Borel \textit{predomains} \citep{VakarKS19}).

\paragraph{Extensions.} We have focused on MCMC, but incremental computation could also benefit other Monte Carlo inference algorithms (see, e.g., the SMC pseudocode in \cref{fig:marquee:inference-algorithms}).
For {variational inference}, the fit is less natural: traces are independently sampled rather than incrementally modified. Our pipeline works by generating \textit{deterministic} density programs and incrementalizing them. However, it is common to \textit{stochastically estimate} densities that would be intractable to compute exactly~\citep{AndrieuR09}. It would be interesting to develop versions of incrementalization that work for stochastic computations, perhaps building on ideas from~\citet{Cusumano-TownerBGVM18}.

\clearpage

\section*{Data Availability Statement}

The artifact for this paper (consisting of the Julia implementation and benchmarks) is archived on Zenodo \citep{ZaiserCRML26}.

\begin{acks}
We would like to thank McCoy Becker, Marco Cusumano-Towner, Mathieu Huot, George Matheos, and Tan Zhi-Xuan for helpful conversations at various stages of this project.
We are also grateful to Cameron Freer and the anonymous reviewers for their helpful feedback, which improved the presentation of the paper.
\end{acks}

\bibliographystyle{ACM-Reference-Format}
\bibliography{literature.bib}

\par\bigskip\noindent{\small\normalfont Received 2025-11-14; accepted 2026-04-03\par}

\clearpage

\makeatletter
\edef\@mainBodyTotPages{\arabic{TotPages}}%
\AtEndDocument{%
  \immediate\write\@mainaux{%
    \string\newlabel{TotPages}{%
      {\@mainBodyTotPages}{\@mainBodyTotPages}{}{page.\@mainBodyTotPages}{}%
    }%
  }%
}%
\makeatother

\appendix

\newcommand{\catset}{\mathsf{Set}}
\newcommand{\catmeas}{\mathsf{Meas}}
\newcommand{\homset}[3]{#1(#2,#3)}
\newcommand{\A}{\mathcal{A}}
\newcommand{\J}{\mathcal{J}}
\newcommand{\mProb}{\mathcal{T}}
\newcommand{\mMeas}{\mathcal{P}}
\newcommand{\qbs}{\mathsf{Qbs}}
\newcommand{\pshqbs}{\texorpdfstring{\ensuremath{\nu\qbs}}{νQbs}}
\newcommand{\curry}{\mathsf{curry}}
\newcommand{\counting}{\mathsf{counting}}
\newcommand{\score}{\mathsf{score}}

\newcommand{\mur}{\mu_{\Omega}}

\newcommand{\llbr}[1]{\llbracket #1 \rrbracket}
\newcommand{\Qsem}[1]{\sem{#1}}
\newcommand{\Ssem}[1]{\sem{#1}}
\newcommand{\Csem}[1]{\sem{#1}}

\tableofcontents

\clearpage
\allowdisplaybreaks 

\section{Full syntax, typing rules, and semantics}

The list of constants in our calculi are given in \cref{fig:constants-full}, the syntactic sugar is given in \cref{fig:sugar}, the restrictions on ground types are given in \cref{fig:ground-types}, the typing rules are given in \cref{fig:type-system-full}, and the type and term semantics for $\inclang$ and $\corelang$ are given in \cref{fig:type-semantics-full,fig:term-semantics-full} respectively.
The semantics of $\genlang$ is presented in its own subsection (\cref{sec:gen-semantics}).
Note that we regard $\Dist[\kappa]{\tau}$ as syntactic sugar in the appendix.

\begin{figure}\footnotesize
  \begin{align*}
    &\text{\emph{Distributions} (only \genlang):} & \\
    &\quad \const{bernoulli} : \Real \to \Dist[\Real]{\Bool} &&\text{Bernoulli distribution} \\
    &\quad\const{normal}_n : \Real^n \times \Real^{n \times n} \to \Dist[{\Real^n\times\Real^{n\times n}}]{\Real^n} &&\text{$n$-dim. Gaussian with parameters } \mu, \Sigma \\
    &\quad\const{beta} : \Real \times \Real \to \Dist[{\Real \times \Real}]{\Real} &&\text{Beta distribution} \\
    &\quad\const{poisson} : \Real \to \Dist[\Real]{\Nat} &&\text{Poisson distribution} \\
    &\quad\const{geometric} : \Real \to \Dist[\Real]{\Nat} &&\text{Geometric distribution} \\
    &\quad\const{dirac}_{\sigma} : \sigma \to \Dist[\sigma]{\sigma} &&\text{Dirac distribution for } \sigma \in \{ \Nat, \NameSet \} \\
    &\quad\const{unif} : \NameSet \to \Dist[\NameSet]{\Name} &&\text{Uniform distribution on a finite set} \\
    &\quad\const{inverseWishart}_n : \Real \times \Real^{n \times n} \to \Dist[{\Real \times \Real^{n \times n}}]{\Real^{n \times n}} &&\text{Inverse Wishart distribution} \\
    &\quad\const{dirichlet} : \NameSet \times \Real \to \Dist[{\NameSet \times \Real}]{\NameMap{\Real}} &&\text{Dirichlet distribution} \\
    &\quad\const{multinomial} : \NameMap{\Real} \times \Nat \to \Dist[{\NameMap{\Real} \times \Nat}]{\List{\Name}} &&\text{Multinomial distribution} \\
    &\quad\const{fresh} : \Dist[\kappa]{\Nat} \to \Dist[\kappa]{\NameSet} &&\text{random-size set of fresh names} \\[4pt]
    &\text{\emph{List operations} (all calculi):} & \\
    &\quad\const{nil}_{\tau} : \List{\tau} &&\text{Empty list} \\
    &\quad\const{cons}_{\tau} : \tau \times \List{\tau} \to \List{\tau} &&\text{Construct a list} \\
    &\quad\const{get}_{\tau} : \Nat \times \List{\tau} \to \tau &&\text{Value at given index} \\
    &\quad\const{length}_{\tau} : \List{\tau} \to \Nat &&\text{Length of a list} \\
    &\quad\const{unzip}_{\tau, \tau'} : \List{(\tau \times \tau')} \to \List{\tau} \times \List{\tau'} &&\text{Splits list of pairs into pair of lists} \\
    &\quad\const{range} : \Nat \to \List{\Nat} &&\text{$[1, 2, \dots, n]$ for input } n \\[4pt]
    &\text{\emph{Map operations} (all calculi):} & \\
    &\quad\const{empty}_{\tau} : \NameMap{\tau} &&\text{Empty map} \\
    &\quad\const{insert}_{\tau} : \Name \times \tau \times \NameMap{\tau} \to \NameMap{\tau} &&\text{Insert into a map} \\
    &\quad\const{remove}_{\tau} : \Name \times \NameMap{\tau} \to \NameMap{\tau} &&\text{Remove from a map} \\
    &\quad\const{get}_{\tau} : \Name \times \NameMap{\tau} \to \tau &&\text{Value at given key} \\[4pt]
    &\text{\emph{Arithmetic} (all calculi):} & \\
    &\quad\const{-}_{\tau} : \tau \to \tau &&\text{for } \tau \in \{\Nat, \Real\} \\
    &\quad\const{+}_{\tau}, \const{-}_{\tau} : \tau \times \tau \to \tau &&\text{for } \tau \in \{\Nat, \Real\} \\
    &\quad\const{\cdot}, \const{/} : \Real \times \Real \to \Real &&\text{multiplication and division} \\
    &\quad\const{product} : \List{\Real} \to \Real &&\text{product of list elements} \\
    &\quad\const{==}_{\sigma} : \sigma \times \sigma \to \Bool &&\text{for ground types } \sigma \\[4pt]
    &\text{\emph{Other} (all):} & \\
    &\quad\const{ite}_{\sigma} : \Bool \times \sigma \times \sigma \to \sigma &&\text{if-then-else at ground types} \\
    &\quad\const{default}_{\tau} : \tau &&\text{Default value for out-of-bounds access} \\[4pt]
    &\text{\emph{Change-related primitives} (only \corelang; see \cref{fig:diff-helpers}):} & \\
    &\quad\same[\tau] : \Changed{\tau} &&\text{trivial change (only for some $\tau$)} \\
    &\quad\sameAs[\tau] : \tau \to \Changed{\tau} &&\text{trivial change at a certain value (for ground types)} \\
    &\quad\apply[\tau] : \tau \times \Changed{\tau} \to \tau &&\text{apply a change to a value (for ground types $\tau$)}
  \end{align*}
  \caption{List of constant symbols $(c : \typeOf{c}) \in \Constants$.}
  \label{fig:constants-full}
\end{figure}

\begin{figure}\footnotesize
  \begin{align*}
    \Dist[\sigma]{\tau} &\quad\leadsto\quad \ProbTy[\sigma]{\tau}{\Unit} & \text{(in \genlang)} \\
    \sigma \to \tau &\quad\leadsto\quad \sigma \clos[\Unit] \tau & \text{(in \genlang and \inclang)} \\
    \forUsing{x}{\mvar{xs}: \NameSet}{\dots}{t} &\quad\leadsto\quad \forSetUsing{x}{\const{dirac} \; \mvar{xs}}{\dots}{t} & \text{(in \genlang)} \\
    \forUsing{x}{\mvar{xs}: \List{\tau}}{\dots}{t} &\quad\leadsto\quad \forRangeWithUsing{i}{\const{dirac}(\const{length} \; \mvar{xs})}{z}{\unit \\
    &\qquad \qquad }{x := \mvar{xs}[i], \dots}{y \gets t; \return{(y, \unit)}} & \text{(in \genlang)} \\
    \forWith{x}{\mvar{xs}}{z}{s}{t} &\quad\leadsto\quad \forWithUsing{x}{\mvar{xs}}{z}{s}{}{t} \\
    \forUsing{x}{\mvar{xs}: \List{\tau}}{\dots}{t} &\quad\leadsto\quad \forWithUsing{x}{\mvar{xs}}{z}{\unit \\
    &\qquad \qquad }{\dots}{y := t; (y, \unit)} & \text{(in \inclang)} \\
    \forIn{x}{\mvar{xs}}{t} &\quad\leadsto\quad \forUsing{x}{\mvar{xs}}{}{t} \\
    \NameSet &\quad\leadsto\quad \NameMap{\Unit} \\
    x \sim s; t &\quad\leadsto\quad x \gets \sample{s} \at x ; t \\
    s[t] &\quad\leadsto\quad \const{get}(t, s) \qquad \text{(see \cref{fig:constants-full})} \\
    (x, y) := s; t &\quad\leadsto\quad z := s; x := z.1; y := z.2; t \quad \text{ where $z$ is fresh} \\
    (x, y) \gets s; t &\quad\leadsto\quad z \gets s; x := z.1; y := z.2; t \quad \text{ where $z$ is fresh} \\
    [t_1, \dots, t_n] &\quad\leadsto\quad \const{cons}(t_1, \dots \const{cons}(t_n, \const{nil})\dots) \\
    \{k_1 \mapsto v_1, \dots, k_n \mapsto v_n\} &\quad\leadsto\quad \const{insert}(k_1, v_1, \dots \const{insert}(k_n, v_n, \const{empty})\dots)
  \end{align*}
  \caption{Syntactic sugar}
  \label{fig:sugar}
\end{figure}

\begin{figure}
  \centering
  \begin{mathpar}

    \inferrule
      { }
      {\isground{\Bool}}

    \inferrule
      { }
      {\isground{\Nat}}

    \inferrule
      { }
      {\isground{\Real}}

    \inferrule
      { }
      {\isground{\Name}}

    \inferrule
      {\isground{\sigma} \\ \isground{\sigma'}}
      {\isground{\sigma \times \sigma'}}

    \inferrule
      {\isground{\sigma}}
      {\isground{\List{\sigma}}}

    \inferrule
      {\isground{\sigma}}
      {\isground{\NameMap{\sigma}}}

    \genhighlight{\inferrule
      {\isground{\sigma} \\ \istype{\kappa}}
      {\istype{\Dist[\kappa]{\sigma}}}}

    \genhighlight{
      \inferrule
        {\isground{\sigma} \\ \istype{\tau} \\ \istype{\kappa}}
        {\istype{\ProbTy[\kappa]{\sigma}{\tau}}}
    }

    \mixhighlight{\inferrule
      {\istype{\sigma} \\ \istype{\tau} \\ \istype{\kappa}}
      {\istype{\sigma \clos[\kappa] \tau}}
    }

    \mixhighlight{\inferrule
      {L \subseteq \Labels \text{ finite} \\ \isground{\sigma_l} \text{ for all $l \in L$}}
      {\isground{\Record{l : \sigma_l \mid l \in L}}}
    }
  \end{mathpar}
  \caption{Restrictions on ground types.}
  \label{fig:ground-types}
\end{figure}

\begin{figure}\footnotesize
  \centering
  \begin{mathpar}
    \inferrule{ }{\toType{\bullet} = \Unit}

    \inferrule
      {\toType{\Gamma} = \tau}
      {\toType{\Gamma, x : \tau'} = \tau \times \tau'}

    \inferrule
      { }
      {\tyjudg{\Gamma, x : \tau, \Gamma'}{x}{\tau}}
      (x \notin \Gamma')

    \inferrule
      {\tyjudg{\Gamma}{s}{\tau \times \tau'}}
      {\tyjudg{\Gamma}{s.1}{\tau}}

    \inferrule
      {\tyjudg{\Gamma}{s}{\tau \times \tau'}}
      {\tyjudg{\Gamma}{s.2}{\tau'}}

    \inferrule
      { }
      {\tyjudg{\Gamma}{c}{\typeOf{c}}}
      (c \in \Constants)

    \inferrule
      {\tyjudg{\Gamma}{s}{\tau} \\ \tyjudg{\Gamma, x: \tau}{t}{\tau'}}
      {\tyjudg{\Gamma}{x := s; t}{\tau'}}

    \inferrule
      {\tyjudg{\Gamma}{s}{\tau} \\ \tyjudg{\Gamma}{t}{\tau'}}
      {\tyjudg{\Gamma}{(s, t)}{\tau \times \tau'}}

    \\

    \genhighlight{\inferrule
      {\tyjudg{\Gamma, x : \tau}{t} {\tau'}}
      {\tyjudg{\Gamma}{\lam{x} t}{\tau \clos[\toType{\Gamma|_{\freeVars{\lam{x} t}}}] \tau'}}}

    \genhighlight{\inferrule
      {\tyjudg{\Gamma}{s}{\tau \clos[\sigma] \tau'} \\ \tyjudg{\Gamma}{t}{\tau}}
      {\tyjudg{\Gamma}{s \; t}{\tau'}}}

    \genhighlight{\inferrule
      {\tyjudg{\Gamma}{t}{\tau}}
      {\tyjudg{\Gamma}{\return{t}}{\ProbTy[\toType{\Gamma|_{\freeVars{t}}}]{\Unit}{\tau}}}}

    \genhighlight{\inferrule
      {\tyjudg{\Gamma}{t}{\Dist[\kappa]{\sigma}}}
      {\tyjudg{\Gamma}{\sample{t}}{\ProbTy[\kappa]{\sigma}{\sigma}}}}

      \genhighlight{\inferrule
      {\tyjudg{\Gamma}{t}{\ProbTy[\kappa]{\sigma}{\tau}}}
      {\tyjudg{\Gamma}{t \at \ell}{\ProbTy[\kappa]{\Record{\ell: \sigma}}{\tau}}}}

    \genhighlight{\inferrule
      {\tyjudg{\Gamma}{s}{\ProbTy[\kappa]{\Record{\ell_1: \sigma_1, \dots, \ell_m: \sigma_m}}{\tau}} \\
      \tyjudg{\Gamma, x : \tau}{t}{\ProbTy[\kappa']{\Record{\ell_1': \sigma_1', \dots, \ell_n': \sigma_n'}}{\tau'}} \\
      \{ \ell_1, \dots, \ell_m \} \cap \{ \ell_1', \dots, \ell_n' \} = \emptyset}
      {\tyjudg{\Gamma}{x \gets s; t}{\ProbTy[\toType{\Gamma|_{\freeVars{x \gets s; t}}}]{\Record{\ell_1: \sigma_1, \dots, \ell_m: \sigma_m, \ell_1': \sigma_1', \dots, \ell_n': \sigma_n'}}{\tau'}}}}

    \genhighlight{\inferrule
      {\tyjudg{\Gamma}{w}{\Dist[\kappa]{\Nat}} \\
      \tyjudg{\Gamma}{s}{\rho} \\
      \tyjudg{\Gamma, x: \Nat, z: \rho}{y_i[u_i]}{\theta_i} \\
      \tyjudg{\Gamma, x: \Nat, z: \rho, v_1: \theta_1, \dots, v_n: \theta_n}{t}{\ProbTy[\kappa']{\sigma}{(\tau \times \rho)}}}
      {\tyjudg{\Gamma}{\forRangeWithUsing{x}{w}{z}{s}{v_1 := y_1[u_1], \dots, v_n := y_n[u_n]}{t}}{\ProbTy[\toType{\Gamma|_{\freeVars{...}}}]{(\List \sigma)}{(\List{\tau})}}}}

    \genhighlight{\inferrule
      {\tyjudg{\Gamma}{\mvar{xs}}{\Dist[\kappa]{\NameSet}} \\
      \tyjudg{\Gamma, x: \Name}{y_i[u_i]}{\theta_i} \\
      \tyjudg{\Gamma, x: \Name, v_1: \theta_1, \dots, v_n: \theta_n}{t}{\ProbTy[\kappa']{\sigma}{\tau'}}}
      {\tyjudg{\Gamma}{\forSetUsing{x}{\mvar{xs}}{v_1 := y_1[u_1], \dots, v_n := y_n[u_n]}{t}}{\ProbTy[\toType{\Gamma|_{\freeVars{...}}}]{(\NameMap{\sigma})}{(\NameMap{\tau'})}}}}

    \inchighlight{\inferrule
      {\tyjudg{\Gamma, x: \tau}{t}{\tau'}}
      {\tyjudg{\Gamma}{\lam{x} t}{\tau \clos[\toType{\Gamma|_{\freeVars{\lam{x} t}}}] \tau'}}}

    \inchighlight{\inferrule
      {\tyjudg{\Gamma}{s}{\tau \clos[\sigma] \tau'} \\ \tyjudg{\Gamma}{t}{\tau}}
      {\tyjudg{\Gamma}{s \; t}{\tau'}}}

    \inchighlight{\inferrule
      {\tyjudg{\Gamma}{\mvar{xs}}{\List \tau} \\
      \tyjudg{\Gamma}{s}{\rho} \\
      \tyjudg{\Gamma, x: \tau, z: \rho}{y_i[u_i]}{\theta_i} \\
      \tyjudg{\Gamma, x: \tau, z: \rho, v_1: \theta_1, \dots, v_n: \theta_n}{t}{\tau' \times \rho}}
      {\tyjudg{\Gamma}{\forWithUsing{x}{\mvar{xs}}{z}{s}{v_1 := y_1[u_1], \dots, v_n := y_n[u_n]}{t}}{\List{\tau'}}}}

    \inchighlight{\inferrule
      {\tyjudg{\Gamma}{\mvar{xs}}{\NameSet} \\
      \tyjudg{\Gamma, x: \Name}{y_i[u_i]}{\theta_i} \\
      \tyjudg{\Gamma, x: \Name, v_1: \theta_1, \dots, v_n: \theta_n}{t}{\tau'}}
      {\tyjudg{\Gamma}{\forUsing{x}{\mvar{xs}}{v_1 := y_1[u_1], \dots, v_n := y_n[u_n]}{t}}{\NameMap{\tau'}}}}

    \mixhighlight{\inferrule
      {\tyjudg{\Gamma}{t_i}{\tau_i} \text{ for } i = 1, \dots, n}
      {\tyjudg{\Gamma}{\recordLit{\ell_1: t_1, \dots, \ell_n: t_n}}{\Record{\ell_1: \tau_1, \dots, \ell_n: \tau_n}}}}

    \mixhighlight{\inferrule
      {\tyjudg{\Gamma}{t}{\Record{\ell_1: \tau_1, \dots, \ell_n: \tau_n}}}
      {\tyjudg{\Gamma}{t.\ell_i}{\tau_i}}}

    \mixhighlight{\inferrule
      {\tyjudg{\Gamma}{t}{\Record{\ell_1: \tau_1, \dots, \ell_n: \tau_n}} \\
      \{i_1 < \dots < i_k\} \subseteq \{1, \dots, n\}}
      {\tyjudg{\Gamma}{\restrict{t}{\{\ell_{i_1}, \dots, \ell_{i_k}\}}}{\Record{\ell_{i_1}: \tau_{i_1}, \dots, \ell_{i_k}: \tau_{i_k}}}}}

    \corehighlight{\inferrule
      {\tyjudg{\Gamma, x : \tau}{t}{\tau'}}
      {\tyjudg{\Gamma}{\lam{x} t}{\tau \to \tau'}}}

    \corehighlight{\inferrule
      {\tyjudg{\Gamma}{s}{\tau \to \tau'} \\ \tyjudg{\Gamma}{t}{\tau}}
      {\tyjudg{\Gamma}{s \; t}{\tau'}}}

    \corehighlight{\inferrule
      {\tyjudg{\Gamma}{s}{\sigma} \\ \tyjudg{\Gamma}{t}{\tau \times \sigma \to \tau' \times \sigma}}
      {\tyjudg{\Gamma}{\mkUpd s t}{\Upd{\tau}{\tau'}}}}

    \corehighlight{\inferrule
      {\tyjudg{\Gamma}{s}{\Upd{\tau}{\tau'}} \\ \tyjudg{\Gamma}{t}{\tau}}
      {\tyjudg{\Gamma}{\applyUpd{s}{t}}{\tau' \times \Upd{\tau}{\tau'}}}}

    \corehighlight{\inferrule
      {\tyjudg{\Gamma}{t_j}{\tau_{i_j}} \text{ for } j = 1, \dots, k}
      {\tyjudg{\Gamma}{\recordLit{\ell_{i_1}: t_1, \dots, \ell_{i_k}: t_k}}{\SubRecord{\ell_1: \tau_1, \dots, \ell_n: \tau_n}}}
      (i_j\text{'s are distinct indices in } \{1, \dots, n\})
    }

    \corehighlight{\inferrule
      {\tyjudg{\Gamma}{t}{\SubRecord{\ell_1: \tau_1, \dots, \ell_n: \tau_n}}}
    {\tyjudg{\Gamma}{t.\ell_i}{\tau_i}}}

    \corehighlight{\inferrule
      {\tyjudg{\Gamma}{t}{\SubRecord{\ell_1: \tau_1, \dots, \ell_n: \tau_n}}}
      {\tyjudg{\Gamma}{\has[\ell_i]{t}}{\Bool}}}

    \corehighlight{\inferrule
      {\tyjudg{\Gamma}{t}{\SubRecord{\ell_1: \tau_1, \dots, \ell_n: \tau_n}} \\
      \{i_1 < \dots < i_k\} \subseteq \{1, \dots, n\}}
      {\tyjudg{\Gamma}{\restrict{t}{\{\ell_{i_1}, \dots, \ell_{i_k}\}}}{\SubRecord{\ell_{i_1}: \tau_{i_1}, \dots, \ell_{i_k}: \tau_{i_k}}}}}
  \end{mathpar}
  \caption{Typing Rules. We write $\toType{\Gamma}$ for the product of the types appearing in $\Gamma$; $\Gamma|_S$ for the restriction of $\Gamma$ to the variables in $S$; and $\freeVars{t}$ for the set of free variables in $t$.}
  \label{fig:type-system-full}
\end{figure}

\begin{figure}\footnotesize
\begin{align*}
  \sem{\Bool} &= \{ 0, 1 \} \\
  \sem{\Nat} &= \NN \\
  \sem{\Real} &= \RR \\
  \sem{\Name} &= [0, 1] \\
  \sem{\sigma \times \tau} &= \sem{\sigma} \times \sem{\tau} \\
  \sem{\List{\sigma}} &= \bigcup_{n \in \NN} \sem{\sigma}^n \\
  \sem{\NameMap{\tau}} &= \{ f : A \to \sem{\tau} \mid A \subseteq \sem{\Name} \text{ finite} \} \\
  \sem{\Record{\ell_1: \tau_1, \dots, \ell_n: \tau_n}} &= \{ r : \{ \ell_1, \dots, \ell_n \} \to \bigcup_{i=1}^n \sem{\tau_i} \mid r(\ell_i) \in \sem{\tau_i} \} \\
  \sem{\SubRecord{\ell_1: \tau_1, \dots, \ell_n: \tau_n}} &= \{ r : S \to \bigcup_{i=1}^n \sem{\tau_i} \mid S \subseteq \{\ell_1, \dots, \ell_n\}, r(\ell_i) \in \sem{\tau_i} \} \\
  \sem{\tau \clos[\kappa] \tau'} &= \sem{\kappa} \times (\sem{\kappa} \times \sem{\tau} \To \sem{\tau'}) \\
  \sem{\tau \to \tau'} &= \sem{\tau} \To \sem{\tau'} \\
  \sem{\Upd{\tau}{\tau'}} &= \sem{\tau}^+ \To \sem{\tau'} \qquad \text{ where } S^+ = \bigcup_{n = 1}^\infty S^n \\
  \sem{\Gamma} &= \{ \gamma : \{ x_1, \dots, x_n \} \to \bigcup_{i=1}^n \sem{\tau_i} \mid \gamma(x_i) \in \sem{\tau_i} \} \text{ for contexts } \Gamma = x_1: \tau_1, \dots, x_n: \tau_n
\end{align*}
\caption{Type semantics for $\inclang$ and $\corelang$.}
\label{fig:type-semantics-full}
\end{figure}

\begin{figure}
\begin{align*}
  \sem{x}(\gamma) &= \gamma[x] \\
  \sem{c}(\gamma) &= \text{semantics of primitive } c \\
  \sem{\falseLit}(\gamma) &= 0 \\
  \sem{\trueLit}(\gamma) &= 1 \\
  \sem{x := s; t}(\gamma) &= \sem{t}(\gamma[x \mapsto \sem{s}(\gamma)]) \\
  \sem{\lam{x} t}(\gamma) &= \big((\gamma[x_1], \dots, \gamma[x_n]), \\
  &\qquad ((v_1, \dots, v_n), v) \mapsto \sem{t}(\gamma[x_1 \mapsto v_1, \dots, x_n \mapsto v_n, x \mapsto v])\big) \\
  &\qquad\text{ where } \Gamma|_{\freeVars{\lam{x} t}} = x_1: \tau_1, \dots, x_n: \tau_n \text{ in } \inclang \\
  \sem{\lam{x} t}(\gamma) &= v \mapsto \sem{t}(\gamma[x \mapsto v]) \text{ in } \corelang \\
  \sem{s \; t}(\gamma) &= f(\mvar{env}, \sem{t}(\gamma)) \text{ where } \sem{s}(\gamma) = (\mvar{env}, f) \text{ in } \inclang \\
  \sem{s \; t}(\gamma) &= \sem{s}(\gamma)(\sem{t}(\gamma)) \text{ in } \corelang \\
  \sem{(s, t)}(\gamma) &= (\sem{s}(\gamma), \sem{t}(\gamma)) \\
  \sem{t.1}(\gamma) &= \pi_1(\sem{t}(\gamma)) \\
  \sem{t.2}(\gamma) &= \pi_2(\sem{t}(\gamma)) \\
  \sem{\ite{s}{t}{t'}}(\gamma) &= \begin{cases}
    \sem{t}(\gamma) & \text{if } \sem{s}(\gamma) = \sem{\trueLit} \\
    \sem{t'}(\gamma) & \text{if } \sem{s}(\gamma) = \sem{\falseLit}
\end{cases} \\
  \sem{t.\ell}(\gamma) &= \begin{cases} \sem{t}(\gamma)(\ell) & \text{if } \ell \in \domain(\sem{t}(\gamma)) \\ \default[\tau] & \text{otherwise, where $t.\ell : \tau$} \end{cases} \\
  \sem{\restrict{t}{\{\ell_1, \dots, \ell_k\}}}(\gamma) &= \ell_i \mapsto \sem{t}(\gamma)(\ell_i) \text{ for } i = 1, \dots, k \\
  \sem{\has[\ell]{t}}(\gamma) &= \begin{cases}
    \sem{\trueLit} & \text{if } \ell \in \domain(\sem{t}(\gamma)) \\
    \sem{\falseLit} & \text{otherwise}
  \end{cases} \\
  \sem{\mkUpd{s}{t}}(\gamma) &= \mkUpdSem[\sem{\sigma},\sem{\tau},\sem{\tau'}](\sem{s}(\gamma), \sem{t}(\gamma)) \text{ where }  \\
  &\qquad \mkUpdSem[S,T,T'] \in S \times (T \times S \To T' \times S) \To T^+ \To T' \\
  &\qquad \mkUpdSem[S,T,T'](c, f)((x_1, \dots, x_n) \in T^n) = y_n \text{ for all } n \in \NN \\
  &\qquad \quad \text{where } c_1 = c \text{ and } (y_i, c_{i+1}) = f(x_i, c_i) \text{ for } i \in \NN \\
  \sem{\applyUpd{s}{t}}(\gamma) &= \stepUpd[\sem{\tau}, \sem{\tau'}](\sem{s}(\gamma), \sem{t}(\gamma)) \text{ where } s : \Upd{\tau}{\tau'}, \\
  &\qquad \stepUpd[T, T'] \in (T^+ \To T') \times T \To T' \times (T^+ \To T') \\
  &\qquad \stepUpd[T, T'](u, x) = (u(x \in T^1), (\mvar{xs} \in T^n) \mapsto u((x, \mvar{xs}) \in T^{n+1})) \\
  \sem{\scriptsize\begin{aligned}&\forWithUsing{x}{\mvar{xs}}{z}{s \\&}{v_1 := y_1[u_1], \\&\quad \dots, v_n := y_n[u_n]}{t}\end{aligned}}(\gamma) &= (t_1, \dots, t_N) \text{ where }  \\
  &\qquad \sem{\mvar{xs}}(\gamma) = (a_1, \dots, a_N), \; z_0 = \sem{s}(\gamma), \text{ and for } j = 1, \dots, N, \\
  &\qquad (t_j, z_j) = \sem{v_1 := y_1[u_1]; \dots; v_n := y_n[u_n]; t}(\gamma[x \mapsto a_j, z \mapsto z_{j-1}]) \\
  \sem{\scriptsize\begin{aligned}&\forUsing{x}{\mvar{xs}}{v_1 := y_1[u_1], \\& \quad \dots, v_n := y_n[u_n]}{t}\end{aligned}}(\gamma) &= \{ a \mapsto t_a \mid a \in \domain(\sem{\mvar{xs}}(\gamma)) \} \text{ where, for each } a \in \domain(\sem{\mvar{xs}}(\gamma)) \text{,} \\
  &\qquad t_a = \sem{v_1 := y_1[u_1]; \dots; v_n := y_n[u_n]; t}(\gamma[x \mapsto a])
\end{align*}
\caption{Term semantics for $\inclang$ and $\corelang$.}
\label{fig:term-semantics-full}
\end{figure}

\clearpage 

\subsection{Denotational Semantics for \genlang{} in \pshqbs}
\label{sec:gen-semantics}

The main paper gives a semantics of $\genlang$ in the category of quasi-Borel spaces. Some additional structure is useful for reasoning about names. We develop this extended semantic setting here.

\paragraph{Intuition.} The key idea is to define, for each type $\tau$ in $\genlang$, a family of subspaces $\sem{\tau}_U \subseteq \sem{\tau}$, indexed by countable sets $U \subseteq [0,1]$.\footnote{We also allow $U=[0,1]$, with $\sem{\tau}_U = \sem{\tau}$. When $U \subseteq V$, we ensure $\sem{\tau}_U \subseteq \sem{\tau}_V$, and for all $U,V$, $\sem{\tau}_{U \cap V} = \sem{\tau}_U \cap \sem{\tau}_V$. Such indexed families form a cartesian-closed category $\pshqbs$, with morphisms $f : X \to Y$ satisfying $f(X_U)\subseteq Y_U$; we develop this formally below.} Intuitively, $\sem{\tau}_U$ contains only those values of $\sem{\tau}$ that use \textit{at most} the names in $U$. Thus, types that do not use names (e.g., $\Bool$ and $\Real$) have $\sem{\sigma}_U=\sem{\sigma}$ for all $U$, whereas $\sem{\Name}_U=U$. Products and other containers use names at most in $U$ if all their elements do. Given two indexed families $X_U \subseteq X$ and $Y_U \subseteq Y$, we can define the indexed family $(X \Rightarrow Y)_U \subseteq X \Rightarrow Y$ by $\{f \in X \Rightarrow Y \mid \forall \text{ countable } V \subseteq [0,1],\, \forall x \in X_V, f(x) \in Y_{U \cup V}\}$. We use this construction to interpret function types: it encodes that a function uses a name if it can return values that use that name even on inputs that do not use the name. Given an indexed family $X_U \subseteq X$, we can also construct an indexed family $(\mathsf{Prob}~X)_U \subseteq \mathsf{Prob}~X$, as follows. First, following~\citet{DashKPS23}, define $\Omega \coloneq [0,1]^{\mathbb{N}^*}$ to be the quasi-Borel space of rose trees, which admits a uniform probability measure $\mu_\Omega$. Then define the indexed family of subsets $\Omega \supseteq \Omega_U \coloneq U^{\mathbb{N}^*}$. Then $(\mathsf{Prob}~X)_U \subseteq \mathsf{Prob}~X$ is the set of all probability measures $\mu = g_*\mu_\Omega$ that arise as pushforwards by some $g \in (\Omega \Rightarrow X)_U$. Intuitively, a probability measure uses a name if the event that the outcome uses that name occurs with positive probability. Probability measures do not use the \textit{fresh} names they generate, which are continuous samples and each occur with probability 0. Together with the function space definition, the $\mathsf{Prob}$ construction lets us define $\sem{\ProbTy{\sigma}{\tau}}_U = \{(\mu, f) \mid f \in \sem{\sigma \to \tau}_U, \mu \in (\mathsf{Prob}~\sem{\sigma})_U\}$.

We develop this perspective more formally below, showing how these indexed families of quasi-Borel spaces arise as the objects of a particular functor category.

\begin{definition}[$\pshqbs$]
  $\pshqbs$ is the full subcategory of $[\A, \qbs]$
  containing functors that preserve pullbacks,
  where $\qbs$ is the category of quasi-Borel spaces \cite{HeunenKSY17}
  and $\A$ is the category whose objects are countable subsets  of the interval $[0,1]$, as well as the interval $[0, 1]$ itself, and whose morphisms are given by subset inclusion.
\end{definition}

\begin{proposition}[$\A$ is bicartesian]
  The category $\A$ is finitely complete and finitely cocomplete.
\end{proposition}
\begin{proof}
  Clearly, $[0,1]$ is terminal in $\A$,  and pullbacks in $\A$ are given by set intersection.
  Hence, $\A$ is finitely complete \cite[Exercise III.4.10]{MacLane98}.
  Dually, $\emptyset$ is initial in $\A$,
  and pushouts in $\A$ are given by set union.
  This means that $\A$ is finitely cocomplete.
\end{proof}

\begin{proposition}[$\pshqbs$ is complete] \label{prop:pshqbs-complete}
  The category $\pshqbs$ is (small) complete,
  with limits given pointwise.
  Moreover, the inclusion functor $\pshqbs \hookrightarrow [\A,\qbs]$
  preserves these limits.
\end{proposition}
\begin{proof}
  Consider a pullback diagram $F: (\gets\cdot\to) \to \A$
  and a diagram $G: \J \to [\A, \qbs]$,
  where $\J$ is a small category
  and $G(j)$ is a pullback-preserving functor in $[\A,\qbs]$ for all $j$ in $\J$.
  Now, define $H = \lim_{j \in \J} G(j)$.
  We may then compute:
  \begin{align*}
    H\left(\lim_{X \in \gets\cdot\to} F(X)\right)
    & = \left(\lim_{j \in \J} G(j)\right)\left(\lim_{X \in \gets\cdot\to} F(X)\right) \\
    & \cong \lim_{j \in \J} G(j)\left(\lim_{X \in \gets\cdot\to} F(X)\right) & \textrm{limits computed pointwise \cite[\S I]{MacLaneM94}} \\
    & \cong \lim_{j \in \J} \lim_{X \in \gets\cdot\to} G(j)(F(X)) & G(j) \textrm{ preserves pullbacks} \\
    & \cong \lim_{X \in \gets\cdot\to} \lim_{j \in \J} G(j)(F(X)) & \textrm{interchange of limits \cite[\S IX]{MacLane98}} \\
    & \cong \lim_{X \in \gets\cdot\to} \left(\lim_{j \in \J} G(j)\right)(F(X)) & \textrm{limits computed pointwise} \\
    & = \lim_{X \in \gets\cdot\to} H(F(X))
  \end{align*}
  Note,
  limits in $[\A,\qbs]$
  are computed pointwise
  because $\qbs$ is (small) complete \cite{HeunenKSY17,MatacheMS22,BaezH11}.
  Thus, $H: \A \to \qbs$ preserves pullbacks.
  This means that $\pshqbs$
  is complete
  and inherits all limits
  from $[\A,\qbs]$.
\end{proof}

\begin{proposition}
  There is a faithful functor $\mathcal{C} : \pshqbs \to \qbs$
  that maps each functor $F: \A \to \qbs$ to $F[0,1]$
  and each natural transformation $\alpha: F \to G$ to $\alpha_{[0,1]} : F[0,1] \to G[0,1]$.
\end{proposition}
\begin{proof}
  $\mathcal{C}$ clearly forms a functor.
  To show that $\mathcal{C}$ is faithful,
  suppose that $\alpha_{[0,1]} = \beta_{[0,1]}$
  for two natural transformations $\alpha, \beta: F \to G$.
  We aim to show that this implies $\alpha = \beta$---or, equivalently,
  that $\alpha_U = \beta_U$ for all $U$ in $\A$.
  Because $\alpha$ and $\beta$ are natural, the following two diagrams commute:
  \begin{center}
    \begin{tikzpicture}[node distance=2cm]
      \node (0) {$F(U)$};
      \node[right of=0] (1) {$F[0,1]$};
      \node[below of=0] (2) {$G(U)$};
      \node[below of=1] (3) {$G[0,1]$};

      \path [>->] (0) edge node[above] {$F(\subseteq)$} (1)
        (2) edge node[below] {$G(\subseteq)$} (3);
      \path [->] (0) edge node[left] {$\alpha_U$} (2)
        (1) edge node[right] {$\alpha_{[0,1]}$} (3);
    \end{tikzpicture}
    \hspace*{1cm}
    \begin{tikzpicture}[node distance=2cm]
      \node (0) {$F(U)$};
      \node[right of=0] (1) {$F[0,1]$};
      \node[below of=0] (2) {$G(U)$};
      \node[below of=1] (3) {$G[0,1]$};

      \path [>->] (0) edge node[above] {$F(\subseteq)$} (1)
        (2) edge node[below] {$G(\subseteq)$} (3);
      \path [->] (0) edge node[left] {$\beta_U$} (2)
        (1) edge node[right] {$\beta_{[0,1]}$} (3);
    \end{tikzpicture}
  \end{center}
  for any $U$ in $\A$.
  The above diagrams then imply that:
  \begin{equation*}
    G(\subseteq) \circ \alpha_U = \alpha_{[0,1]} \circ F(\subseteq) = \beta_{[0,1]} \circ F(\subseteq) = G(\subseteq) \circ \beta_U
  \end{equation*}
  Since $F$ and $G$ preserve pullbacks, they also preserve monomorphisms \cite[\S II]{MacLaneM94}.
  In particular,
  $G(\subseteq)$ is monic.
  Hence, we may conclude that $\alpha_U = \beta_U$.
\end{proof}

\begin{corollary}[$\pshqbs$ is concrete] \label{cor:pshqbs-concrete}
  There is a faithful functor $\_ : \pshqbs \to \catset$
  that maps each functor $F : \A \to \qbs$ to $|F[0,1]|$
  and each natural transformation $\alpha: F \to G$
  to $|\alpha_{[0,1]}| : |F[0,1]| \to |G[0,1]|$.
  In other words, the category $\pshqbs$ is concrete.
\end{corollary}
\begin{proof}
  By the previous proposition, the functor
  that maps each $F: \A \to \qbs$ to $F[0,1]$ is faithful.
  It is also well-known that mapping each quasi-Borel space $(X, M_X)$
  to $X$---i.e., mapping the corresponding concrete sheaf $X: \mathsf{Sbs} \to \catset$ to $X(1)$---also forms a faithful functor \cite[Proposition 28]{BaezH11}.
  Hence, their composition is faithful.
\end{proof}

\begin{proposition}[Coproducts computed pointwise]
  Consider objects $F,G: \A \to \qbs$ in $\pshqbs$.
  Their coproduct $F + G: \A \to \qbs$
  can be computed pointwise:
  \begin{equation*}
    (F + G)(U) = F(U) + G(U)
  \end{equation*}
\end{proposition}

\begin{proposition}[$\pshqbs$ is cartesian closed]
  The category $\pshqbs$ is cartesian closed,
  with the exponential object $F^G: \A \to \qbs$
  given by:
  \begin{equation*}
    F^G(U) = \{ h \in \homset{\qbs}{G[0,1]}{F[0,1]} \mid \forall V \in \mathcal{A}, h(G(V)) \subseteq F(U \cup V)\}
  \end{equation*}
  viewed as a subobject of the $\qbs$ exponential $F[0,1]^{G[0,1]}$.
\end{proposition}

\begin{definition}[Rose trees]
  Let $\Omega = [0,1]^{\Nat^*}$ denote the
  quasi-Borel space of rose trees \cite[\S 5.4.1]{DashKPS23}
  and take $\mur$ to be the uniform probability
  distribution on $\Omega$.
  The space $\Omega$ extends to a functor:
  \begin{equation*}
    \Omega(U) = \Bigl\{
      \omega \in \Omega \mid \forall \rho \in \Nat^*,
      \omega(\rho) \in U
    \Bigr\}
  \end{equation*}
\end{definition}

\begin{definition}[Probability monad]
  The \emph{probability monad} on $\pshqbs$ is given by the quadruple $(\mMeas, \eta, \mu, \tau)$, where the endofunctor $\mMeas : \pshqbs \to \pshqbs$ is defined pointwise by
  \begin{equation*}
    \mMeas(F)(U) = \Bigl\{
      f_* \mu_{\Omega}
      \mid
      f \in (\Omega \Rightarrow F)(U)
    \Bigr\}
  \end{equation*}
  and the natural transformations $\eta$, $\mu$, and $\tau$ for each $U$ are the appropriate restrictions of the unit, multiplication, and strength of the Qbs probability monad.
\end{definition}

\begin{definition}[Trace probability functor]
  For each $F$ in $\pshqbs$, the \emph{trace probability functor}
  is given by $\mProb(F, -):  \pshqbs \to \pshqbs$,
  defined pointwise by:
  \begin{equation*}
    \mProb(F, G)(U) = \mMeas(F)(U) \times (F \Rightarrow G)(U)
  \end{equation*}
\end{definition}

In \cref{fig:nuqbs-type-semantics,fig:nuqbs-term-semantics}, we give a denotational semantics $\Qsem{-}$ to the types and terms of $\genlang{}$. Types are interpreted as objects in $\pshqbs$, and terms as morphisms (i.e., natural transformations). For readability,
we present the $\catset$
concretization
of each term interpretation
(given by the faithful functor $\_ : \pshqbs \to \catset$). Note that for each $U$, we have $\Qsem{\Gamma \vdash t : \tau}_U = \Csem{\Gamma \vdash t : \tau}|_{\Qsem{\Gamma}(U)}$. Throughout, when $\Csem{t}(\gamma)$ is a pair (as for $t : \ProbTy{\sigma}{\tau}$, where it is a distribution-and-return-map pair), we write $\Csem{t}_i := \pi_i \circ \Csem{t}$ for the $i$-th projection.

\begin{figure}
\begin{align*}
  \Qsem{\Bool}(U) & = (\Bool, \homset{\catmeas}{\Real}{\Bool}) & \Bool \textrm{ standard Borel} \\
  \Qsem{\Real}(U) & = (\Real, \homset{\catmeas}{\Real}{\Real}) & \Real \textrm{ standard Borel} \\
  \Qsem{\Nat}(U) & = (\Nat, \homset{\catmeas}{\Real}{\Nat}) & \Nat \textrm{ standard Borel} \\
  \Qsem{\Name}(U) & = (U, \homset{\catmeas}{\Real}{U}) & U \textrm{ standard Borel (discrete or $[0,1]$)} \\
  \Qsem{\sigma \times \tau}(U) & = \Qsem{\sigma}(U) \times \Qsem{\tau}(U) \\
  \Qsem{\List{\sigma}}(U) & = \coprod_{n \in \Nat} \Qsem{\sigma}^n(U) \\
  \Qsem{\NameMap{\tau}}(U) & = \left\{
    \ell \in \coprod_{n \in \Nat} (U \times \Qsem{\tau})^n
    \mid
    \ell \textrm{ sorted}
  \right\} \\
  \Qsem{\ProbTy[\kappa]{\sigma}{\tau}}(U)
  & = \mProb(\Qsem{\sigma}, \Qsem{\tau})(U) \\
  \Qsem{\tau \clos[\sigma] \tau'}(U) & = (\Qsem{\tau} \Rightarrow \Qsem{\tau'})(U) \\
  \Qsem{\Record{\ell_1: \tau_1, \dots, \ell_n: \tau_n}}(U)
  & = \prod_{(\ell_i: \tau_i) \in \{\ell_1 : \tau_1, \dots, \ell_n: \tau_n\}} \Qsem{\tau_i}(U)
\end{align*}
\caption{Semantics of $\genlang$ types as $\pshqbs$ objects.}\label{fig:nuqbs-type-semantics}
\end{figure}

\begin{figure}
\begin{align*}
  \Csem{\Gamma \vdash x : \tau} & = \pi_{(x, \tau)} \\
  \Csem{\Gamma \vdash (s,t): \sigma \times \tau}(\gamma) & = (\Csem{s}(\gamma), \Csem{t}(\gamma)) \\
  \Csem{\Gamma \vdash t.i : \tau_i}(\gamma) & = \pi_i(\Csem{t}(\gamma)) \\
  \Csem{\tyjudg{\Gamma}{\ite{s}{t}{t'}}{\sigma}}(\gamma) & =
  \begin{cases}
    \Csem{t}(\gamma) & \Csem{s}(\gamma) = \top \\
    \Csem{t'}(\gamma) & \textrm{otherwise}
  \end{cases} \\
  \Csem{\Gamma \vdash \lam{x}{t} : \tau \clos[\Gamma'] \tau'}(\gamma)
  & = v \mapsto \Csem{t}(\gamma, v) \\
  \Csem{\tyjudg{\Gamma}{s \; t}{\tau'}}(\gamma)
  & = \Csem{s}(\gamma)(\Csem{t}(\gamma)) \\
  \Csem{\tyjudg{\Gamma}{\return{t}}{\ProbTy[\kappa]{\Unit}{\tau}}}(\gamma)
  & = (\delta_{\Csem{\Unit}}, \_ \mapsto \Csem{t}(\gamma)) \\
  \Csem{\tyjudg{\Gamma}{\sample{t}}{\ProbTy[\kappa]{\sigma}{\sigma}}}(\gamma)
  & = (\pi_1(\Csem{t}(\gamma)), \id[\Csem{\sigma}]) \\
  \Csem{\tyjudg{\Gamma}{x \gets t; s}{\ProbTy[\kappa]{\Record{\cdots}}{\tau'}}}(\gamma)
  & = \begin{aligned}[t]
      \bigl( & \rho \leftsquigarrow (\pi_1 \circ \Csem{t})(\gamma); v =  (\pi_2 \circ \Csem{t})(\gamma)(\rho); \\
      & \rho' \leftsquigarrow (\pi_1 \circ \Csem{s})(\gamma,v); \ret{(\rho\concat\rho')}, \\
      & \rho \mapsto (\pi_2 \circ \Csem{s})(\gamma, (\pi_2 \circ \Csem{t})(\gamma)(\rho|_{L_1}))(\rho|_{L_2})
      \bigr)
  \end{aligned} \\
  \Csem{\tyjudg{\Gamma}{t \at \ell}{\ProbTy[\kappa]{\Record{\ell: \sigma}}{\tau}}}(\gamma)
  & = \begin{aligned}[t]\bigl(
    & (x \mapsto \{\ell \mapsto x\})_*((\pi_1 \circ \Csem{t})(\gamma)), \\
    & u \mapsto (\pi_2 \circ \Csem{t})(\gamma)(\pi_{\ell} u)
  \bigr)
  \end{aligned} \\
  \Csem{\tyjudg{\Gamma}{\forInRange{x}{t}{s}}{\ProbTy[\kappa]{(\List{\sigma})}{(\List{\tau})}}}(\gamma) 
  &= \begin{aligned}[t] ( &N \leftsquigarrow \sem{t}_1(\gamma); \bigotimes_{i=1}^N \sem{s}_1(\gamma, i), \\
     & \rho \mapsto (\sem{s}_2(\gamma, i)(\pi_i\rho))_{i=1}^{|\rho|})
     \end{aligned}
  \\
  \Csem{\tyjudg{\Gamma}{\const{fresh}~t}{\ProbTy[\kappa]{\NameSet}{\Unit}}}(\gamma)
  & = \begin{aligned}[t]
      \bigl( & n \leftsquigarrow (\pi_1 \circ \Csem{t})(\gamma); u \leftsquigarrow \bigotimes_{i=1}^n \Lambda_{[0,1]}; \ret(\mathsf{sort}(u)), \\
      & \_ \mapsto \unit
      \bigr)
  \end{aligned}
\end{align*}
\caption{Semantics of $\genlang$ terms as $\pshqbs$ morphisms. For readability, we present the $\catset$ concretization of each interpretation, $\Csem{\Gamma \vdash t : \tau} = \Qsem{\Gamma \vdash t : \tau}_{[0,1]}$. In general, $\Qsem{\Gamma \vdash t : \tau}_{U} = \Csem{\Gamma \vdash t : \tau}|_{\Qsem{\Gamma}(U)}$.}\label{fig:nuqbs-term-semantics}
\end{figure}

\clearpage 

\section{Density Transformation: Full Definition and Correctness}\label{sec:full-correctness-density}

The full type and term translations are given in \cref{fig:gen-to-inc-types-full,fig:gen-to-inc-terms-full}, extending the selected rules from \cref{fig:gen-to-inc-types,fig:gen-to-inc-terms}.

\begin{figure}
  \begin{align*}
    \GenToInc{\sigma} & = \sigma & \sigma \in \{\Bool, \Nat, \Real, \Name\} \\
    \GenToInc{\sigma \times \tau} & = \GenToInc{\sigma} \times \GenToInc{\tau} \\
    \GenToInc{\List{\sigma}} & = \List{\GenToInc{\sigma}} \\
    \GenToInc{\NameMap{\tau}} & = \NameMap{\GenToInc{\tau}} \\
    \GenToInc{\NameSet} & = \GenToInc{\NameMap{\Unit}} \\
    \GenToInc{\Record{\ell_1: \tau_1, \dots, \ell_n: \tau_n}} & = \Record{
      \ell_1: \GenToInc{\tau_1},
      \dots,
      \ell_n: \GenToInc{\tau_n}
    } \\
    \GenToInc{\ProbTy[\kappa]{\sigma}{\tau}} & =
    \GenToInc{\sigma} \times \NameSet
    \clos[\GenToInc{\kappa}]
    \GenToInc{\tau} \times \Real \\
    \GenToInc{\Dist[\kappa]{\sigma}} & = \GenToInc{\ProbTy[\kappa]{\sigma}{\Unit}} \\
    \GenToInc{\sigma \clos[\kappa] \tau} & =
    \GenToInc{\sigma} \clos[\GenToInc{\kappa}] \GenToInc{\tau}
  \end{align*}
  \caption{Density transformation from $\genlang$ to $\inclang$ on types.}
  \label{fig:gen-to-inc-types-full}
\end{figure}

\begin{figure}
  \begin{align*}
    \GenToInc{x} &= x \\
    \GenToInc{c} &= c_{\mathsf{inc}} \\
    \GenToInc{\return{t}} &= \lam{(\tr, \var{U})}{(\GenToInc{t}, 1)} \\
    \GenToInc{x \gets t_1; t_2} &= \lam{(\tr, \var{U})} (x, w_1) := \GenToInc{t_1}(\restrict{\tr}{L_1}, \var{U});\\
      &\qquad\qquad\;\;\; (y, w_2) := \GenToInc{t_2} (\restrict{\tr}{L_2}, \var{U} \cup \names{\restrict{\tr}{L_1}}); \\
      &\qquad\qquad\;\;\; (y, w_1 \cdot w_2) \\
      & \text{ where } t_1 : \ProbTy{\Record{\ell: \sigma_\ell \mid \ell \in L_1}}{\tau_1} \text{ and } t_2 : \ProbTy{\Record{\ell: \sigma_\ell' \mid \ell \in L_2}}{\tau_2} \\
    \GenToInc{\sample{t}} &= \lam{(\tr, \var{U})} (\tr, (\GenToInc{t}(\tr, \var{U})).2) \\
    \GenToInc{t \at \ell} &= \lam{(\tr, \var{U})} \GenToInc{t}(\tr.\ell, \var{U}) \\
    \scalebox{0.9}{\(\displaystyle\GenToInc{\begin{aligned}
      &\forRangeWithUsing{x}{d: \Dist{\Nat}}{z}{s \\
      &}{v_1 := y_1[u_1], \dots, v_n := y_n[u_n]}{t}
    \end{aligned}}\)} &= \scalebox{0.9}{\(\displaystyle \left(\begin{aligned}
      \lam{(\tr, \var{U})} {} & \var{N} := \const{length} \; \tr; \; (\_, \var{w}_{\var{N}}) := \GenToInc{d}(\var{N}, \var{U}); \\
        &\var{ys} := \forWithUsing{x}{\const{range} \; \var{N}}{(z, \var{V})}{(\GenToInc{s}, \var{U}) \\
        &\qquad }{v_1 := y_1[\GenToInc{u_1}], \dots, v_n := y_n[\GenToInc{u_n}], \\
        &\qquad \quad \tr' := \tr[x]}{  ((r, z), w) :=  \GenToInc{t}(\tr', \var{V}); \\ &\qquad \quad ((r, w), (z, \var{V} \cup \names{\tr'})) }; \\
        &(\var{rs}, \var{ws}) := \const{unzip} \; \var{ys}; \; (\var{rs}, \var{w}_{\var{N}} \cdot \const{product} \; \var{ws})
    \end{aligned} \right)\)} \\
    \GenToInc{\lam{x}{t}} & = \lam{x}{\GenToInc{t}} \\
    \GenToInc{s \; t} &= \GenToInc{s} \; \GenToInc{t} \\
    \GenToInc{(s,t)} &= (\GenToInc{s}, \GenToInc{t}) \\
    \GenToInc{t.i} &= \GenToInc{t}.i \\
    \GenToInc{\ite{s}{t}{t'}} &= \ite{\GenToInc{s}}{\GenToInc{t}}{\GenToInc{t'}}
  \end{align*}
  \caption{Density transformation from $\genlang$ to $\inclang$ on terms.}
  \label{fig:gen-to-inc-terms-full}
\end{figure}

\begin{figure}
\[
  \sem{\const{fresh}_{\mathsf{inc}}}((\gamma, d)) = \left(\gamma,  (\gamma, (\tr, \var{U})) \mapsto \begin{cases}
  (\unit, \pi_2(d(\gamma, (|\tr|, U))) \cdot |\tr|!) & \domain(\tr) \cap U = \emptyset \\
  (\unit, 0) & \textsf{otherwise}
  \end{cases}\right)
\]
\caption{Semantics for the density transformation of $\const{fresh}$. $|\tr|$ indicates the size of the name set $\tr$.}
\end{figure}

\subsection{Stock Measures}
\label{sec:stock-measures}

As presented in \cref{sec:density-transformation}, we define for each ground type $\sigma$ in $\genlang$ a family of reference measures $\nu_U^\sigma$, indexed by \textit{countable} name sets $U$.

\begin{align*}
  \nu_{U}^{\sigma} & = \counting_{\Csem{\sigma}} & \sigma \in \{\Bool,\Nat\} \\
  \nu_{U}^{\Real} & = \Lambda_{\Real} \\
  \nu_{U}^{\Name} & = \counting_{U} + \Lambda_{[0,1]} \\
  \nu_U^{\sigma \times \tau} & = \begin{aligned}[t]
      & \rho_1 \leftsquigarrow \nu_U^\sigma; \\
      & \rho_2 \leftsquigarrow
      \nu_{U \cup \names{\rho_1}}^\tau; \\
      & \ret{(\rho_1,\rho_2)}
  \end{aligned} \\
  \nu_{U}^{\Record{\ell_1: \tau_1, \dots, \ell_n: \tau_n}}
  & = \nu_U^{\tau_1 \times \cdots \times \tau_n}
  & (\nu_U^{\Unit} = \delta_{\Csem{\Unit}}) \\
  \nu_U^{\List{\sigma}} & = \sum_{n \in \Nat} (\iota_n)_* \nu_U^{\sigma^n} & (\textsf{with coproduct injections } \iota_n) \\
  \nu_U^{\NameMap{\sigma}}
  & = \sum_{n \in \Nat} (\kappa_n)_* (\nu_U^{(\Name \times \sigma)^n}) & (\textsf{with coproduct injections } \kappa_n)
\end{align*}

\paragraph{Symmetry of product reference measures.}
A footnote in \cref{sec:density-transformation} mentions that the product clause above is order-independent. We give a proof below.

\begin{lemma}
  \label{lem:nu-product-swap}
  For all countable $U \subseteq [0,1]$ and all ground types $\sigma,\tau$,
  \[
    \nu_U^{\sigma \times \tau}
    = \swap_* \nu_U^{\tau \times \sigma},
  \]
  where $\swap(x,y) = (y,x)$.
\end{lemma}
\begin{proof}
  We define the shorthand $P(\sigma, \tau)$ for the proposition that
  \[\forall U \subseteq [0,1]\text{ countable},\;
    \forall h:\sem{\sigma}\times\sem{\tau}\to\RR_{\ge 0}\text{ measurable},
  \]
  \[
    \int h(a,b)\,\nu_U^{\sigma\times\tau}(\dif{(a,b)})
    =
    \int h(a,b)\,(\swap_*\nu_U^{\tau\times\sigma})(\dif{(a,b)}).
  \]
  The lemma states that this is true for all ground types $\sigma$ and $\tau$.
  We prove this by structural induction on the pair of types $(\sigma,\tau)$. Note that $P(\sigma, \tau) \iff P(\tau, \sigma)$:
  $\swap\circ\swap = \id$, so
  \[
    \nu_U^{\sigma\times\tau}=\swap_*\nu_U^{\tau\times\sigma}
    \quad\Longleftrightarrow\quad
    \nu_U^{\tau\times\sigma}=\swap_*\nu_U^{\sigma\times\tau}.
  \]
  by applying $\swap_*$ to each side of the equality.
  For convenience, we write $N(x):=\names{x}$.

\begin{enumerate}[leftmargin=*]

\item \textbf{Base cases.}

  \begin{itemize}

  \item \textbf{Base case 1.} If $\sigma \in \{\Unit, \Bool, \Nat, \Real\}$, then $N(x) = \emptyset$ for all $x \in \sem{\sigma}$, and $\nu^\sigma_U = \nu^\sigma_V$ for all $U, V$. By the definition of $\nu_U^{\sigma\times\tau}$, we have:
\[
    \nu_U^{\sigma\times\tau}(\dif{(a,b)})
    =
    \nu_U^\sigma(\dif{a})\,\nu_{U\cup N(a)}^\tau(\dif{b})
    =
    \nu^\sigma_U(\dif{a})\,\nu^\tau_U(\dif{b}),
  \]
  where the second equality follows because $N(a) = \emptyset$. Similarly,
  \[
    \nu_U^{\tau\times\sigma}(\dif{(b,a)})
    =
    \nu_U^\tau(\dif{b})\,\nu_{U\cup N(b)}^\sigma(\dif{a}) = \nu^\tau_U(\dif{b})\,\nu^\sigma_U(\dif{a}),
  \]
  where the second equality follows because $\nu^\sigma_U = \nu^\sigma_{U \cup N(b)}$.
  By commutativity of the Qbs measure monad, $\nu^\sigma_U\otimes\nu^\tau_U=\swap_*(\nu^\tau_U\otimes\nu^\sigma_U)$.
  The case where $\tau \in \{\Unit, \Bool, \Nat, \Real\}$ follows by symmetry.

  \item \textbf{Base case 2.} Suppose $\sigma=\tau=\Name$.
  Using $\nu_U^\Name=\counting(U)+\Lambda_{[0,1]}$ and $N(a)=\{a\}$:
  \[
    \nu_U^{\Name\times\Name}(\dif{(a,b)})
    =
    \nu_U^\Name(\dif{a})\,\nu_{U\cup\{a\}}^\Name(\dif{b})
    =
    \nu_U^\Name(\dif{a})\bigl(\nu_U^\Name(\dif{b})+\delta_a(\dif{b})\bigr).
  \]
  Therefore, for measurable $h\ge 0$,
  \[
    \int h\,\dif{\nu}_U^{\Name\times\Name}
    =
    \iint h(a,b)\,\nu_U^\Name(\dif{b})\,\nu_U^\Name(\dif{a})
    +
    \int h(a,a)\,\nu_U^\Name(\dif{a}).
  \]
  The same expansion is obtained for
  $\int h\,\dif{(\swap_*\nu_U^{\Name\times\Name})}$:
  \[
    \begin{aligned}
      \int h\,\dif{(\swap_*\nu_U^{\Name\times\Name})}
      &=
      \iint h(a,b)\,\nu_{U\cup\{b\}}^\Name(\dif{a})\,\nu_U^\Name(\dif{b}) \\
      &=
      \iint h(a,b)\,\nu_U^\Name(\dif{a})\,\nu_U^\Name(\dif{b})
      +
      \int h(b,b)\,\nu_U^\Name(\dif{b}).
    \end{aligned}
  \]
\end{itemize}
  \item \textbf{Product constructor.}
  Assume $\sigma=\sigma_1\times\sigma_2$ and $\tau$ is a ground type.
  Assume by induction that $P(\sigma_1,\tau)$ and $P(\sigma_2,\tau)$ hold.
  Let $h:\sem{\sigma_1\times\sigma_2}\times\sem{\tau}\to\RR_{\ge0}$ be measurable.
  Expanding $\nu_U^{(\sigma_1\times\sigma_2)\times\tau}$:
  \[
  \begin{aligned}
    I
    &:=
    \int h((x,y),z)\,\nu_U^{(\sigma_1\times\sigma_2)\times\tau}(\dif{((x,y),z)}) \\
    &=
    \iiint
      h((x,y),z)\,
      \nu_{U\cup N(x)\cup N(y)}^\tau(\dif{z})\,
      \nu_{U\cup N(x)}^{\sigma_2}(\dif{y})\,
      \nu_U^{\sigma_1}(\dif{x}).
  \end{aligned}
  \]
  Here we used $N((x,y))=N(x)\cup N(y)$ and unfolded both product clauses.
  Apply $P(\sigma_2,\tau)$ at base set $U\cup N(x)$ to the inner two integrals,
  with the measurable function $h_x(y,z):=h((x,y),z)$ (for a fixed $x$). This allows us to switch the order in which $z$ and $y$ are integrated:
  \[
    I
    =
    \iiint
      h((x,y),z)\,
      \nu_{U\cup N(x)\cup N(z)}^{\sigma_2}(\dif{y})\,
      \nu_{U\cup N(x)}^\tau(\dif{z})\,
      \nu_U^{\sigma_1}(\dif{x}).
  \]
  Now apply $P(\sigma_1,\tau)$ at base set $U$ to the measurable  function
  \[
    g(x,z):=\int h((x,y),z)\,\nu_{U\cup N(x)\cup N(z)}^{\sigma_2}(\dif{y}),
  \]
  to obtain
  \[
    I
    =
    \iiint
      h((x,y),z)\,
      \nu_{U\cup N(z)\cup N(x)}^{\sigma_2}(\dif{y})\,
      \nu_{U\cup N(z)}^{\sigma_1}(\dif{x})\,
      \nu_U^\tau(\dif{z}).
  \]
  The inner two integrals are exactly $\nu_{U\cup N(z)}^{\sigma_1\times\sigma_2}$ by definition, so
  \[
    \begin{aligned}
      I
      &=
      \int h((x,y),z)\,\nu_U^{\tau\times(\sigma_1\times\sigma_2)}(\dif{(z,(x,y))}) \\
      &=
      \int h\,\dif{(\swap_*\nu_U^{\tau\times(\sigma_1\times\sigma_2)})}.
    \end{aligned}
  \]
  Thus $P(\sigma_1\times\sigma_2,\tau)$ holds.
  The case where the {right} component is a product follows by symmetry.

  \item \textbf{List constructor.}
  Assume $\sigma=\List{\rho}$, and $\tau$ is some ground type (the case with the list on the right is symmetric).
  Define $\rho^0:=\Unit$ and $\rho^{n+1}:=\rho^n\times\rho$.
 Let $\iota_n : \sem{\rho^n} \to \sem{\List{\rho}}$ be the $n$-th coproduct injection.
  The list reference measure is
  \[
    \nu_U^{\List{\rho}}
    =
    \sum_{n=0}^\infty (\iota_n)_* \nu_U^{\rho^n}.
  \]
  Also, by definition of $N$ on lists, for every $x\in\sem{\rho^n}$,
  \[
    N(\iota_n(x)) = N(x).
  \]
  Therefore,
  \[
    \begin{aligned}
      \nu_U^{\List{\rho}\times\tau}
      &=
      \left(
        a \leftsquigarrow \sum_{n=0}^\infty (\iota_n)_* \nu_U^{\rho^n};\;
        z \leftsquigarrow \nu_{U\cup N(a)}^\tau;\;
        \ret{(a,z)}
      \right) \\
      &=
      \sum_{n=0}^\infty
      (\iota_n\times\id[\tau])_*\,\nu_U^{\rho^n\times\tau}.
    \end{aligned}
  \]
  Likewise,
  \[
    \begin{aligned}
      \nu_U^{\tau\times\List{\rho}}
      &=
      \left(
        z \leftsquigarrow \nu_U^\tau;\;
        a \leftsquigarrow \sum_{n=0}^\infty (\iota_n)_* \nu_{U\cup N(z)}^{\rho^n};\;
        \ret{(z,a)}
      \right) \\
      &=
      \sum_{n=0}^\infty (\id[\tau]\times\iota_n)_*\,\nu_U^{\tau\times\rho^n}.
    \end{aligned}
  \]

  We first prove, by induction on $n$, the auxiliary claim
  \[
    C_n:\quad P(\rho^n,\tau).
  \]
  For $n=0$, $\rho^0=\Unit$, and $C_0 = P(\Unit, \tau)$ has already been covered as a base case in the overall proof.
  For the inductive case, suppose $C_n$ by induction, and recall that we also have the outer inductive hypothesis $P(\rho,\tau)$. Then by the Product Constructor case, we have $P(\rho^{n+1},\tau)$.
  Therefore, $C_n$ holds for all $n$.

  Now, for measurable $h\ge 0$:
  \[
  \begin{aligned}
    \int h\,\dif{\nu}_U^{\List{\rho}\times\tau}
    &=
    \sum_{n=0}^\infty
    \int h\,\dif{\!\left((\iota_n\times\id[\tau])_*\,\nu_U^{\rho^n\times\tau}\right)}
    \\
    &=
    \sum_{n=0}^\infty
    \int h\,\dif{\!\left((\iota_n\times\id[\tau])_*\,\swap_*\,\nu_U^{\tau\times\rho^n}\right)}
    &&\text{(by $C_n$)}\\
    &=
    \sum_{n=0}^\infty
    \int h\,\dif{\!\left(\swap_*(\id[\tau]\times\iota_n)_*\,\nu_U^{\tau\times\rho^n}\right)}\\
    &=
    \int h\,\dif{\!\left(
      \swap_*
      \sum_{n=0}^\infty (\id[\tau]\times\iota_n)_*\,\nu_U^{\tau\times\rho^n}
    \right)}\\
    &=
    \int h\,\dif{(\swap_*\nu_U^{\tau\times\List{\rho}})}.
  \end{aligned}
  \]
  Therefore $P(\List{\rho},\tau)$ holds (and by symmetry, so does $P(\tau, \List{\rho})$).

  \item \textbf{NameMap constructor.}
  Assume $\sigma=\NameMap{\rho}$ and $\tau$ is some ground type.
  By definition,
  \[
    \nu_U^{\NameMap{\rho}}
    =
    \sum_{n=0}^\infty (\kappa_n)_* \nu_U^{(\Name\times\rho)^n}.
  \]
  where $\kappa_n : \sem{(\Name\times\rho)^n}\to\sem{\NameMap{\rho}}$ are the coproduct injections.
  As in the List case, using $N(\kappa_n(x))=N(x)$,
  \[
    \begin{aligned}
      \nu_U^{\NameMap{\rho}\times\tau}
      &=
      \left(
        a \leftsquigarrow \sum_{n=0}^\infty (\kappa_n)_* \nu_U^{(\Name\times\rho)^n};\;
        z \leftsquigarrow \nu_{U\cup N(a)}^\tau;\;
        \ret{(a,z)}
      \right) \\
      &=
      \sum_{n=0}^\infty (\kappa_n\times\id[\tau])_*\,\nu_U^{((\Name\times\rho)^n)\times\tau},\\
      \nu_U^{\tau\times\NameMap{\rho}}
      &=
      \left(
        z \leftsquigarrow \nu_U^\tau;\;
        a \leftsquigarrow \sum_{n=0}^\infty (\kappa_n)_* \nu_{U\cup N(z)}^{(\Name\times\rho)^n};\;
        \ret{(z,a)}
      \right) \\
      &=
      \sum_{n=0}^\infty (\id[\tau]\times\kappa_n)_*\,\nu_U^{\tau\times(\Name\times\rho)^n}.
    \end{aligned}
  \]
  We can now repeat exactly the same argument as in the List case, but with $\rho$ replaced by $(\Name\times\rho)$.
\end{enumerate}
\end{proof}

\subsection{Logical Relation}

In \cref{fig:logrel1}, for each type $\tau$ in $\genlang$,
we define a binary relation $\logrel[\tau] \subseteq \underline{\Qsem{\tau}} \times \Ssem{\GenToInc{\tau}}$,
where $\_ : \pshqbs \to \catset$ is the faithful functor given in \Cref{cor:pshqbs-concrete}. At ground types, the relation is the identity. At function types, we use a variant of the usual lifting of logical relations to function types, adapted to account for the fact that $\Ssem{\tau \to_\sigma \tau'}$ is a space of explicit closures. The interesting case is the one for probabilistic programs of type $\ProbTy[\kappa]{\sigma}{\tau}$. It relates the meaning of the $\genlang$ program, $(\mu, f)$, to the translation $(\gamma, \langle g, w\rangle)$ when $g(\gamma, -)$ is related to $f$, and $w(\gamma, -, U)$ is the density $\frac{\dif{\mu}}{\dif{\nu_U^\sigma}}$.
\footnote{There is not a case for $\mathsf{Dist}$, because in this appendix, we have identified $\Dist[\kappa]{\sigma}$ with $\ProbTy[\kappa]{\sigma}{\Unit}$, by the bijection $\mu \mapsto (\mu, !)$, where $! : \llbracket\sigma\rrbracket \to 1$ is the unique map to the terminal object of $\pshqbs$.}

With these relations defined, we establish the fundamental lemma:

\begin{figure}
\begin{align*}
  \logrel[\Bool] & = \{
    (b,b') \in \Bool \times \Bool
    \mid
    b = b'
  \} \\
  \logrel[\Real] & = \{
    (r,r') \in \Real \times \Real
    \mid
    r = r'
  \} \\
  \logrel[\Nat] & = \{
    (n,n') \in \Nat \times \Nat
    \mid
    n = n'
  \} \\
  \logrel[\Name] & = \{
    (r_1, r_2) \in [0,1] \times [0,1]
    \mid
    r_1 = r_2
  \} \\
  \logrel[\sigma \times \tau] & = \begin{aligned}[t]
    \Bigl\{
      & ((s_1,s_2), (t_1, t_2)) \in \underline{\Qsem{\sigma \times \tau}} \times \Ssem{\GenToInc{\sigma \times \tau}}
      \mid \\
      & (s_1,t_1) \in \logrel[\sigma], (s_2,t_2) \in \logrel[\tau]
    \Bigr\}
  \end{aligned} \\
  \logrel[\List{\sigma}] & = \begin{aligned}[t]
    \Bigl\{
      & ((n_1, t_1), (n_2, t_2)) \in \Csem{\List{\sigma}} \times \Ssem{\GenToInc{\List{\sigma}}}
      \mid \\
      & n_1 = n_2, \forall 1 \leq i \leq n_1, (\pi_i(t_1), \pi_i(t_2)) \in \logrel[\sigma]
    \Bigr\}
  \end{aligned} \\
  \logrel[\NameMap{\tau}] & = \begin{aligned}[t]
    \Bigl\{
      & ((n_1, t_1), m) \in \Csem{\NameMap{\tau}} \times \Ssem{\GenToInc{\NameMap{\tau}}}
      \mid \\
      & n_1 = |\domain(m)|, \forall 1 \leq i \leq n_1, \pi_1(\pi_i(t_1)) \in \domain(m) \wedge (\pi_2(\pi_i(t_1)), m(\pi_1(\pi_i(t_1)))) \in \logrel[\tau]
    \Bigr\}
  \end{aligned} \\
  \logrel[{\ProbTy[\kappa]{\sigma}{\tau}}] & = \begin{aligned}[t]
    \Bigl\{
      & ((\mu,f), (\gamma, \langle g,w \rangle)) \in \underline{\Qsem{\ProbTy[\kappa]{\sigma}{\tau}}} \times \Ssem{\GenToInc{\ProbTy[\kappa]{\sigma}{\tau}}}
      \mid \\
      & ((\rho, U) \mapsto f(\rho), (\gamma,g)) \in \logrel[\sigma \times \NameSet \to \tau], \\
      & \forall \rho \in \Ssem{\sigma},
      \forall U \supseteq \supp{\mu},
      w(\gamma, \rho, U) = \frac{\dif{\mu}}{\dif{\nu_U^{\sigma}}}(\rho)
    \Bigr\}
  \end{aligned} \\
  \logrel[\tau {\clos[\sigma]} \tau'] & = \begin{aligned}[t]
    \Bigl\{
      & (f_1, (\gamma, f_2)) \in \Csem{\tau {\clos[\sigma]} \tau'} \times \Ssem{\GenToInc{\tau {\clos[\sigma]} \tau'}}
      \mid \\
      & \forall (x_1, x_2) \in \logrel[\tau],
      (f_1(x_1), f_2(\gamma, x_2)) \in \logrel[\tau']
    \Bigr\}
  \end{aligned} \\
  \logrel[\Record{\ell_1: \tau_1, \dots, \ell_n: \tau_n}] & = \begin{aligned}[t]
    \Bigl\{
      & (t_1, t_2) \in \underline{\Qsem{\Record{\ell_1: \tau_1, \dots, \ell_n: \tau_n}}} \times \Ssem{\GenToInc{\Record{\ell_1: \tau_1, \dots, \ell_n: \tau_n}}}
      \mid \\
      & \forall 1 \leq i \leq n. \bigl(
        \pi_{(\ell_i: \tau_i)}(t_1),
        \pi_{(\ell_i: \tau_i)}(t_2)
      \bigr) \in \logrel[\tau_i]
    \Bigr\}
  \end{aligned}
\end{align*}
\caption{The logical relation $\logrel[\tau]$.}
\label{fig:logrel1}
\end{figure}

\begin{lemma}[Fundamental Lemma]
  Let $\tyjudg{\Gamma}{t}{\tau}$ in \genlang{}.
  Then,
  for all $(\gamma_1, \gamma_2) \in \logrel[\Gamma]$,
  it follows that $(\Csem{t}(\gamma_1), \Ssem{\GenToInc{t}}(\gamma_2)) \in \logrel[\tau]$.
\end{lemma}
\begin{proof}[Proof Sketch]
  The proof is by induction on $\tyjudg{\Gamma}{t}{\tau}$.
  The cases
  for variables, pairs, projections,
  and function application
  are standard. We give several illustrative cases
  below.

  \paragraph{Case: Lambda Terms}
  Suppose $\tyjudg{\Gamma}{\lam{x} t}{\tau \clos[\toType{\Gamma|_{\freeVars{\lam{x} t}}}] \tau'}$
  and let $f = \Csem{\lam{x}{t}}(\gamma_1)$ and $(\gamma, g) = \Ssem{\GenToInc{\lam{x}{t}}}(\gamma_2)$.
  Consider any $(s_1, s_2) \in \logrel[\tau]$.
  Let $F:=\freeVars{\lam{x}t}$.
  By the $\inclang$ closure construction semantics, $\gamma=\pi_F(\gamma_2)$ and
  \[
    g(\gamma,s_2)=
    \Ssem{\tyjudg{\GenToInc{\Gamma|_F}, x:\GenToInc{\tau}}{\GenToInc{t}}{\GenToInc{\tau'}}}(\pi_F(\gamma_2),s_2).
  \]
  Since $(\gamma_1,\gamma_2)\in\logrel[\Gamma]$, we have
  $(\pi_F(\gamma_1),\pi_F(\gamma_2))\in\logrel[\Gamma|_F]$.
  Hence by the inductive hypothesis for $t$ in context $\Gamma|_F,x:\tau$, we obtain
  \[
    \left(
      \Csem{\tyjudg{\Gamma|_F, x:\tau}{t}{\tau'}}(\pi_F(\gamma_1), s_1),
      \Ssem{\tyjudg{\GenToInc{\Gamma|_F}, x:\GenToInc{\tau}}{\GenToInc{t}}{\GenToInc{\tau'}}}(\pi_F(\gamma_2), s_2)
    \right) \in \logrel[\tau'].
  \]
  Because $t$ depends only on $F\cup\{x\}$, we have
  \[
    \Csem{\tyjudg{\Gamma, x:\tau}{t}{\tau'}}(\gamma_1,s_1)
    =
    \Csem{\tyjudg{\Gamma|_F, x:\tau}{t}{\tau'}}(\pi_F(\gamma_1),s_1).
  \]
  Therefore, with the equality above for $g$,
  $(f(s_1), g(\gamma,s_2)) \in \logrel[\tau']$.

  \paragraph{Case: Return}
  Suppose $\tyjudg{\Gamma}{\return{t}}{\ProbTy[\toType{\Gamma|_{\freeVars{t}}}]{\Unit}{\tau}}$
  and
  let $(\mu,f) = \Csem{\return{t}}(\gamma_1)$
  and $(\gamma, \langle g, w\rangle) = \Ssem{\GenToInc{\return{t}}}(\gamma_2)$.
  By the inductive hypothesis,
  $(t_1, t_2) := (\Csem{t}(\gamma_1), \Ssem{\GenToInc{t}}(\gamma_2)) \in \logrel[\tau]$.
  \begin{itemize}
      \item For any $((s_1,U_1),(s_2,U_2)) \in \logrel[\Unit \times \NameSet]$,
      we have
      $(f(s_1), g(\gamma, s_2, U_2)) = (t_1, t_2) \in \logrel[\tau]$.
      \item Recall $\mu = \delta_{\Csem{\unit}}$. For $\rho = \Ssem{\unit}$ and any $U$,
      it follows that
      $\left(\frac{\dif{\delta}_{\Csem{\unit}}}{\dif{\nu}_{U}^{\Unit}}\right)(\rho) = 1 = w(\gamma,\rho, U)$.
  \end{itemize}

  \paragraph{Case: Sample}
  Suppose $\tyjudg{\Gamma}{\sample{t}}{\ProbTy[\toType{\Gamma|_{\freeVars{t}}}]{\sigma}{\sigma}}$
  and let $(\mu,f) = \Csem{\sample{t}}(\gamma_1)$
  and $(\gamma, \langle g, w\rangle) = \Ssem{\GenToInc{\sample{t}}}(\gamma_2)$.
  Let $((\mu',f'), (\gamma', \langle g', w'\rangle)) := (\Csem{t}(\gamma_1), \Ssem{\GenToInc{t}}(\gamma_2)) \in \logrel[{\ProbTy[\kappa]{\sigma}{\Unit}}]$ (by the inductive hypothesis).
  By the definition of $\GenToInc{\sample{t}}$, the translated sampler reuses the same density component as $\GenToInc{t}$, so
  \[
    \gamma = \gamma' \quad\text{and}\quad w = w'
  \]
  (as functions of $(\rho,U)$).
  \begin{itemize}
      \item Consider $((s_1, U_1),(s_2, U_2)) \in \logrel[\sigma \times \NameSet]$.
      By semantics of $\sample{\cdot}$ and its translation, for all traces $\rho$ and name-sets $U$ we have
      $f(\rho)=\rho$ and $g(\gamma,\rho,U)=\rho$, hence
      \[
        (f(s_1), g(\gamma, s_2, U_2)) = (s_1, s_2) \in \logrel[\sigma].
      \]
      \item For any $\rho \in \Ssem{\sigma}$ and $U \supseteq \supp{\mu} = \supp{\mu'}$,
      \[
        \left(\dfrac{\dif{\mu}}{\dif{\nu_U^{\sigma}}}\right)(\rho)
        =
        \left(\dfrac{\dif{\mu'}}{\dif{\nu_U^{\sigma}}}\right)(\rho)
        =
        w'(\gamma', \rho, U)
        =
        w(\gamma, \rho, U).
      \]
  \end{itemize}

  \paragraph{Case: Let}
  Suppose $\tyjudg{\Gamma}{x \gets s; t}{\ProbTy[\toType{\Gamma|_{\freeVars{x \gets s; t}}}]{\Record{\ell_1: \sigma_1, \dots, \ell_m: \sigma_m, \ell_1': \sigma_1', \dots, \ell_n': \sigma_n'}}{\tau'}}$
  and let $(\mu,f) = \Csem{x \gets s; t}(\gamma_1)$
  and $(\gamma, \langle g, w \rangle) = \Ssem{\GenToInc{x \gets s; t}}(\gamma_2)$.
  By the inductive hypothesis, let
  \[
    ((\mu_1, f_1), (\gamma_1', \langle g_1, w_1 \rangle))
    =
    (\Csem{s}(\gamma_1), \Ssem{\GenToInc{s}}(\gamma_2))
    \in
    \logrel[{\ProbTy[\kappa']{\Record{\ell_1:\sigma_1,\dots,\ell_m:\sigma_m}}{\tau}}].
  \]
  For each related pair $(v_1,v_2)\in\logrel[\tau]$, define
  \[
    ((\mu_2^{v_1}, f_2^{v_1}), (\gamma_2'^{v_2}, \langle g_2^{v_2}, w_2^{v_2}\rangle))
    :=
    (\Csem{t}(\gamma_1, v_1), \Ssem{\GenToInc{t}}(\gamma_2, v_2)).
  \]
  By the inductive hypothesis on $t$, each such pair lies in
  $\logrel[{\ProbTy[\kappa']{\Record{\ell_1':\sigma_1',\dots,\ell_n':\sigma_n'}}{\tau'}}]$.
  \begin{itemize}
      \item For any $((s_1,n_1),(s_2,n_2)) \in \logrel[\Record{\ell_1:\sigma_1,\dots,\ell_m:\sigma_m,\ell_1':\sigma_1',\dots,\ell_n':\sigma_n'} \times \NameSet]$, let
      \[
        v_1 := f_1(s_1|_{L_1}), \qquad
        v_2 := g_1(\gamma_1', s_2|_{L_1}, n_2).
      \]
      From the first conjunct of
      $\logrel[{\ProbTy[\kappa']{\Record{\ell_1:\sigma_1,\dots,\ell_m:\sigma_m}}{\tau}}]$, we get $(v_1,v_2)\in\logrel[\tau]$.
      Therefore, by the inductive hypothesis instance for $t$ at $((\gamma_1, v_1),(\gamma_2, v_2))$,
      \[
        \bigl(f_2^{v_1}(s_1|_{L_2}),\;
        g_2^{v_2}(\gamma_2'^{v_2}, s_2|_{L_2}, n_2 \cup \names{s_2|_{L_1}})\bigr)\in\logrel[\tau'].
      \]
      This is exactly the required functional part for $(f,g)$.
      \item Consider any $U \supseteq \supp{\mu}$. We will proceed by reasoning in the monadic metalanguage \cite{Moggi91,Moggi89} of the $\qbs$ measure monad \cite{ScibiorKVSYCOMHG18,VakarO18,Staton17}:
      \begin{align*}
        \mu & = \cdot \vdash \begin{aligned}[t]
          & \rho_1 \leftsquigarrow \mu_1; \\
          & v = f_1(\rho_1); \\
          & \rho_2 \leftsquigarrow \mu_2^{v}; \\
          & \ret{\rho_1 \concat \rho_2}
        \end{aligned}
        = \cdot \vdash \begin{aligned}[t]
          & \rho_1 \leftsquigarrow \nu_{U}^{\Record{\ell_1:\sigma_1, \dots, \ell_m:\sigma_m}}; \\
          & \score(w_1(\gamma_1', \rho_1, U)); \\
          & v = f_1(\rho_1); \\
          & v_2 = g_1(\gamma_1', \rho_1, U); \\
          & \rho_2 \leftsquigarrow \nu_{U \cup \names{\rho_1}}^{\Record{\ell_1':\sigma_1',\dots,\ell_n':\sigma_n'}}; \\
          & \score(w_2^{v_2}(\gamma_2'^{v_2}, \rho_2, U \cup \names{\rho_1})); \\
          & \ret{\rho_1 \concat \rho_2}
        \end{aligned} \\
        & = \cdot \vdash \begin{aligned}[t]
          & \rho_1 \leftsquigarrow \nu_{U}^{\Record{\ell_1:\sigma_1, \dots, \ell_m:\sigma_m}}; \\
          & v = f_1(\rho_1); \\
          & v_2 = g_1(\gamma_1', \rho_1, U); \\
          & \rho_2 \leftsquigarrow \nu_{U \cup \names{\rho_1}}^{\Record{\ell_1':\sigma_1',\dots,\ell_n':\sigma_n'}}; \\
          & \rho \leftsquigarrow \ret{\rho_1 \concat \rho_2}; \\
          & \score(w_1(\gamma_1', \rho|_{L_1}, U) \cdot w_2^{g_1(\gamma_1', \rho|_{L_1}, U)}(\gamma_2'^{g_1(\gamma_1', \rho|_{L_1}, U)}, \rho|_{L_2}, U \cup \names{\rho|_{L_1}})); \\
          & \ret{\rho}
        \end{aligned} \\
        & = \cdot \vdash \begin{aligned}[t]
          & \rho \leftsquigarrow \nu_{U}^{\Record{\ell_1: \sigma_1, \dots, \ell_m: \sigma_m, \ell_1': \sigma_1',\dots,\ell_n':\sigma_n'}}; \\
          & \score(w_1(\gamma_1', \rho|_{L_1}, U) \cdot w_2^{g_1(\gamma_1', \rho|_{L_1}, U)}(\gamma_2'^{g_1(\gamma_1', \rho|_{L_1}, U)}, \rho|_{L_2}, U \cup \names{\rho|_{L_1}})); \\
          & \ret{\rho}
        \end{aligned} \\
        & = \cdot \vdash \begin{aligned}[t]
          & \rho \leftsquigarrow \nu_{U}^{\Record{\ell_1: \sigma_1, \dots, \ell_m: \sigma_m, \ell_1': \sigma_1',\dots,\ell_n':\sigma_n'}}; \\
          & \score(w(\gamma, \rho, U)); \\
          & \ret{\rho}
        \end{aligned}
      \end{align*}
      The first equality unrolls the semantics of $x \gets s; t$. The second uses the inductive hypothesis to rewrite draws from $\mu_1$ and $\mu_2^v$ as compositions of the appropriate base measure and a score command, exploiting the fact that $w_1$ and $w_2$ are correct densities. The third equality uses the commutativity of the Qbs measure monad to move all score commands to the end, and then exploits the identity that $\score(r_1); \score(r_2) = \score(r_1 \cdot r_2)$ to combine them into one score command. Next, we use the definition of the reference measure for products (and thus records) to rewrite the three lines assigning $\rho_1$, $\rho_2$, and $\rho$ into a single assignment of $\rho$ from the joint reference measure (also relying on commutativity to bring the three lines together). Finally, we recognize the expression inside the $\score$ command as being precisely $w(\gamma, \rho, U)$. We conclude that $w(\gamma, -, U)$ is a density of $\mu$ w.r.t. $\nu_U^{\Record{\ell_1:\sigma_1,\dots,\ell_m:\sigma_m,\ell_1':\sigma_1',\dots,\ell_n':\sigma_n'}}$, as desired.
  \end{itemize}

  \paragraph{Case: Label}
  Suppose $\tyjudg{\Gamma}{t \at \ell}{\ProbTy[\toType{\Gamma|_{\freeVars{t}}}]{\Record{\ell: \sigma}}{\tau}}$
  and let $(\mu,f) = \Csem{t \at \ell}(\gamma_1)$ and $(\gamma, \langle g, w\rangle) = \Ssem{\GenToInc{t \at \ell}}(\gamma_2)$.
  By the inductive hypothesis,
  take
  \[
    ((\mu', f'), (\gamma', \langle g', w' \rangle))
    =
    (\Csem{t}(\gamma_1), \Ssem{\GenToInc{t}}(\gamma_2)) \in \logrel[{\ProbTy[\kappa]{\sigma}{\tau}}].
  \]
  \begin{itemize}
      \item Consider any $((s_1,n_1),(s_2,n_2)) \in \logrel[\Record{\ell:\sigma}\times\NameSet]$
      and let $(s_1', s_2') = (\pi_{(\ell:\sigma)}(s_1), \pi_{(\ell:\sigma)}(s_2)) \in \logrel[\sigma]$.
      It is immediate that $(f(s_1), g(\gamma, s_2,n_2)) = (f'(s_1'), g'(\gamma', s_2', n_2)) \in \logrel[\tau]$.
      \item Consider any $U \supseteq \supp{\mu}$.
      As with the previous case,
      we will work in the monadic metalanguage for the measure monad on $\qbs$:
      \begin{align*}
        \mu
        & = \cdot \vdash \begin{aligned}[t]
            & \rho' \leftsquigarrow \mu'; \\
            & \ret{\Record{\ell: \rho'}}
        \end{aligned}
        = \cdot \vdash \begin{aligned}[t]
            & \rho' \leftsquigarrow \nu_U^\sigma; \\
            & \score(w'(\gamma', \rho', U)); \\
            & \ret{\Record{\ell: \rho'}}
        \end{aligned}
        = \cdot \vdash \begin{aligned}[t]
            & \rho' \leftsquigarrow \nu_U^{\sigma}; \\
            & \rho \leftsquigarrow \ret{\Record{\ell: \rho'}}; \\
            & \score(w'(\gamma', \pi_{(\ell: \sigma)}(\rho), U)); \\
            & \ret{\rho}
        \end{aligned} \\
        & = \cdot \vdash \begin{aligned}[t]
            & \rho \leftsquigarrow (\rho' \leftsquigarrow \nu_U^{\sigma}; \ret{\Record{\ell: \rho'}}); \\
            & \score(w'(\gamma', \pi_{(\ell: \sigma)}(\rho), U)); \\
            & \ret{\rho}
        \end{aligned}
        = \cdot \vdash \begin{aligned}[t]
            & \rho \gets \nu_U^{\Record{\ell: \sigma}}; \\
            & \score(w(\gamma, \rho, U)); \\
            & \ret{\rho}
        \end{aligned}
      \end{align*}
        meaning $w$ is a density of $\mu$
        w.r.t. $\nu_U^{\Record{\ell:\sigma}}$.
  \end{itemize}

  \paragraph{Case: Loop (Simple)} We consider the special case
  \[
    T = \forInRange{x}{d}{b}
    \qquad
    \tyjudg{\Gamma}{d}{\Dist[\kappa_d]{\Nat}},
    \qquad
    \tyjudg{\Gamma,x:\Nat}{b}{\ProbTy[\kappa_b]{\sigma}{\tau}},
  \]
  with
  \[
    \tyjudg{\Gamma}{T}{\ProbTy[\toType{\Gamma|_{\freeVars{T}}}]{(\List{\sigma})}{(\List{\tau})}}.
  \]
  This is the specialization of the loop rule in \cref{fig:gen-to-inc-terms-full} with no \keyword{using} clauses and a trivial accumulator.
  In this case, the translation we are proving correct is:
  {\footnotesize
  \[
  \GenToInc{\forInRange{x}{d}{b}}
  =
  \left(\begin{aligned}
    \lam{(\tr,\var{U})}{}&
      \var{N} := \const{length}\;\tr;\;
      (\_, \var{w}_{\var{N}}) := \GenToInc{d}(\var{N},\var{U});\\
    & \var{ys} :=
      \forWith{x}{\const{range}\;\var{N}}{\var{V}}{\var{U}}
      { (r,w) := \GenToInc{b}(\tr[x],\var{V});\\
    & \qquad\qquad\qquad\qquad\quad
        ((r,w), \var{V}\cup\names{\tr[x]})};\\
    & (\var{rs},\var{ws}) := \const{unzip}\;\var{ys};\\
    & (\var{rs},\,\var{w}_{\var{N}}\cdot\const{product}\;\var{ws})
  \end{aligned}\right).
  \]}
  That is, the translated program computes the density of $d$ at the length of the trace, and then multiplies that by the densities of the body subtraces.
  Let
  \[
    ((\mu_d,!),
      (\gamma_d,\langle g_d,w_d\rangle))
    =
    (\Csem{d}(\gamma_1),\Ssem{\GenToInc{d}}(\gamma_2)) \in \logrel[{\ProbTy[\kappa_d]{\Nat}{\Unit}}].
  \]
  Instantiating the inductive hypothesis for the body at any loop index $i\in\Nat$, in the extended related environments $(\gamma_1, i)$ and $(\gamma_2, i)$, gives
  \[
    ((\mu_i,f_i),
      (\gamma_i,\langle g_i,w_i\rangle))
    =
    (\Csem{b}(\gamma_1,i),\Ssem{\GenToInc{b}}(\gamma_2,i))
    \in \logrel[{\ProbTy[\kappa_b]{\sigma}{\tau}}].
  \]

  \begin{itemize}[leftmargin=*]
    \item To show the return value function is correct, consider any
    $((\rho,U_1),(\rho',U_2))\in\logrel[\List{\sigma}\times\NameSet]$.
    Write $N=|\rho|=|\rho'|$, $\rho_i=\pi_i(\rho)$, and $\rho_i'=\pi_i(\rho')$.
    Define $V_0=U_1$, $V_0'=U_2$ and
    $V_i=V_{i-1}\cup\names{\rho_i}$, $V_i'=V_{i-1}'\cup\names{\rho_i'}$.
    By the list clause of the logical relation, $(\rho_i,\rho_i')\in\logrel[\sigma]$ for each $i$.
    Since $(V_{i-1},V_{i-1}')\in\logrel[\NameSet]$, applying the functional part of the body induction hypothesis at
    $((\rho_i,V_{i-1}),(\rho_i',V_{i-1}'))$ yields
    \[
      (f_i(\rho_i), g_i(\gamma_i,\rho_i',V_{i-1}'))\in\logrel[\tau].
    \]
    The source loop returns $(f_i(\rho_i))_{i=1}^N$, and the translated loop returns
    $(g_i(\gamma_i,\rho_i',V_{i-1}'))_{i=1}^N$.
    Hence the two results are related by $\logrel[\List{\tau}]$.

    \item To show the density correct, let $U\supseteq\supp{\mu}$, where $\mu$ is the trace distribution of $T$.
    We write $V_0=U$ and, for a length-$N$ trace $\rho=(\rho_1,\dots,\rho_N)$,
    \[
      V_i = V_{i-1}\cup\names{\rho_i}.
    \]
    By the inductive hypotheses,
    $w_d(\gamma_d,N,U)$ is the density of $\mu_d$ w.r.t. $\nu_U^{\Nat}$, and
    $w_i(\gamma_i,\rho_i,V_{i-1})$ is the density of $\mu_i$ with respect to $\nu_{V_{i-1}}^\sigma$.
    Therefore, in the monadic metalanguage,
    \begin{align*}
      \mu
      &=
      \cdot \vdash \begin{aligned}[t]
        & N \leftsquigarrow \mu_d; \\
        & \rho_1 \leftsquigarrow \mu_1; \ \dots \ ; \rho_N \leftsquigarrow \mu_N; \\
        & \ret{\iota_N(\rho_1,\dots,\rho_N)}
      \end{aligned} \\
      &=
      \cdot \vdash \begin{aligned}[t]
        & N \leftsquigarrow \nu_U^{\Nat}; \\
        & \rho_1 \leftsquigarrow \nu_U^\sigma; \\
        & \rho_2 \leftsquigarrow \nu_{V_1}^\sigma; \ \dots \ ;
          \rho_N \leftsquigarrow \nu_{V_{N-1}}^\sigma; \\
        & \score\!\left(
            w_d(\gamma_d,N,U)\cdot
            \prod_{i=1}^N w_i(\gamma_i,\rho_i,V_{i-1})
          \right); \\
        & \ret{\iota_N(\rho_1,\dots,\rho_N)}
      \end{aligned} \\
      &=
      \cdot \vdash \begin{aligned}[t]
        & \rho \leftsquigarrow \nu_U^{\List{\sigma}}; \\
        & \score(w(\gamma,\rho,U)); \\
        & \ret{\rho}.
      \end{aligned}
    \end{align*}
    The second equality uses the body and bound induction hypotheses and combines the resulting score commands.
    The last equality uses the definition
    $\nu_U^{\List{\sigma}}=\sum_{N\in\Nat}(\iota_N)_*\nu_U^{\sigma^N}$, where
    $\nu_U^{\sigma^N}$ is the iterated product measure that samples the $i$th component from $\nu_{V_{i-1}}^\sigma$.
    The score expression is exactly the weight computed by the specialized translation of $T$.
    Thus $w(\gamma,-,U)$ is a density of $\mu$ with respect to $\nu_U^{\List{\sigma}}$.
  \end{itemize}
\end{proof}

\clearpage 

\section{Incrementalizing Transformation: Full Definition and Correctness}
\label{sec:incrementalization-proof}

The full type and term translations are given in \cref{fig:inc-to-core-types-full,fig:inc-to-core-terms-full}, and the full version of the for-expression translation is given in \cref{fig:inc-to-core-terms-for-full}. The full definitions of the change-validity relation and logical relation are given in \cref{fig:change-rel-full,fig:logical-rel-full}.

\begin{figure}
  \begin{align*}
    \IncToCore{\sigma} &= \sigma \quad \text{ for ground types } \sigma \\
    \IncToCore{\sigma \times \tau} &= \IncToCore{\sigma} \times \IncToCore{\tau} \\
    \IncToCore{\List{\sigma}} &= \List{\IncToCore{\sigma}} \\
    \IncToCore{\NameMap{\sigma}} &= \NameMap{\IncToCore{\sigma}} \\
    \IncToCore{\Record{\ell_1: \tau_1, \dots, \ell_n: \tau_n}} &= \Record{\ell_1: \IncToCore{\tau_1}, \dots, \ell_n: \IncToCore{\tau_n}} \\
    \IncToCore{\sigma \clos[\kappa] \tau} &= \IncToCore{\sigma} \to \IncToCore{\tau} \times \Upd{\Changed{\kappa \times \sigma}}{\Changed{\tau}} \\
    \IncToCore{\Gamma} &= (x_1: \IncToCore{\sigma_1}, \dots, x_n: \IncToCore{\sigma_n}) \text{ for contexts } \Gamma = (x_1: \sigma_1, \dots, x_n: \sigma_n) \\[1ex]
    \Changed{\sigma} &= \Record{\lbl{new}: \sigma, \lbl{same}: \Bool} \quad \text{ for } \sigma = \Bool, \Nat, \Real, \Name \\
    \Changed{\sigma \times \tau} &= \Changed{\sigma} \times \Changed{\tau} \\
    \Changed{\List{\sigma}} &= \Record{\lbl{insert}: \List{(\Nat \times \IncToCore{\sigma})}, \lbl{remove}: \List{\Nat}, \lbl{change}: \List{(\Nat \times \Changed{\sigma})}} \\
    \Changed{\NameMap{\sigma}} &= \Record{\lbl{add}: \NameMap{\IncToCore{\sigma}}, \lbl{remove}: \NameSet, \lbl{change}: \NameMap{\Changed{\sigma}}} \\
    \Changed{\Record{\ell_1: \tau_1, \dots, \ell_n: \tau_n}} &= \SubRecord{\ell_1: \Changed{\tau_1}, \dots, \ell_n: \Changed{\tau_n}} \\
    \Changed{\sigma \clos[\kappa] \tau} &= \Changed{\kappa} \\
    \Changed{\Gamma} &= (x_1 : \Changed{\sigma_1}, \dots, x_n: \Changed{\sigma_n}) \text{ for contexts } \Gamma = (x_1: \sigma_1, \dots, x_n: \sigma_n)
  \end{align*}
  \caption{Incrementalizing transformation from $\inclang$ to $\corelang$ on types.}
  \label{fig:inc-to-core-types-full}
\end{figure}

\begin{figure}
  \begin{align*}
    \IncToCore[\Gamma]{x} &= (x, \mkUpd{\unit}{\lam{(d\gamma, \_: \Unit)} (\projectVar{\Changed{\Gamma}}{x} \; d\gamma, \unit)}) \\
    \IncToCore[\Gamma]{c} &= c_{\mathsf{core}} \\
    \IncToCore[\Gamma]{x := s; t} &=
      (x, u_x) := \IncToCore[\Gamma]{s}; \; (y, u_y) := \IncToCore[\Gamma, x: \sigma]{t}; \\
      &\quad (y, \mkUpd{(u_x, u_y)}{\lam{(d\gamma, (u_x, u_y))} \\
      &\qquad \qquad (\mvar{dx}, u_x') := \applyUpd{u_x}{d\gamma}; \; (\mvar{dy}, u_y') := \applyUpd{u_y}{(d\gamma, \mvar{dx})}; \; (\mvar{dy}, (u_x', u_y'))}) \\
    \IncToCore[\Gamma]{\lam{x} t} &=
      (\lam{x} \IncToCore[\Gamma|_{\freeVars{\lam{x} t}}, x: \sigma]{t}, \\
      &\quad \mkUpd{\unit}{\lam{(d\gamma, \_: \Unit)} (\projectVars{\Changed{\Gamma}}{\freeVars{\lam{x} t}} \; d\gamma, \unit)}) \\
    \IncToCore[\Gamma]{(s : \sigma \clos[\kappa] \tau) \; t} &=
      (f: \IncToCore{\sigma \clos[\kappa] \tau}, u_f: \Upd{\Changed{\toType{\Gamma}}}{\Changed{\kappa}}) := \IncToCore[\Gamma]{s}; \\
      &\quad (x: \IncToCore{\sigma}, u_x: \Upd{\Changed{\toType{\Gamma}}}{\Changed{\sigma}}) := \IncToCore[\Gamma]{t}; \\
      &\quad (y: \IncToCore{\tau}, u_y: \Upd{\Changed{\kappa \times \sigma}}{\Changed{\tau}}) := f \; x; \\
      &\quad (y, \mkUpd{(u_f, u_x, u_y)}{\lam{(d\gamma, (u_f, u_x, u_y))} \\
      &\qquad \qquad (\mvar{df}: \Changed{\kappa}, u_f') := \applyUpd{u_f}{d\gamma}; \\
      &\qquad \qquad (\mvar{dx}: \Changed{\sigma}, u_x') := \applyUpd{u_x}{d\gamma}; \\
      &\qquad \qquad (\mvar{dy}: \Changed{\tau}, u_y') := \applyUpd{u_y}{(\mvar{df}, \mvar{dx})}; \\
      &\qquad \qquad (\mvar{dy}, (u_f', u_x', u_y'))}) \\
    \IncToCore[\Gamma]{(s, t)} &=
      (x, u_x) := \IncToCore[\Gamma]{s}; \; (y, u_y) := \IncToCore[\Gamma]{t}; \\
      &\quad ((x, y), \mkUpd{(u_x, u_y)}{\lam{(d\gamma, (u_x, u_y))} \\
      &\qquad \qquad (\mvar{dx}, u_x') := \applyUpd{u_x}{d\gamma}; \; (\mvar{dy}, u_y') := \applyUpd{u_y}{d\gamma}; \; ((\mvar{dx}, \mvar{dy}), (u_x', u_y'))}) \\
    \IncToCore[\Gamma]{t.1} &=
      (x, u_x) := \IncToCore[\Gamma]{t}; \\
      &\quad (x.1, \mkUpd{u_x}{\lam{(d\gamma, u_x)} \\
      &\qquad \qquad (\mvar{dx}, u_x') := \applyUpd{u_x}{d\gamma}; \; (\mvar{dx}.1, u_x'.1)}) \\
    \IncToCore[\Gamma]{\recordLit{\ell_1: t_1, \dots, \ell_n: t_n}} &= (y_1, u_1) := \IncToCore[\Gamma]{t_1}; \dots; (y_n, u_n) := \IncToCore[\Gamma]{t_n}; \\
      &\quad (\recordLit{\ell_1: y_1, \dots, \ell_n: y_n}, \mkUpd{(u_1, \dots, u_n)}{\lam{(d\gamma, (u_1, \dots, u_n))} \\
      &\qquad \qquad (\mvar{dy}_1, u_1') := \applyUpd{u_1}{d\gamma}; \dots; (\mvar{dy}_n, u_n') := \applyUpd{u_n}{d\gamma}; \\
      &\qquad \qquad (\recordLit{\ell_1: \mvar{dy}_1, \dots, \ell_n: \mvar{dy}_n}, (u_1', \dots, u_n'))}) \\
    \IncToCore[\Gamma]{t.\ell : \tau} &=
      (r, u_r) := \IncToCore[\Gamma]{t}; \; (r.\ell, \mkUpd{u_r}{\lam{(d\gamma, u_r)} \\
      &\qquad \qquad (\mvar{dr}, u_r') := \applyUpd{u_r}{d\gamma}; \; (\ite{\has[\ell]{\mvar{dr}}}{\mvar{dr}.\ell}{\same}, u_r')}) \\
      &\text{ if } \same[\tau] \text{ exists} \\
    \IncToCore[\Gamma]{t.\ell : \tau} &=
      (r, u_r) := \IncToCore[\Gamma]{t}; \; (r.\ell, \mkUpd{(r.\ell, u_r)}{\lam{(d\gamma, (r_l, u_r))} \\
      &\qquad \qquad (\mvar{dr}, u_r') := \applyUpd{u_r}{d\gamma}; \\
      &\qquad \qquad \ite{\has[\ell]{\mvar{dr}}}{(\mvar{dr}.\ell, (\apply(r_l, \mvar{dr}.\ell), u_r'))}{(\sameAs(r_l), (r_l, u_r'))}}) \\
      &\text{ otherwise, if } \sameAs[\tau], \apply[\tau] \text{ exist} \\
    \IncToCore[\Gamma]{\restrict{t}{\{\ell_1, \dots, \ell_k\}}} &=
      (r, u_r) := \IncToCore[\Gamma]{t}; \\
      &\quad (\restrict{r}{\{\ell_1, \dots, \ell_k\}}, \mkUpd{u_r}{\lam{(d\gamma, u_r)} \\
      &\qquad \qquad (\mvar{dr}, u_r') := \applyUpd{u_r}{d\gamma}; \; (\restrict{\mvar{dr}}{\{\ell_1, \dots, \ell_k\}}, u_r')})
  \end{align*}
  \caption{Incrementalizing transformation from $\inclang$ to $\corelang$ on terms.
  It satisfies the invariant that $\tyjudg{\Gamma}{t}{\tau}$ in $\inclang$ implies $\tyjudg{\IncToCore{\Gamma}}{\IncToCore[\Gamma]{t}}{\IncToCore{\tau} \times \Upd{\Changed{\toType{\Gamma}}}{\Changed{\tau}}}$ in $\corelang$.}
  \label{fig:inc-to-core-terms-full}
\end{figure}

\begin{figure}
  \footnotesize
  \setlength{\jot}{1pt}
  \hspace*{\fill}
  \begin{minipage}[t]{0.49\linewidth}
  \noindent\\[-1em]
  \(\begin{aligned}
    &\IncToCore[\Gamma]{\begin{aligned}
        &\forWithUsing{x : \tau}{\mvar{xs}}{z}{s : \rho \\
        &}{v_1 := y_1[u_1] : \theta_1, \dots, \\
        &\quad v_n := y_n[u_n] : \theta_n}{t : \tau' \times \rho}
      \end{aligned}} = \\
    &(\var{xs}, \var{xsUpd}) := \IncToCore{\mvar{xs}}; \\
    &(\var{z}, \var{seedUpd}) := \IncToCore{s}; \\
    &\var{seeds} := [\var{z}]; \, \var{results} := [\,]; \, \var{updaters} := [\,]; \\
    &\var{idxUpdaters}_{\var{j}} := [\,] \text{ for } \var{j} = 1..n; \\
    &\var{itersToIndices}_{\var{j}} := \emptyset \text{ for } \var{j} = 1..n; \\
    &\var{indicesToIters}_{\var{j}} := \emptyset \text{ for } \var{j} = 1..n; \\
    &\textbf{for } \var{i} \textbf{ from } 1 \textbf{ to } \const{length}(\var{xs}) \textbf{ do} \\
    &\quad \var{x} := \var{xs}[\var{i}]; \\
    &\quad \textbf{for } \var{j} \textbf{ from } 1 \textbf{ to } n \textbf{ do} \\
    &\qquad (\var{idx}_{\var{j}}, \var{idxUpd}_{\var{j}}) := \IncToCore{u_{\var{j}}}; \\
    &\qquad (\var{collections}_{\var{j}}, \var{collectionUpds}_{\var{j}}) := \IncToCore{y_{\var{j}}}; \\
    &\qquad v_{\var{j}} := \var{collections}_{\var{j}}[\var{idx}_{\var{j}}]; \\
    &\qquad \const{push}(\var{idxUpdaters}_{\var{j}}, \var{idxUpd}_{\var{j}}); \\
    &\qquad \var{itersToIndices}_{\var{j}}[\var{i}] := \var{idx}_{\var{j}}; \\
    &\qquad \var{indicesToIters}_{\var{j}}[\var{idx}_{\var{j}}] := \var{indicesToIters}_{\var{j}}[\var{idx}_{\var{j}}] \cup \{ \var{i} \}; \\
    &\quad ((\var{y}, \var{z}'), \var{upd}) := \IncToCore{t}; \\
    &\quad \const{push}(\var{seeds}, \var{z}'); \; \const{push}(\var{results}, \var{y}); \\
    &\quad \const{push}(\var{updaters}, \var{upd}); \; \var{z} := \var{z}'; \\
    &\var{cache} := \recordLit{ \field{inputs} \var{xs}, \field{xsUpd} \var{xsUpd}, \\
    &\quad  \field{seeds} \var{seeds}, \field{seedUpd} \var{seedUpd}, \\
    &\quad \field{collections} \var{collections}, \field{collectionUpds} \var{collectionUpds}, \\
    &\quad \field{idxUpdaters} \var{idxUpdaters}, \field{itersToIndices} \var{itersToIndices}, \\
    &\quad \field{indicesToIters} \var{indicesToIters}, \field{updaters} \var{updaters} }; \\
    &(\var{results}, \mkUpd{\var{cache}}{\var{update}})
  \end{aligned}\)
  \\[1ex]

  \noindent
  \(\begin{aligned}
    &\CacheType = \Record{ \\
    &\quad \lbl{inputs}: \List{\tau}, \\
    &\quad \lbl{xsUpd}: \Upd{\Changed{\toType{\Gamma}}}{\Changed{\List{\tau}}}, \\
    &\quad \lbl{seeds}: \List{\rho}, \\
    &\quad \lbl{seedUpd}: \Upd{\Changed{\toType{\Gamma}}}{\Changed{\rho}}, \\
    &\quad \lbl{collections}: \Collection{\theta_1} \times \dots \times \Collection{\theta_n}, \\
    &\quad \lbl{collectionUpds}: \Upd{\Changed{\toType{\Gamma}}}{\Changed{\Collection{\theta_1}}} \times \\
    &\qquad \cdots \times \Upd{\Changed{\toType{\Gamma}}}{\Changed{\Collection{\theta_n}}}, \\
    &\quad \lbl{itersToIndices}: \Map{\Nat}{\Index_1} \times \cdots \\
    &\qquad \times \Map{\Nat}{\Index_n}, \\
    &\quad \lbl{idxUpdaters}: \List{\Upd{\Changed{\toType{\Gamma_{\mvar{lkp}}}}}{\Changed{\Index_1}}} \times \\
    &\qquad \cdots \times \List{\Upd{\Changed{\toType{\Gamma_{\mvar{lkp}}}}}{\Changed{\Index_n}}}, \\
    &\quad \lbl{indicesToIters}: \Map{\Index_1}{(\Set{\Nat})} \times \cdots \\
    &\qquad \times \Map{\Index_n}{(\Set{\Nat})}, \\
    &\quad \lbl{updaters}: \List{\Upd{\Changed{\toType{\Gamma_{\mvar{ext}}}}}{\Changed{\tau' \times \rho}}} \\
    &}
  \end{aligned}\)\\
  \text{where each } $\Collection{\theta_j}$ is a list or name map; $\Index_j$ is $\Nat$ or $\Name$ accordingly; $\Gamma_{\mvar{lkp}} = \Gamma, x{:}\tau, z{:}\rho$; $\Gamma_{\mvar{ext}} = \Gamma_{\mvar{lkp}}, v_1{:}\theta_1, \dots, v_n{:}\theta_n$.
  \end{minipage}%
  \hfill%
  \begin{minipage}[t]{0.48\linewidth}
  \noindent\\[-1em]
  \(\begin{aligned}
    &\textbf{\textit{update}}\; (d\gamma: \Changed{\toType{\Gamma}}, \ \var{cache} : \CacheType) := \\
    &\var{inputs} := \var{cache}.\lbl{inputs}; \dots; \var{updaters} := \var{cache}.\lbl{updaters}; \\
    &(\var{dxs}, \var{xsUpd}) := \applyUpd{\var{xsUpd}}{d\gamma}; \\
    &(\var{dz}, \var{seedUpd}) := \applyUpd{\var{seedUpd}}{d\gamma}; \\
    &\textbf{for } \var{j} = 1..n \textbf{ do } \\
    &\quad (\var{dy}_{\var{j}}, \var{collectionUpds}_{\var{j}}) := \applyUpd{\var{collectionUpds}_{\var{j}}}{d\gamma}; \\
    &\quad \var{collections}'_{\var{j}} := \apply(\var{collections}_{\var{j}}, \var{dy}_{\var{j}}); \\
    &\var{seeds}[1] := \apply(\var{seeds}[1], \var{dz}); \\
    &\var{toVisit} := \{ \var{i} \mid (\var{i}, \_) \in \var{dxs}.\lbl{change} \}; \\
    &\textbf{for } \var{j} = 1..n \textbf{ do } \\
    &\quad \textbf{for each } (\var{idx}, {\,}\_)\! \in \var{dy}_{\var{j}}.\lbl{change} \ \textbf{do} \\
    &\qquad \var{toVisit} := \var{toVisit} \cup \var{indicesToIters}_{\var{j}}[\var{idx}]; \\
    &\textbf{if } d\gamma \text{ changes any $\Gamma$-variable free in $t$ or any $u_k$} \textbf{ then }  \\
    &\quad\var{toVisit} := \{1,\dots,\const{length}(\var{inputs})\}; \\
    &\var{changed} := \emptyset; \quad \var{i} := 1; \\
    &\textbf{while } \var{i} \le \const{length}(\var{inputs}) \textbf{ do } \\
    &\quad \textbf{if } \issame{\var{dz}} \land \var{i} \notin \var{toVisit} \textbf{ then } \\
    &\qquad \textbf{if } \var{toVisit} = \emptyset \textbf{ then break}; \\
    &\qquad \var{i} := \text{pop min from } \var{toVisit}; \\
    &\qquad \var{dz} := \sameAs(\var{seeds}[\var{i}]); \\
    &\quad \textbf{if } (\var{i}, \mvar{dx}) \in \var{dxs}.\lbl{change} \textbf{ then } \var{dx} := \mvar{dx} \\
    &\quad \textbf{else } \var{dx} := \sameAs(\var{inputs}[\var{i}]); \\
    &\quad \textbf{for } \var{j} = 1..n \textbf{ do } \\
    &\qquad (\var{didx}_{\var{j}}, \var{idxUpdaters}_{\var{j}}[\var{i}]) := \\
    &\qquad \qquad \applyUpd{\var{idxUpdaters}_{\var{j}}[\var{i}]}{(d\gamma, \var{dx}, \var{dz})}; \\
    &\qquad \textbf{if } \issame{\var{didx}_{\var{j}}} \textbf{ then } \\
    &\qquad\quad \var{idx} := \var{didx}_{\var{j}}.\lbl{new}; \\
    &\qquad\quad \textbf{if } (\var{idx}, \var{dval}) \in \var{dy}_{\var{j}}.\lbl{change} \textbf{ then } \var{dv}_{\var{j}} := \var{dval} \\
    &\qquad\quad \textbf{else } \var{dv}_{\var{j}} := \sameAs(\var{collections}_{\var{j}}[\var{idx}]) \\
    &\qquad \textbf{else} \\
    &\qquad\quad \var{old} := \var{itersToIndices}_{\var{j}}[\var{i}]; \ \var{new} := \var{didx}_{\var{j}}.\lbl{new}; \\
    &\qquad\quad \var{itersToIndices}_{\var{j}}[\var{i}] := \var{new}; \\
    &\qquad\quad \var{indicesToIters}_{\var{j}}[\var{old}] := \var{indicesToIters}_{\var{j}}[\var{old}] \setminus \{ \var{i} \}; \\
    &\qquad\quad \var{indicesToIters}_{\var{j}}[\var{new}] := \\
    &\qquad\quad\qquad \var{indicesToIters}_{\var{j}}[\var{new}] \cup \{ \var{i} \}; \\
    &\qquad\quad \var{dv}_{\var{j}} := \text{change from } \var{collections}_{\var{j}}[\var{old}] \\
    &\qquad\qquad \text{to } \var{collections}'_{\var{j}}[\var{new}] \\
    &\quad ((\var{dy}, \var{dz}), \var{updaters}[\var{i}]) := {} \\
    &\qquad \applyUpd{\var{updaters}[\var{i}]}{(d\gamma, \var{dx}, \var{dz}, \var{dv}_1, \dots, \var{dv}_n)}; \\
    &\quad \var{seeds}[\var{i}+1] := \apply(\var{seeds}[\var{i}+1], \var{dz}); \\
    &\quad \var{inputs}[\var{i}] := \apply(\var{inputs}[\var{i}], \var{dx}); \\
    &\quad \textbf{if } \lnot \issame{\var{dy}} \textbf{ then } \var{changed}[\var{i}] := \var{dy}; \\
    &\quad \var{toVisit} := \var{toVisit} \setminus \{ \var{i} \}; \\
    &\quad \var{i} := \var{i} + 1; \\
    &\var{dresults} := \{ \field{change} \var{changed} \}; \var{collections} := \var{collections}'; \\
    &\var{cache}.\lbl{inputs} := \var{inputs}; \, \dots; \, \var{cache}.\lbl{updaters} := \var{updaters}; \\
    &(\var{dresults}, \var{cache})
  \end{aligned}\)
  \end{minipage}%
  \hfill
  \caption{Incrementalizing transformation from $\inclang$ to $\corelang$ on for-expressions.
  For simplicity, the figure only handles modifications of list elements, not insertions or deletions, and we further assume that each collection expression $y_j$ does not depend on the loop variable $x$ or the accumulator $z$, so it can be typed in just $\Gamma$, matching the cache type declared for $\var{collectionUpds}_j$.}
  \label{fig:inc-to-core-terms-for-full}
\end{figure}

\begin{figure}
  \begin{align*}
    \changerel[\sigma] &= \{ (x, \mvar{dx}, x') \mid \mvar{dx}(\lbl{new}) = x' \land (\mvar{dx}(\lbl{same}) =  \trueLit \implies x = x') \} \\
    &\quad \text{ for } \sigma = \Bool, \Nat, \Real, \Name \\
    \changerel[\sigma \times \tau] &= \{ ((x, y), (\mvar{dx}, \mvar{dy}), (x', y')) \mid (x, \mvar{dx}, x') \in \changerel[\sigma] \land (y, \mvar{dy}, y') \in \changerel[\tau] \} \\
    \changerel[\List{\sigma}] &= \{ (\mvar{xs}, \recordLit{\field{insert} \mvar{is}, \field{remove} \mvar{rs}, \field{change} \mvar{cs}}, \mvar{xs}') \\
    &\qquad \mid \mvar{xs}'' = (\mvar{xs} \text{ with all indices in } \mvar{rs} \text{ removed}); \\
    &\qquad \phantom{{}\mid{}} \mvar{xs}''' = (\mvar{xs}'' \text{ with changes in } \mvar{cs} \text{ applied}), \\
    &\qquad \phantom{{}\mid{}} \text{ i.e., for all } i, \begin{cases}
      (\mvar{xs}''[i], \mvar{dx}, \mvar{xs}'''[i]) \in \changerel[\sigma] & \text{for } (i, \mvar{dx}) \in \mvar{cs} \\
      \mvar{xs}'''[i] = \mvar{xs}''[i] & \text{otherwise};
    \end{cases} \\
    &\qquad \phantom{{}\mid{}} \mvar{xs}' = (\mvar{xs}''' \text{ with } x \text{ inserted at index $i$ for all } (i, x) \in \mvar{is}) \} \\
    \changerel[\NameMap{\sigma}] &= \{ (m, \recordLit{\field{change} \mvar{cm}, \field{remove} \mvar{rm}, \field{add} \mvar{am}}, m') \\
    &\qquad \mid \forall (n \mapsto \mvar{dv}) \in \mvar{cm} \ldotp (m(n), \mvar{dv}, m'(n)) \in \changerel[\sigma] \\
    &\qquad \land \forall (n \mapsto v) \in \mvar{am} \ldotp n \notin m \land m'(n) = v \\
    &\qquad \land \forall n \in \mvar{rm} \ldotp n \in m \land n \notin m' \\
    &\qquad \land \forall n \in \domain(m) \setminus (\domain(\mvar{cm}) \cup \mvar{rm}) \ldotp n \in \domain(m') \land m(n) = m'(n) \} \\
    \changerel[\Record{\ell_1: \tau_1, \dots, \ell_n: \tau_n}] &= \{ (r, \mvar{dr}, r') \mid \forall i \in \{1, \dots, n\} \ldotp \text{if } \ell_i \text{ in } \mvar{dr} \text{ then } (r(\ell_i), \mvar{dr}(\ell_i), r'(\ell_i)) \in \changerel[\tau_i] \text{ else } r(\ell_i) = r'(\ell_i) \} \\
    \changerel[{\sigma \clos[\kappa] \tau}] &= \{ ((e, f), \mvar{de}, (e', f')) \in \sem{\sigma \clos[\kappa] \tau} \times \sem{\Changed{\kappa}} \times \sem{\sigma \clos[\kappa] \tau} \mid (e, \mvar{de}, e') \in \changerel[\kappa], f = f' \} \\
    \changerel[\Gamma] &= \{ (\gamma, d\gamma, \gamma') \in \sem{\Gamma} \times \sem{\Changed{\Gamma}} \times \sem{\Gamma} \mid \forall (x: \sigma) \in \Gamma \ldotp (\gamma(x), d\gamma(x), \gamma'(x)) \in \changerel[\sigma] \}
  \end{align*}
  \caption{Relation of valid changes.}
  \label{fig:change-rel-full}
\end{figure}

\begin{figure}
  \begin{align*}
    \logrel[\sigma] &= \{ (x, x') \mid x = x' \} \quad \text{ for ground types } \sigma \\
    \logrel[\sigma \times \tau] &= \{ ((x, y), (x', y')) \mid (x, x') \in \logrel[\sigma] \land (y, y') \in \logrel[\tau] \} \\
    \logrel[\List{\sigma}] &= \{ (\mvar{xs}, \mvar{xs}') \mid \mathrm{length}(\mvar{xs}) = \mathrm{length}(\mvar{xs}') \land \forall i < \mathrm{length}(\mvar{xs}) \ldotp (\mvar{xs}[i], \mvar{xs}'[i]) \in \logrel[\sigma] \} \\
    \logrel[\NameMap{\sigma}] &= \{ (m, m') \mid \mathrm{keys}(m) = \mathrm{keys}(m') \land \forall n \in \mathrm{keys}(m) \ldotp (m(n), m'(n)) \in \logrel[\sigma] \} \\
    \logrel[\Record{\ell_1: \tau_1, \dots, \ell_n: \tau_n}] &= \{ (r, r') \mid \forall i \in \{1, \dots, n\} \ldotp (r(\ell_i), r'(\ell_i)) \in \logrel[\tau_i] \} \\
    \logrel[{\sigma \clos[\kappa] \tau}] &= \{ ((\mvar{env}, f), f') \mid \forall (x, x') \in \logrel[\sigma] \ldotp (y, y') \in \logrel[\tau] \land u \in \ValidUpd[\kappa \times \sigma \leadsto \tau]{f}{(\mvar{env}, x)} \\
    &\qquad \text{ where } y = f(\mvar{env}, x), (y', u) = f'(x') \} \\
    \logrel[\Gamma] &= \{ (\gamma, \gamma') \mid \forall (x: \sigma) \in \Gamma \ldotp (\gamma(x), \gamma'(x)) \in \logrel[\sigma] \}
  \end{align*}
  \caption{Logical relation between $\inclang$ and $\corelang$.}
  \label{fig:logical-rel-full}
\end{figure}

For this section, we use the context tuple conversion functions defined in \cref{fig:context-tuple-conversion}.
We require some properties about contexts and context types.

\begin{lemma}[Properties of context types]
  \label{lem:context-types}
  Let $\Gamma$ be a context, $\gamma \in \sem{\Gamma}$.
  Then for all $x: \sigma$ in $\Gamma$, we have $\gamma(x) = \sem{\projectVar{\Gamma}{x}}(\ctxToTpl[\Gamma](\gamma))$.
  Similarly, for all $\mvar{xs} \subseteq \freeVars{t}$, we have $\ctxToTpl[\Gamma|_{\mvar{xs}}](\gamma|_{\mvar{xs}}) = \sem{\projectVars{\Gamma}{\mvar{xs}}}(\ctxToTpl[\Gamma](\gamma))$.
  As a consequence, we have $(\gamma, \gamma') \in \logrel[\Gamma]$ if and only if $(\ctxToTpl[\Gamma](\gamma), \ctxToTpl[\Gamma](\gamma')) \in \logrel[\toType{\Gamma}]$.
  Furthermore, $\tplToCtx[\Gamma]$ and $\ctxToTpl[\Gamma]$ are inverses of each other.
\end{lemma}
\begin{proof}
  By induction on the structure of $\Gamma$.
\end{proof}

\begin{lemma}[Valid Updater Construction]
  \label{lem:valid-updater-construction}
  Let $f : \sem{\tau} \to \sem{\tau'}$ with cache type $\sigma$, and update function $u \in \sem{\Changed{\tau} \times \sigma \to \Changed{\tau'} \times \sigma}$.
  For each $x \in \sem{\tau}$ let $A_x \subseteq \sem{\sigma}$ be a set (the ``admissible caches at $x$'').
  Assume ``cache preservation'' holds, i.e., for every valid change $(x, \mvar{dx}, x') \in \changerel[\tau]$ and every $c \in A_x$, if
  \[
  y = f(x), \quad y' = f(x'), \quad (\mvar{dy}, c') = u(\mvar{dx}, c),
  \]
  then $(y, \mvar{dy}, y') \in \changerel[\tau']$ and $c' \in A_{x'}$.
  Then for any admissible initial cache $c \in A_x$, the updater built by $\mkUpdSem{c}{u}$ (cf. \cref{fig:term-semantics-full}) is valid for $f$ at $x$:
  \[
  \mkUpdSem(c, u) \in \ValidUpd[\tau \leadsto \tau']{f}{x}.
  \]
\end{lemma}
\begin{proof}
  Construct the set $U_x^f := \{ \mkUpdSem(c, u) \mid c \in A_x \}$ of updaters built from caches in $A_x$.
  Prove by definition of $\mkUpdSem$ (\cref{fig:term-semantics-full}) and by the assumption that for every valid change $(x, \mvar{dx}, x') \in \changerel[\tau]$ and every $\mkUpdSem(c, u) \in U_x^f$, we have $(y, \mvar{dy}, y') \in \changerel[\tau']$ and $\mkUpdSem(c', u) \in U_{x'}^f$ where $y = f(x), y' = f(x'), (\mvar{dy}, c') = u(\mvar{dx}, c)$.
  This proves that $U_x^f$ is a post-fixpoint of the same operator used to define $\ValidUpd[\tau \leadsto \tau']{f}{x}$, and thus $U_x^f \subseteq \ValidUpd[\tau \leadsto \tau']{f}{x}$.
  Thus $\mkUpdSem(c, u) \in U_x^f \subseteq \ValidUpd[\tau \leadsto \tau']{f}{x}$ for any $c \in A_x$.
\end{proof}

We abbreviate $\Ione{t} := \pi_1 \circ \sem{\IncToCore{t}}$, $\semTuple{t} := \sem{t} \circ \tplToCtx[\Gamma]$, and $\Itwo{t} := \pi_2 \circ \sem{\IncToCore{t}}$.

\begin{lemma}[Fundamental Lemma]
  Assume $\tyjudg{\Gamma}{t}{\tau}$ in $\inclang$, and thus $\tyjudg{\IncToCore{\Gamma}}{\IncToCore[\Gamma]{t}}{\IncToCore{\tau} \times \Upd{\Changed{\toType{\Gamma}}}{\Changed{\tau}}}$ in $\corelang$.
  Then for all $(\gamma, \gamma_*) \in \logrel[\Gamma]$, we have
\[ (\sem{t}(\gamma), \Ione{t}(\gamma_*)) \in \logrel[\tau] \quad \text{ and } \quad \Itwo{t}(\gamma_*) \in \ValidUpd[\toType{\Gamma} \leadsto \tau]{\semTuple{t}}{\ctxToTpl[\Gamma](\gamma)} \]
\end{lemma}
\begin{proof}
  By induction on the typing derivation of $\tyjudg{\Gamma}{t}{\tau}$.
  Let $(\gamma, d\gamma, \gamma') \in \changerel[\Gamma]$ be a valid context change and let $\hat\gamma = \ctxToTpl[\Gamma](\gamma)$, $\hat\gamma' = \ctxToTpl[\Gamma](\gamma')$, and $\mvar{d\hat\gamma} = \ctxToTpl[\Changed{\Gamma}](d\gamma)$ be the corresponding tuple values.

  \paragraph{High-level structure.}
  Every case follows the same two-step pattern.
  First, the relational part is shown compositionally: using the induction hypotheses on subterms together with the appropriate clause of the logical relation, we derive
  $ (\sem{t}(\gamma), \Ione{t}(\gamma_*)) \in \logrel[\tau]$.
  Second, for updater correctness we specify a family of admissible caches $A_{\hat\gamma}$ and a single-step function $u$ that threads the input change $\mvar{d\hat\gamma}$ through the cached sub-updaters, producing an output change and a new cache.
  The only nontrivial check is cache preservation, which reduces to the induction hypotheses and the rules of the change relation; once this holds, \cref{lem:valid-updater-construction} yields
  $\Itwo{t}(\gamma_*) \in \ValidUpd[\toType{\Gamma} \leadsto \tau]{\semTuple{t}}{\hat\gamma}$.
  We handle each case below.

  \paragraph{Variables.}
  Let $(x : \tau) \in \Gamma$.
  Then $\IncToCore{x} = (x, \mkUpd{\unit}{\lam{(d\gamma, \_)} (\projectVar{\Changed{\Gamma}}{x} \; d\gamma, \unit)})$.
  So \[(\sem{x}(\gamma), \Ione{x}(\gamma_*)) = (\gamma(x), \gamma_*(x)) \in \logrel[\tau]\] since $(\gamma, \gamma_*) \in \logrel[\Gamma]$.

  Updater correctness.
  Let $u = \lam{(d\gamma, \_)} (\projectVar{\Changed{\Gamma}}{x} \; d\gamma, \unit)$.
  Then $y := \semTuple{x}(\hat\gamma) = \sem{x}(\gamma) = \gamma(x)$ and $y' := \semTuple{x}(\hat\gamma') = \sem{x}(\gamma') = \gamma'(x)$.
  Also $\mvar{dy} = \pi_1(\sem{u}(\mvar{d\hat\gamma}, \unit)) = \sem{\projectVar{\Changed{\Gamma}}{x}}(\mvar{d\hat\gamma}) = d\gamma(x)$, where the last equality is \cref{lem:context-types} applied to $\Changed{\Gamma}$ and $d\gamma$.
  Thus $(y, \mvar{dy}, y') \in \changerel[\tau]$ by the assumption that $(\gamma, d\gamma, \gamma') \in \changerel[\Gamma]$.
  By \cref{lem:valid-updater-construction}, $\Itwo{x}(\gamma_*) = \sem{\mkUpd{\unit}{u}} \in \ValidUpd[\toType{\Gamma} \leadsto \tau]{\semTuple{x}}{\hat\gamma}$.

  \paragraph{Constants.}
  The incrementalization of constants is proven on a case-by-case basis depending on the constant, typically using \cref{lem:valid-updater-construction}.

  \paragraph{Let bindings.}
  Let $T = (x := s;\, t)$ with $\tyjudg{\Gamma}{s}{\sigma}$ and $\tyjudg{\Gamma, x: \sigma}{t}{\tau}$.
  Then by induction hypothesis, we have $(\sem{s}(\gamma), \Ione{s}(\gamma_*)) \in \logrel[\sigma]$ and $\Itwo{s}(\gamma_*) \in \ValidUpd[\toType{\Gamma} \leadsto \sigma]{\semTuple{s}}{\hat\gamma}$.
  Let $\gamma_1 = \gamma[x \mapsto \sem{s}(\gamma)]$ and $\gamma_{1*} = \gamma_*[x \mapsto \Ione{s}(\gamma_*)]$ be the extended contexts.
  By the definition of the logical relation for contexts, we have $(\gamma_1, \gamma_{1*}) \in \logrel[\Gamma, x: \sigma]$.
  Thus by induction hypothesis again, we have $(\sem{t}(\gamma_1), \Ione{t}(\gamma_{1*})) \in \logrel[\tau]$ and $\Itwo{t}(\gamma_{1*}) \in \ValidUpd[\toType{\Gamma, x: \sigma} \leadsto \tau]{\semTuple{t}}{\hat\gamma_1}$ where $\hat\gamma_1 = \ctxToTpl[\Gamma, x: \sigma](\gamma_1) = (\hat\gamma, \semTuple{s}(\hat\gamma))$.
  The first part now follows from $\sem{T}(\gamma) = \sem{t}(\gamma_1)$ and $\Ione{T}(\gamma_*) = \Ione{t}(\gamma_{1*})$.

  Updater correctness.
  Let $u_x = \Itwo{s}(\gamma_*)$ and $u_y = \Itwo{t}(\gamma_{1*})$.
  Then $y = \semTuple{T}(\hat\gamma) = \sem{T}(\gamma) = \sem{t}(\gamma_1)$ and $y' = \semTuple{T}(\hat\gamma') = \sem{t}(\gamma_1')$ where $\gamma_1' = \gamma'[x \mapsto \sem{s}(\gamma')]$.
  To apply \cref{lem:valid-updater-construction}, we define admissible caches
  \[ A_{\hat\gamma} := \{ (u_x, u_y) \mid u_x \in \ValidUpd[\toType{\Gamma} \leadsto \sigma]{\semTuple{s}}{\hat\gamma}, u_y \in \ValidUpd[\toType{\Gamma, x: \sigma} \leadsto \tau]{\semTuple{t}}{(\hat\gamma, \semTuple{s}(\hat\gamma))} \}\]
  The cache update function is
  \[ u := \lam{(d\gamma, (u_x, u_y))} (\mvar{dx}, u_x') := \applyUpd{u_x}{d\gamma}; (\mvar{dy}, u_y') := \applyUpd{u_y}{(d\gamma, \mvar{dx})}; (\mvar{dy}, (u_x', u_y')) \]
  Stepping $u_x$ with $\mvar{d\hat\gamma}$ yields $(\mvar{dv}, u_x')$ with $(v, \mvar{dv}, v') \in \changerel[\sigma]$ and $u_x' \in \ValidUpd[\toType{\Gamma} \leadsto \sigma]{\semTuple{s}}{\hat\gamma'}$, where $v = \sem{s}(\gamma)$ and $v' = \sem{s}(\gamma')$.
  Stepping $u_y$ with $(\mvar{d\hat\gamma}, \mvar{dv})$ yields valid $(\mvar{dy}, u_y')$, i.e.\ $(y, \mvar{dy}, y') \in \changerel[\tau]$ and $u_y' \in \ValidUpd[\toType{\Gamma, x: \sigma} \leadsto \tau]{\semTuple{t}}{\hat\gamma_1'}$, where $y = \sem{t}(\gamma_1)$ and $y' = \sem{t}(\gamma_1')$.
  Then by the definition of the semantics, we have $y = \sem{T}(\gamma)$ and $y' = \sem{T}(\gamma')$.
  By unfolding, we find $\sem{u}(\mvar{d\hat\gamma}, (u_x, u_y)) = (\mvar{dy}, (u_x', u_y'))$.
  Thus we have $(y, \mvar{dy}, y') \in \changerel[\tau]$ and $(u_x', u_y') \in A_{\hat\gamma'}$.
  By \cref{lem:valid-updater-construction}, we have $\Itwo{T}(\gamma_*) = \mkUpdSem((u_x, u_y), u) \in \ValidUpd[\toType{\Gamma} \leadsto \tau]{\semTuple{T}}{\hat\gamma}$.

  \paragraph{Pairs.}
  Let $T = (s, t)$ with $\tyjudg{\Gamma}{s}{\sigma}$ and $\tyjudg{\Gamma}{t}{\tau}$.
  By induction hypothesis, the translations of $s$ and $t$ are valid: $(\sem{s}(\gamma), \Ione{s}(\gamma_*)) \in \logrel[\sigma]$ and $(\sem{t}(\gamma), \Ione{t}(\gamma_*)) \in \logrel[\tau]$.
  Since $\sem{T}(\gamma) = (\sem{s}(\gamma), \sem{t}(\gamma))$ and $\Ione{T}(\gamma_*) = (\Ione{s}(\gamma_*),\allowbreak \Ione{t}(\gamma_*))$, the definition of $\logrel[\sigma \times \tau]$ implies $(\sem{T}(\gamma), \Ione{T}(\gamma_*)) \in \logrel[\sigma \times \tau]$.

  Updater correctness.
  By induction hypothesis, $\Itwo{s}(\gamma_*) \in \ValidUpd[\toType{\Gamma} \leadsto \sigma]{\semTuple{s}}{\hat\gamma}$ and $\Itwo{t}(\gamma_*) \in \ValidUpd[\toType{\Gamma} \leadsto \tau]{\semTuple{t}}{\hat\gamma}$.
  Define admissible caches
  \[ A_{\hat\gamma} := \{ (u_x, u_y) \mid u_x \in \ValidUpd[\toType{\Gamma} \leadsto \sigma]{\semTuple{s}}{\hat\gamma},\, u_y \in \ValidUpd[\toType{\Gamma} \leadsto \tau]{\semTuple{t}}{\hat\gamma} \} \]
  The cache update function is
  \[ u := \lam{(d\gamma, (u_x, u_y))} (\mvar{dx}, u_x') := \applyUpd{u_x}{d\gamma};\; (\mvar{dy}, u_y') := \applyUpd{u_y}{d\gamma};\; ((\mvar{dx}, \mvar{dy}), (u_x', u_y')) \]
  Stepping $u_x$ and $u_y$ with $\mvar{d\hat\gamma}$ yields valid $(\mvar{dx}, u_x')$ and $(\mvar{dy}, u_y')$, so we have $(\sem{s}(\gamma), \mvar{dx}, \sem{s}(\gamma')) \in \changerel[\sigma]$, $u_x' \in \ValidUpd[\toType{\Gamma} \leadsto \sigma]{\semTuple{s}}{\hat\gamma'}$, $(\sem{t}(\gamma), \mvar{dy}, \sem{t}(\gamma')) \in \changerel[\tau]$, and $u_y' \in \ValidUpd[\toType{\Gamma} \leadsto \tau]{\semTuple{t}}{\hat\gamma'}$.
  As a consequence $(\sem{T}(\gamma), (\mvar{dx}, \mvar{dy}), \sem{T}(\gamma')) \in \changerel[\sigma \times \tau]$ and $(u_x', u_y') \in A_{\hat\gamma'}$.
  Finally, \cref{lem:valid-updater-construction} yields $\Itwo{T}(\gamma_*) = \mkUpdSem((u_x, u_y), u) \in \ValidUpd[\toType{\Gamma} \leadsto \sigma \times \tau]{\semTuple{T}}{\hat\gamma}$.

  \paragraph{Projections.}
  Let $T = t.1$ with $\tyjudg{\Gamma}{t}{\sigma \times \tau}$.
  By induction hypothesis, $(\sem{t}(\gamma), \Ione{t}(\gamma_*)) \in \logrel[\sigma \times \tau]$.
  By the definition of $\logrel[\sigma \times \tau]$, we have $(\pi_1(\sem{t}(\gamma)), \pi_1(\Ione{t}(\gamma_*))) \in \logrel[\sigma]$.
  Since $\sem{T}(\gamma) = \pi_1(\sem{t}(\gamma))$ and $\Ione{T}(\gamma_*) = \pi_1(\Ione{t}(\gamma_*))$, we obtain $(\sem{T}(\gamma), \Ione{T}(\gamma_*)) \in \logrel[\sigma]$.

  Updater correctness.
  Let $u_t = \Itwo{t}(\gamma_*)$.
  Define $A_{\hat\gamma} := \ValidUpd[\toType{\Gamma} \leadsto \sigma \times \tau]{\semTuple{t}}{\hat\gamma}$ and
  \[ u := \lam{(d\gamma, u_t)} (\mvar{dx}, u_t') := \applyUpd{u_t}{d\gamma};\; (\mvar{dx}.1, u_t') \]
  Stepping $u_t$ with $\mvar{d\hat\gamma}$ yields valid $(\mvar{dx}, u_t')$, i.e.\ $(\sem{t}(\gamma), \mvar{dx}, \sem{t}(\gamma')) \in \changerel[\sigma \times \tau]$ and $u_t' \in A_{\hat\gamma'}$.
  By the definition of $\changerel[\sigma \times \tau]$, we have $(\pi_1(\sem{t}(\gamma)), \pi_1(\mvar{dx}), \pi_1(\sem{t}(\gamma'))) \in \changerel[\sigma]$.
  By \cref{lem:valid-updater-construction}, it follows that $\Itwo{T}(\gamma_*) \in \ValidUpd[\toType{\Gamma} \leadsto \sigma]{\semTuple{T}}{\hat\gamma}$.
  The case $T = t.2$ is symmetric.

  \paragraph{Abstractions.}
  Let $T = \lam{x} t$ with $\tyjudg{\Gamma, x : \sigma}{t}{\tau}$, so $T : \sigma \clos[\kappa] \tau$ where $\kappa = \toType{\Gamma|_{\freeVars{T}}}$.
  Then we have $\IncToCore{T} = (\lam{x} \IncToCore[\Gamma|_{\freeVars{T}}, x: \sigma]{t},\, \mkUpd{\unit}{\lam{(d\gamma, \_)} (\projectVars{\Changed{\Gamma}}{\freeVars{T}} \; d\gamma, \unit)})$.
  Let $\mvar{env} = \ctxToTpl[\Gamma|_{\freeVars{T}}](\gamma|_{\freeVars{T}})$ and $f = \semTuple[\Gamma|_{\freeVars{T}}, x : \sigma]{t}$, so that $\sem{T}(\gamma) = (\mvar{env}, f)$.
  Let $f' = \Ione{T}(\gamma_*) = \lam{x'} \IncToCore[\Gamma|_{\freeVars{T}}, x : \sigma]{t}(\gamma_*|_{\freeVars{T}}, x')$.
  To show $(\sem{T}(\gamma), \Ione{T}(\gamma_*)) \in \logrel[{\sigma \clos[\kappa] \tau}]$, consider any $(v, v') \in \logrel[\sigma]$.
  Let $\gamma_1 = \gamma|_{\freeVars{T}}[x \mapsto v]$ and $\gamma_{1*} = \gamma_*|_{\freeVars{T}}[x \mapsto v']$.
  Since $(\gamma, \gamma_*) \in \logrel[\Gamma]$, we have $(\gamma_1, \gamma_{1*}) \in \logrel[\Gamma|_{\freeVars{T}}, x: \sigma]$.
  By induction hypothesis on $t$ in context $\Gamma|_{\freeVars{T}}, x: \sigma$:
  \[ (\sem{t}(\gamma_1), \Ione{t}(\gamma_{1*})) \in \logrel[\tau] \quad \text{ and } \quad \Itwo{t}(\gamma_{1*}) \in \ValidUpd[\kappa \times \sigma \leadsto \tau]{f}{(\mvar{env}, v)} \]
  Since $f(\mvar{env}, v) = \sem{t}(\gamma_1)$ and $f'(v') = (\Ione{t}(\gamma_{1*}), \Itwo{t}(\gamma_{1*}))$, this is exactly what the closure logical relation requires.

  Updater correctness.
  Let $u_{\lambda} = \lam{(d\gamma, \_)} (\projectVars{\Changed{\Gamma}}{\freeVars{T}} \; d\gamma, \unit)$.
  Define $A_{\hat\gamma} := \{ \unit \}$ (a singleton).
  Then $\mvar{dy} = \projectVars{\Changed{\Gamma}}{\freeVars{T}} \; \mvar{d\hat\gamma}$, which is the projection of the context change to the free variables of $T$.
  By \cref{lem:context-types}, $(\mvar{env}, \mvar{dy}, \mvar{env}') \in \changerel[\kappa]$ where $\mvar{env}' = \ctxToTpl[\Gamma|_{\freeVars{T}}](\gamma'|_{\freeVars{T}})$.
  Since $\Changed{\sigma \clos[\kappa] \tau} = \Changed{\kappa}$, and the underlying function $f$ is the same before and after the change (only the captured environment changes), we have $(\sem{T}(\gamma), \mvar{dy}, \sem{T}(\gamma')) \in \changerel[{\sigma \clos[\kappa] \tau}]$.
  The cache is $\unit \in A_{\hat\gamma'}$ trivially.
  By \cref{lem:valid-updater-construction}, $\Itwo{T}(\gamma_*) = \mkUpdSem(\unit, u_{\lambda}) \in \ValidUpd[{\toType{\Gamma} \leadsto \sigma \clos[\kappa] \tau}]{\semTuple{T}}{\hat\gamma}$.

  \paragraph{Application.}
  Let $T = s \; t$ with $\tyjudg{\Gamma}{s}{\sigma \clos[\kappa] \tau}$ and $\tyjudg{\Gamma}{t}{\sigma}$.
  By induction hypothesis on $s$, we have $(\sem{s}(\gamma), \Ione{s}(\gamma_*)) \in \logrel[{\sigma \clos[\kappa] \tau}]$ and $\Itwo{s}(\gamma_*) \in \ValidUpd[{\toType{\Gamma} \leadsto \sigma \clos[\kappa] \tau}]{\semTuple{s}}{\hat\gamma}$.
  By induction hypothesis on $t$, we have $(\sem{t}(\gamma), \Ione{t}(\gamma_*)) \in \logrel[\sigma]$ and $\Itwo{t}(\gamma_*) \in \ValidUpd[\toType{\Gamma} \leadsto \sigma]{\semTuple{t}}{\hat\gamma}$.
  Let $(\mvar{env}, f) = \sem{s}(\gamma)$, $f' = \Ione{s}(\gamma_*)$, $x = \sem{t}(\gamma)$, and $x' = \Ione{t}(\gamma_*)$; by the IH on $t$ above, $(x, x') \in \logrel[\sigma]$.
  Set $y = f(\mvar{env}, x)$ and $(y', u_y) = f'(x')$.
  By the closure logical relation, we have $(y, y') \in \logrel[\tau]$ and $u_y \in \ValidUpd[\kappa \times \sigma \leadsto \tau]{f}{(\mvar{env}, x)}$.
  Since $\sem{T}(\gamma) = y$ and $\Ione{T}(\gamma_*) = y'$, the relational part follows.

  Updater correctness.
  Let $u_f = \Itwo{s}(\gamma_*)$ and $u_x = \Itwo{t}(\gamma_*)$.
  Define admissible caches
  \begin{align*}
    A_{\hat\gamma} := \{ (u_f, u_x, u_y) &\mid u_f \in \ValidUpd[{\toType{\Gamma} \leadsto \sigma \clos[\kappa] \tau}]{\semTuple{s}}{\hat\gamma}, u_x \in \ValidUpd[\toType{\Gamma} \leadsto \sigma]{\semTuple{t}}{\hat\gamma}, u_y \in \ValidUpd[\kappa \times \sigma \leadsto \tau]{f}{(\mvar{env}, x)} \\
    &\quad \text{ where } (\mvar{env}, f) = \semTuple{s}(\hat\gamma), x = \semTuple{t}(\hat\gamma) \}
  \end{align*}
  The cache update function is
  \begin{align*} u &:= \lam{(d\gamma, (u_f, u_x, u_y))} (\mvar{df}, u_f') := \applyUpd{u_f}{d\gamma};\, (\mvar{dx}, u_x') := \applyUpd{u_x}{d\gamma};\, \\
    &\qquad\qquad\qquad\qquad\quad\, (\mvar{dy}, u_y') := \applyUpd{u_y}{(\mvar{df}, \mvar{dx})};\, (\mvar{dy}, (u_f', u_x', u_y'))
  \end{align*}
  Stepping $u_f$ and $u_x$ with $\mvar{d\hat\gamma}$ yields valid $(\mvar{df}, u_f')$ and $(\mvar{dx}, u_x')$, i.e.\ $(\sem{s}(\gamma), \mvar{df}, \sem{s}(\gamma')) \in \changerel[{\sigma \clos[\kappa] \tau}]$, $u_f' \in \ValidUpd[{\toType{\Gamma} \leadsto \sigma \clos[\kappa] \tau}]{\semTuple{s}}{\hat\gamma'}$, $(x, \mvar{dx}, x') \in \changerel[\sigma]$, and $u_x' \in \ValidUpd[\toType{\Gamma} \leadsto \sigma]{\semTuple{t}}{\hat\gamma'}$, where $x' = \sem{t}(\gamma')$.
  Let $(\mvar{env}', f') = \sem{s}(\gamma')$.
  By the definition of $\changerel[{\sigma \clos[\kappa] \tau}]$ (recalling $\Changed{\sigma \clos[\kappa] \tau} = \Changed{\kappa}$), the validity of $\mvar{df}$ above unpacks to $(\mvar{env}, \mvar{df}, \mvar{env}') \in \changerel[\kappa]$ together with $f' = f$.
  Hence $((\mvar{env}, x), (\mvar{df}, \mvar{dx}), (\mvar{env}', x')) \in \changerel[\kappa \times \sigma]$.
  Stepping $u_y$ with $(\mvar{df}, \mvar{dx})$ yields valid $(\mvar{dy}, u_y')$, i.e.\ $(y, \mvar{dy}, y') \in \changerel[\tau]$ and $u_y' \in \ValidUpd[\kappa \times \sigma \leadsto \tau]{f}{(\mvar{env}', x')}$.
  Since $f' = f$, the function component named in the definition of $A_{\hat\gamma'}$ (extracted from $\semTuple{s}(\hat\gamma')$) is again $f$, so $(u_f', u_x', u_y') \in A_{\hat\gamma'}$.
  By \cref{lem:valid-updater-construction}, it follows that $\Itwo{T}(\gamma_*) = \mkUpdSem((u_f, u_x, u_y), u) \in \ValidUpd[\toType{\Gamma} \leadsto \tau]{\semTuple{T}}{\hat\gamma}$.

  \paragraph{Record literals.}
  Record literals $T = \recordLit{\ell_1: t_1, \dots, \ell_n: t_n}$ are the $n$-ary generalization of pairs.
  (Note, however, that their change types differ from those of pairs: changes to records are subrecords, which may omit unchanged fields.)
  The relational part follows from the induction hypothesis on each $t_i$ and the definition of $\logrel[\Record{\ell_1: \tau_1, \dots, \ell_n: \tau_n}]$.
  For updater correctness, the admissible caches are $A_{\hat\gamma} := \prod_{i=1}^n \ValidUpd[\toType{\Gamma} \leadsto \tau_i]{\semTuple{t_i}}{\hat\gamma}$, and the cache update function threads $\mvar{d\hat\gamma}$ through each sub-updater independently.
  Cache preservation follows from the validity of each sub-updater and the definition of $\changerel[\Record{\ell_1: \tau_1, \dots, \ell_n: \tau_n}]$, using \cref{lem:valid-updater-construction}.

  \paragraph{Record projection.}
  Let $T = t.\ell_j$ with $\tyjudg{\Gamma}{t}{\Record{\ell_1: \tau_1, \dots, \ell_n: \tau_n}}$ and result type $\tau_j$.
  By induction hypothesis, $(\sem{t}(\gamma), \Ione{t}(\gamma_*)) \in \logrel[\Record{\ell_1: \tau_1, \dots, \ell_n: \tau_n}]$ and $\Itwo{t}(\gamma_*) \in \ValidUpd[\toType{\Gamma} \leadsto \Record{\dots}]{\semTuple{t}}{\hat\gamma}$.
  By the definition of $\logrel[\Record{\dots}]$, the $\ell_j$-component is related: $(\sem{T}(\gamma), \Ione{T}(\gamma_*)) \in \logrel[\tau_j]$.

  Updater correctness.
  We distinguish two sub-cases depending on whether $\same[\tau_j]$ exists.

  \emph{Case 1: $\same[\tau_j]$ exists.}
  The cache is $u_r = \Itwo{t}(\gamma_*)$, so define $A_{\hat\gamma} := \ValidUpd[\toType{\Gamma} \leadsto \Record{\dots}]{\semTuple{t}}{\hat\gamma}$.
  The cache update function steps $u_r$ with $\mvar{d\hat\gamma}$ to obtain $(\mvar{dr}, u_r')$.
  If $\ell_j$ is present in $\mvar{dr}$, the output is $\mvar{dr}(\ell_j)$; otherwise it is $\sem{\same[\tau_j]}$.
  For cache preservation, let $(\mvar{d\hat\gamma}, u_r) \in \changerel[\toType{\Gamma}] \times A_{\hat\gamma}$.
  By validity of $u_r$, stepping it with $\mvar{d\hat\gamma}$ yields $(\mvar{dr}, u_r')$ with $(\semTuple{t}(\hat\gamma), \mvar{dr}, \semTuple{t}(\hat\gamma')) \in \changerel[\Record{\dots}]$ and $u_r' \in A_{\hat\gamma'}$.
  If $\ell_j \in \mvar{dr}$: by $\changerel[\Record{\dots}]$, $(\semTuple{t}(\hat\gamma)(\ell_j), \mvar{dr}(\ell_j), \semTuple{t}(\hat\gamma')(\ell_j)) \in \changerel[\tau_j]$, so the output change is valid.
  If $\ell_j \notin \mvar{dr}$: by $\changerel[\Record{\dots}]$, $\semTuple{t}(\hat\gamma)(\ell_j) = \semTuple{t}(\hat\gamma')(\ell_j)$, so $\sem{\same[\tau_j]}$ is a valid change.
  By \cref{lem:valid-updater-construction}, $\Itwo{T}(\gamma_*) \in \ValidUpd[\toType{\Gamma} \leadsto \tau_j]{\semTuple{T}}{\hat\gamma}$.

  \emph{Case 2: $\same[\tau_j]$ does not exist (but $\sameAs[\tau_j]$ and $\apply[\tau_j]$ exist).}
  The cache is $(r_\ell, u_r)$ where $r_\ell = \semTuple{t}(\hat\gamma)(\ell_j)$ and $u_r = \Itwo{t}(\gamma_*)$.
  Define
  \[ A_{\hat\gamma} := \{ (r_\ell, u_r) \mid u_r \in \ValidUpd[\toType{\Gamma} \leadsto \Record{\dots}]{\semTuple{t}}{\hat\gamma} \;\land\; r_\ell = \semTuple{t}(\hat\gamma)(\ell_j) \} \]
  The cache update function steps $u_r$ with $\mvar{d\hat\gamma}$ to obtain $(\mvar{dr}, u_r')$.
  If $\ell_j \in \mvar{dr}$: the output is $\mvar{dr}(\ell_j)$ and the new cache is $(\sem{\apply[\tau_j]}(r_\ell, \mvar{dr}(\ell_j)),\; u_r')$.
  If $\ell_j \notin \mvar{dr}$: the output is $\sem{\sameAs[\tau_j]}(r_\ell)$ and the new cache is $(r_\ell, u_r')$.
  For cache preservation, let $(\mvar{d\hat\gamma}, (r_\ell, u_r)) \in \changerel[\toType{\Gamma}] \times A_{\hat\gamma}$, so $r_\ell = \semTuple{t}(\hat\gamma)(\ell_j)$.
  By validity of $u_r$, stepping it yields $(\mvar{dr}, u_r')$ with $(\semTuple{t}(\hat\gamma), \mvar{dr}, \semTuple{t}(\hat\gamma')) \in \changerel[\Record{\dots}]$ and $u_r' \in \ValidUpd[\toType{\Gamma} \leadsto \Record{\dots}]{\semTuple{t}}{\hat\gamma'}$.
  If $\ell_j \in \mvar{dr}$: by $\changerel[\Record{\dots}]$, $(r_\ell, \mvar{dr}(\ell_j), \semTuple{t}(\hat\gamma')(\ell_j)) \in \changerel[\tau_j]$, so the output change is valid.
  The new cached value $\sem{\apply[\tau_j]}(r_\ell, \mvar{dr}(\ell_j)) = \semTuple{t}(\hat\gamma')(\ell_j)$, so the new cache is in $A_{\hat\gamma'}$.
  If $\ell_j \notin \mvar{dr}$: by $\changerel[\Record{\dots}]$, $r_\ell = \semTuple{t}(\hat\gamma)(\ell_j) = \semTuple{t}(\hat\gamma')(\ell_j)$, so $\sem{\sameAs[\tau_j]}(r_\ell)$ is a valid change and $(r_\ell, u_r') \in A_{\hat\gamma'}$.
  By \cref{lem:valid-updater-construction}, $\Itwo{T}(\gamma_*) \in \ValidUpd[\toType{\Gamma} \leadsto \tau_j]{\semTuple{T}}{\hat\gamma}$.

  \paragraph{Record restriction.}
  Let $T = \restrict{t}{\{\ell_{i_1}, \dots, \ell_{i_k}\}}$ with $\tyjudg{\Gamma}{t}{\Record{\ell_1: \tau_1, \dots, \ell_n: \tau_n}}$ and result type $\Record{\ell_{i_1}: \tau_{i_1}, \dots, \ell_{i_k}: \tau_{i_k}}$.
  The relational part follows from the induction hypothesis on $t$ and the definition of $\logrel[\Record{\dots}]$.
  For updater correctness, the cache is $u_r = \Itwo{t}(\gamma_*)$, so define $A_{\hat\gamma} := \ValidUpd[\toType{\Gamma} \leadsto \Record{\dots}]{\semTuple{t}}{\hat\gamma}$.
  The cache update function steps $u_r$ with $\mvar{d\hat\gamma}$ to obtain $(\mvar{dr}, u_r')$ and outputs $\restrict{\mvar{dr}}{\{\ell_{i_1}, \dots, \ell_{i_k}\}}$.
  For cache preservation: by validity of $u_r$, we get $(\semTuple{t}(\hat\gamma), \mvar{dr}, \semTuple{t}(\hat\gamma')) \in \changerel[\Record{\ell_1: \tau_1, \dots, \ell_n: \tau_n}]$ and $u_r' \in A_{\hat\gamma'}$.
  By $\changerel[\Record{\dots}]$, for each $\ell_{i_p}$: if $\ell_{i_p} \in \mvar{dr}$ then $(\semTuple{t}(\hat\gamma)(\ell_{i_p}), \mvar{dr}(\ell_{i_p}), \semTuple{t}(\hat\gamma')(\ell_{i_p})) \in \changerel[\tau_{i_p}]$, and if $\ell_{i_p} \notin \mvar{dr}$ then the field is unchanged.
  This is exactly what $\changerel[\Record{\ell_{i_1}: \tau_{i_1}, \dots, \ell_{i_k}: \tau_{i_k}}]$ requires for the restricted change, so the output change is valid.
  By \cref{lem:valid-updater-construction}, $\Itwo{T}(\gamma_*) \in \ValidUpd[\toType{\Gamma} \leadsto \Record{\ell_{i_1}: \tau_{i_1}, \dots, \ell_{i_k}: \tau_{i_k}}]{\semTuple{T}}{\hat\gamma}$.

  \paragraph{For loops (list version).}
  Let $T = \forWithUsing{x}{\mvar{xs}}{z}{s}{v_1 := y_1[u_1], \dots, v_n := y_n[u_n]}{t}$ with $\tyjudg{\Gamma}{\mvar{xs}}{\List{\tau}}$, $\tyjudg{\Gamma}{s}{\rho}$, $\tyjudg{\Gamma, x: \tau, z: \rho}{y_j[u_j]}{\theta_j}$ for $j = 1, \dots, n$, and $\tyjudg{\Gamma_{\mvar{ext}}}{t}{\tau' \times \rho}$ where $\Gamma_{\mvar{ext}} = \Gamma, x: \tau, z: \rho, v_1: \theta_1, \dots, v_n: \theta_n$; the result type is $\List{\tau'}$.
  Let $(a_1, \dots, a_N) = \sem{\mvar{xs}}(\gamma)$ be the input list of length~$N$, and let $z_0 = \sem{s}(\gamma)$.
  For each $i = 1, \dots, N$, let $\gamma_i^0 = \gamma[x \mapsto a_i, z \mapsto z_{i-1}]$, the lookup values $w_j = \sem{y_j}(\gamma_i^0)[\sem{u_j}(\gamma_i^0)]$ for $j = 1, \dots, n$ (typed in $\Gamma, x: \tau, z: \rho$ — they do not see prior $v$-bindings), and the extended context $\gamma_i = \gamma_i^0[v_1 \mapsto w_1, \dots, v_n \mapsto w_n]$.
  Let $(r_i, z_i) = \sem{t}(\gamma_i)$ be the per-iteration result and updated accumulator.
  Define the analogous incrementalized contexts: let $(a_1', \dots, a_N') = \Ione{\mvar{xs}}(\gamma_*)$, $z_0' = \Ione{s}(\gamma_*)$, $\gamma_{i*}^0 = \gamma_*[x \mapsto a_i', z \mapsto z_{i-1}']$, $w_j' = \Ione{y_j}(\gamma_{i*}^0)[\Ione{u_j}(\gamma_{i*}^0)]$, $\gamma_{i*} = \gamma_{i*}^0[v_1 \mapsto w_1', \dots, v_n \mapsto w_n']$, and $z_i' = \pi_2(\Ione{t}(\gamma_{i*}))$.

  For the relational part, we have $(\sem{\mvar{xs}}(\gamma), \Ione{\mvar{xs}}(\gamma_*)) \in \logrel[\List{\tau}]$ and $(\sem{s}(\gamma), \Ione{s}(\gamma_*)) \in \logrel[\rho]$ by induction hypothesis applied to $\mvar{xs}$ and $s$.
  For each iteration $i$ (by induction on $i$, since $\gamma_{i*}$ depends recursively on $z_{i-1}'$), the induction hypotheses on $u_j$, $y_j$, and $t$ yield related values for the index computations, collection lookups, and body evaluation, from which $(\sem{t}(\gamma_i), \Ione{t}(\gamma_{i*})) \in \logrel[\tau' \times \rho]$.
  Projecting to the first component gives the per-iteration element relation $(r_i, r_i') \in \logrel[\tau']$, where $(r_i, z_i) = \sem{t}(\gamma_i)$ and $(r_i', z_i') = \Ione{t}(\gamma_{i*})$.
  Since the logical relation for lists requires element-wise relatedness, we obtain $(\sem{T}(\gamma), \Ione{T}(\gamma_*)) \in \logrel[\List{\tau'}]$.

  Updater correctness.
  The incrementalizing transformation for for-expressions (\cref{fig:inc-to-core-terms-for-full}) constructs a cache containing the following components:
  the input list $\var{xs}$ and its updater $\var{xsUpd}$, the sequence of accumulated values $\var{seeds} = [z_0, z_1, \dots, z_N]$ (stored 1-indexed so that $\var{seeds}[i] = z_{i-1}$), the seed updater $\var{seedUpd}$, the per-iteration body updaters $\var{updaters} = [u_1^{\mvar{body}}, \dots, u_N^{\mvar{body}}]$, the collection values and their updaters $\var{collections}_j$ and $\var{collectionUpds}_j$ for $j = 1, \dots, n$, the per-iteration index updaters $\var{idxUpdaters}_j[i]$ for each lookup $j$ and iteration $i$, and the index tracking maps $\var{itersToIndices}_j$ and $\var{indicesToIters}_j$.
  Let $\hat\gamma_i^{\mvar{ext}} = \ctxToTpl[\Gamma_{\mvar{ext}}](\gamma_i)$ be the tuple form of the per-iteration extended context.

  Define admissible caches $A_{\hat\gamma}$ as the set of caches satisfying:
  \begin{enumerate}
    \item $\var{seeds}[i] = z_{i-1}$ for $i = 1, \dots, N+1$;
    \item $\var{seedUpd} \in \ValidUpd[\toType{\Gamma} \leadsto \rho]{\semTuple{s}}{\hat\gamma}$;
    \item $u_i^{\mvar{body}} \in \ValidUpd[\toType{\Gamma_{\mvar{ext}}} \leadsto \tau' \times \rho]{\semTuple[\Gamma_{\mvar{ext}}]{t}}{\hat\gamma_i^{\mvar{ext}}}$ for $i = 1, \dots, N$ (note: the underlying function $\semTuple[\Gamma_{\mvar{ext}}]{t}$ is the same across all iterations; only the input point $\hat\gamma_i^{\mvar{ext}}$ varies);
    \item $\var{idxUpdaters}_j[i] \in \ValidUpd[\toType{(\Gamma, x: \tau, z: \rho)} \leadsto \Index_j]{\semTuple[\Gamma, x: \tau, z: \rho]{u_j}}{\ctxToTpl[\Gamma, x: \tau, z: \rho](\gamma_i^0)}$ for each $j$ and $i$, where $\Index_j$ is $\Nat$ or $\Name$ depending on whether $y_j$ is a list or a name map;
    \item $\var{collectionUpds}_j \in \ValidUpd[\toType{\Gamma} \leadsto \Collection{\theta_j}]{\semTuple{y_j}}{\hat\gamma}$ for each $j$;
    \item the index maps $\var{itersToIndices}_j$ and $\var{indicesToIters}_j$ are consistent with the current index values, i.e., $\var{itersToIndices}_j[i] = \sem{u_j}(\gamma_i^0)$ and $\var{indicesToIters}_j[\var{idx}] = \{ i \mid \var{itersToIndices}_j[i] = \var{idx} \}$;
    \item $\var{xsUpd} \in \ValidUpd[\toType{\Gamma} \leadsto \List{\tau}]{\semTuple{\mvar{xs}}}{\hat\gamma}$ and $\var{collections}_j = \semTuple{y_j}(\hat\gamma)$ for each $j$.
  \end{enumerate}

  The cache update function (\cref{fig:inc-to-core-terms-for-full}, right column) processes a context change $\mvar{d\hat\gamma}$ as follows.
  First, it steps the input-list updater $\var{xsUpd}$ with $\mvar{d\hat\gamma}$ to obtain $(\mvar{dxs}, \var{xsUpd}')$, and the seed updater $\var{seedUpd}$ with $\mvar{d\hat\gamma}$ to obtain $(\mvar{dz}_0, \var{seedUpd}')$, where $\mvar{dz}_0$ is the change to the initial accumulator $z_0$.
  Next, it steps each collection updater $\var{collectionUpds}_j$ with $\mvar{d\hat\gamma}$ to obtain $(\mvar{dy}_j, \var{collectionUpds}_j')$ and computes the post-change collection $\var{collections}'_j := \apply(\var{collections}_j, \mvar{dy}_j)$, keeping the pre-change $\var{collections}_j$ available for change construction inside the per-iteration loop.
  It then determines the set $\var{toVisit}$ of iterations that need re-evaluation: this includes iterations whose input element changed (from $\mvar{dxs}.\lbl{change}$) and iterations referencing changed collection entries (identified via $\var{indicesToIters}_j[\var{idx}]$ for each $\var{idx}$ with $(\var{idx}, \_) \in \mvar{dy}_j.\lbl{change}$). In addition, if $\mvar{d\hat\gamma}$ changes any $\Gamma$-variable free in $t$ or in any $u_j$, then $\var{toVisit}$ is set to all iterations $\{1, \dots, N\}$ (these are the variables whose changes are not already captured by per-iteration mechanisms).

  Then the iterations are processed, threading a running accumulator change $\mvar{dz}$ (initialized to $\mvar{dz}_0$). The cursor $i$ starts at $1$ and, after each processed iteration, $i$ is removed from $\var{toVisit}$ and incremented. At each step:
  \begin{itemize}
    \item If $\issame{\mvar{dz}}$ and $i \notin \var{toVisit}$, the algorithm jumps ahead: if $\var{toVisit}$ is empty it breaks out of the loop; otherwise it pops the smallest element of $\var{toVisit}$ as the new $i$ and sets $\mvar{dz} := \sameAs(\var{seeds}[i])$. Since the invariant ``$\var{toVisit}$ contains only iterations not yet processed'' is maintained by the post-iteration removal, the popped value is always strictly greater than the previous $i$, and the algorithm falls through to process iteration $i$.
    \item To process iteration $i$, the per-iteration index updaters $\var{idxUpdaters}_j[i]$ are stepped with the change $(\mvar{d\hat\gamma}, \mvar{dx}_i, \mvar{dz})$ at the lookup point $(\hat\gamma, a_i, z_{i-1})$, yielding a valid index change $\mvar{didx}_j$. If $\issame{\mvar{didx}_j}$, the index expression evaluates to the same index $\var{idx} = \mvar{didx}_j.\lbl{new}$ as before; we then set $\mvar{dv}_j := \mvar{dval}$ if $(\var{idx}, \mvar{dval}) \in \mvar{dy}_j.\lbl{change}$, and otherwise set $\mvar{dv}_j := \sameAs(\var{collections}_j[\var{idx}])$. If $\lnot\issame{\mvar{didx}_j}$, the index maps are updated and $\mvar{dv}_j$ is constructed as the change from the pre-change value $\var{collections}_j[\var{old}]$ to the post-change value $\var{collections}'_j[\var{new}]$. Then the per-iteration body updater $u_i^{\mvar{body}}$ is stepped with the input change $(\mvar{d\hat\gamma}, \mvar{dx}_i, \mvar{dz}, \mvar{dv}_1, \dots, \mvar{dv}_n)$ to obtain the output change $(\mvar{dr}_i, \mvar{dz}_i)$ and a new updater $u_i^{\mvar{body}\prime}$. (Here $\mvar{dx}_i$ is the element change at position $i$, taken from $\mvar{dxs}.\lbl{change}$ if present and $\sameAs(a_i)$ otherwise.) The accumulator change $\mvar{dz}$ is updated to $\mvar{dz}_i$ and propagated to the next iteration.
  \end{itemize}
  At the end, the cache's $\var{collections}_j$ is replaced by the post-change $\var{collections}'_j$.

  Cache preservation. We need to show that the updated cache lies in $A_{\hat\gamma'}$.
  By validity of $\var{seedUpd}$, stepping it with $\mvar{d\hat\gamma}$ yields $(\mvar{dz}_0, \var{seedUpd}')$ with $(z_0, \mvar{dz}_0, z_0') \in \changerel[\rho]$ where $z_0' = \sem{s}(\gamma')$, and also $\var{seedUpd}' \in \ValidUpd[\toType{\Gamma} \leadsto \rho]{\semTuple{s}}{\hat\gamma'}$, so clause~2 of the admissible-cache definition is preserved.
  By validity of each $\var{collectionUpds}_j$, stepping it yields $(\mvar{dy}_j, \var{collectionUpds}_j')$ with $(\sem{y_j}(\gamma), \mvar{dy}_j, \sem{y_j}(\gamma')) \in \changerel[\Collection{\theta_j}]$ and $\var{collectionUpds}_j' \in \ValidUpd[\toType{\Gamma} \leadsto \Collection{\theta_j}]{\semTuple{y_j}}{\hat\gamma'}$, so clause~5 is preserved.
  By the validity of $\var{xsUpd}$, stepping it yields $(\mvar{dxs}, \var{xsUpd}')$ with $(\sem{\mvar{xs}}(\gamma), \mvar{dxs}, \sem{\mvar{xs}}(\gamma')) \in \changerel[\List{\tau}]$ as well as $\var{xsUpd}' \in \ValidUpd[\toType{\Gamma} \leadsto \List{\tau}]{\semTuple{\mvar{xs}}}{\hat\gamma'}$, so clause~7 is preserved.
  Furthermore, $\var{collections}'_j = \semTuple{y_j}(\hat\gamma')$ holds by validity of $\var{collectionUpds}_j$, justifying the cache's $\var{collections}_j$ being replaced by $\var{collections}'_j$ at the end of the update.

  The remaining clauses (1, 3, 4, 6) are preserved by induction on $i = 1, \dots, N$, maintaining as invariant that $(z_{i-1}, \mvar{dz}, z_{i-1}') \in \changerel[\rho]$ when entering iteration $i$. Let $\hat\gamma_i^{\mvar{ext}\prime} = \ctxToTpl[\Gamma_{\mvar{ext}}](\gamma_i')$ where $\gamma_i'$ is the post-change per-iteration context (defined analogously to $\gamma_i$ but using $\gamma'$ and the post-change lookups).

  For an iteration $i$ with $i \in \var{toVisit}$ or $\lnot\issame{\mvar{dz}}$: by validity of $\var{idxUpdaters}_j[i]$ (clause~4), stepping it with the input change $(\mvar{d\hat\gamma}, \mvar{dx}_i, \mvar{dz})$ — which lies in $\changerel[\toType{(\Gamma, x: \tau, z: \rho)}]$ at the lookup point $(\hat\gamma, a_i, z_{i-1})$ — yields a valid index change $\mvar{didx}_j$. If $\issame{\mvar{didx}_j}$, set $\var{idx} = \mvar{didx}_j.\lbl{new}$; the constructed $\mvar{dv}_j$ is either $\mvar{dval}$ where $(\var{idx}, \mvar{dval}) \in \mvar{dy}_j.\lbl{change}$ (which is a valid entry change by validity of $\var{collectionUpds}_j$) or $\sameAs(\var{collections}_j[\var{idx}])$ (which is valid because absence of an entry of $\mvar{dy}_j.\lbl{change}$ at $\var{idx}$ means the entry is unchanged). If $\lnot\issame{\mvar{didx}_j}$, the constructed $\mvar{dv}_j$ is the change from $\var{collections}_j[\var{old}] = \sem{y_j}(\gamma)[\sem{u_j}(\gamma_i^0)]$ to $\var{collections}'_j[\var{new}] = \sem{y_j}(\gamma')[\sem{u_j}(\gamma_i^{0\prime})]$, which is exactly the change to $v_j$ between $\gamma_i$ and $\gamma_i'$. Index-map consistency (clause~6) at $i$ holds in both sub-cases: if $\issame{\mvar{didx}_j}$ the index is unchanged, and otherwise the explicit updates to $\var{itersToIndices}_j$ and $\var{indicesToIters}_j$ restore consistency.

  Then $(\mvar{d\hat\gamma}, \mvar{dx}_i, \mvar{dz}, \mvar{dv}_1, \dots, \mvar{dv}_n)$ is a valid change in $\changerel[\Gamma_{\mvar{ext}}]$ at $\hat\gamma_i^{\mvar{ext}}$, so by validity of $u_i^{\mvar{body}}$, stepping it produces $(\mvar{dr}_i, \mvar{dz}_i)$ with $(r_i, \mvar{dr}_i, r_i') \in \changerel[\tau']$ and $(z_i, \mvar{dz}_i, z_i') \in \changerel[\rho]$, and the new updater $u_i^{\mvar{body}\prime}$ is valid at $\hat\gamma_i^{\mvar{ext}\prime}$.

  For an iteration $i$ with $\issame{\mvar{dz}}$ and $i \notin \var{toVisit}$: the algorithm skips iteration~$i$, so we must show the cache entries at index~$i$ remain valid at $\hat\gamma'$ without being updated. The $\var{toVisit}$ conditions give $a_i = a_i'$, no $\Gamma$-variable free in $t$ or any $u_j$ has changed (else $\var{toVisit} = \{1,\dots,N\}$), and $(\sem{u_j}(\gamma_i^0), \_) \notin \mvar{dy}_j.\lbl{change}$ for each $j$ (else $i$ would be in $\var{toVisit}$ via $\var{indicesToIters}_j$). Together with $z_{i-1} = z_{i-1}'$ from $\issame{\mvar{dz}}$, these yield $\sem{u_j}(\gamma_i^0) = \sem{u_j}(\gamma_i^{0\prime})$ and $w_j = w_j'$ for each~$j$, hence $\gamma_i^{\mvar{ext}}$ and $\gamma_i^{\mvar{ext}\prime}$ agree on the free variables of $t$ (and analogously $\gamma_i^0$ and $\gamma_i^{0\prime}$ agree on those of each $u_j$), so $(r_i, z_i) = (r_i', z_i')$ by functionality of $\sem{t}$. The unchanged $\var{seeds}[i+1]$, $u_i^{\mvar{body}}$, $\var{idxUpdaters}_j[i]$, and index-map entries at~$i$ therefore still satisfy clauses~1, 3, 4, and~6. When the algorithm later pops some $k$ and sets $\mvar{dz} := \sameAs(\var{seeds}[k]) = \sameAs(z_{k-1})$, chaining $z_m = z_m'$ across the skipped indices between the previous processed/popped index and $k-1$ gives $z_{k-1} = z_{k-1}'$, so the assignment is a valid change in $\changerel[\rho]$ and re-establishes the running-$\mvar{dz}$ invariant for iteration~$k$.

  The output $\var{dresults} = \recordLit{\field{change} \var{changed}}$ records exactly the iterations whose body output element actually changed, which is a valid change in $\changerel[\List{\tau'}]$ between $(r_1, \dots, r_N)$ and $(r_1', \dots, r_N')$.
  By \cref{lem:valid-updater-construction}, the constructed updater is in $\ValidUpd[\toType{\Gamma} \leadsto \List{\tau'}]{\semTuple{T}}{\hat\gamma}$.

  \paragraph{Note.} The pseudocode in \cref{fig:inc-to-core-terms-for-full} only handles modifications to list elements, not insertions or deletions, in order to keep the presentation readable. We also assume each collection expression $y_j$ does not depend on $x$ or $z$, so $\IncToCore{y_j}$ produces the same value on every iteration of the outer loop (otherwise the cached $\var{collections}_j$ would need to be per-iteration, complicating the bookkeeping).
  The for-expression over name sets is analogous, with $\NameMap{\tau'}$ in place of $\List{\tau'}$, iteration over names (with the loop variable bound to a $\Name$) in place of iteration over a list (with the loop variable bound to an element), and without an accumulator (so $\var{seeds}$ and the body's accumulator output drop out of the cache and the body's value type is just $\tau'$ rather than $\tau' \times \rho$).
  Our implementation also handles insertions and deletions for name collections, but not for lists (since this wasn't needed for the benchmarks).

\end{proof}

\end{document}